\documentclass[onecolumn, prx, nofootinbib,superscriptaddress, a4paper,11pt]{revtex4-2}
\makeatletter

\renewcommand*{\p@subsection}{}

\renewcommand*{\p@subsubsection}{}
\makeatother

\makeatletter
\def\l@section#1#2{\@dottedtocline{1}{0em}{1.5em}{#1}{#2}}
\def\l@subsection#1#2{\@dottedtocline{2}{1.5em}{2.5em}{#1}{#2}}
\def\l@subsubsection#1#2{\@dottedtocline{3}{4.0em}{2.5em}{#1}{#2}}
\makeatletter

\usepackage[utf8]{inputenc}
\pdfoutput=1
\usepackage[english]{babel}
\usepackage[T1]{fontenc}

\usepackage{xspace}
\usepackage[normalem]{ulem}
\newcommand{\core}{essential\xspace}
\newcommand{\Core}{Essential\xspace}

\usepackage[usenames,dvipsnames]{xcolor}
\definecolor{darkblue}{rgb}{0.0, 0.0, 0.7}
\definecolor{bordeaux}{rgb}{0.34, 0.01, 0.1}
\definecolor{turquoise}{rgb}{0.0, 0.6, 0.6}

\usepackage[margin=2.5cm,a4paper]{geometry}

\usepackage{wrapfig}

\usepackage{mathtools}
\usepackage{amsmath, amssymb, amstext, amsfonts}
\usepackage{amsthm}
\allowdisplaybreaks
\usepackage{float}
\usepackage{graphicx}
\usepackage{physics}
\usepackage[shortlabels]{enumitem}
\usepackage[colorlinks,linkcolor=bordeaux,citecolor=darkblue,urlcolor=blue,hypertexnames=true]{hyperref}

\numberwithin{equation}{section}
\numberwithin{figure}{section}
\numberwithin{table}{section}

\allowdisplaybreaks

\def\R{\mathbb{R}}
\def\P{\mathbb{C}\langle x\rangle}

\newcommand{\Pnull}{\mathcal J^Z}

\DeclareMathOperator{\SOS}{SOS}
\DeclareMathOperator{\im}{im}
\DeclareMathOperator{\Gr}{Gr}
\def\beq{\begin{equation}}
\def\eeq{\end{equation}}

\newtheorem{theo}{Theorem}[section]
\newtheorem{prop}[theo]{Proposition}
\newtheorem{cor}[theo]{Corollary}
\newtheorem{defin}[theo]{Definition}
\newtheorem{lemma}[theo]{Lemma}
\theoremstyle{definition} 
\newtheorem{remark}[theo]{Remark}

\newtheorem{eg}[theo]{Example}

\def\id{{\mathbb I}}

\def\bra#1{\langle#1|} \def\ket#1{|#1\rangle}
\def\braket#1#2{\langle#1|#2\rangle}

\def\be{\begin{equation}}
\def\ee{\end{equation}}
\def\bea{\begin{eqnarray}}
\def\eea{\end{eqnarray}}
\def\bma{\begin{mathletters}}
\def\ema{\end{mathletters}}

\def\A{{\cal A}}
\def\p{\overline{\pi}}
\def\q0{\underline{0}}

\def\H{{\cal H}}
\def\Z{\C\langle z\rangle}

\def\C{{\mathbb C}}
\def\id{{\mathbb I}}
\def\E{{\cal P}_E}

\def\M{{\cal M}}
\def\H{{\cal H}}
\def\HH{\mathbb{H}}
\def\B{{\cal B}}
\def\F{{\cal P}_F}

\def\R{\mathbb{R}}
\def\CC{\mathbb{C}}
\def\N{\mathbb{N}}

\def\g{\mathbf{g}}
\def\h{\mathbf{h}}
\def\cN{\mathcal N}

\begin{document}

\title{First-order optimality conditions for non-commutative optimization problems}

\author{Mateus Ara\'ujo}
\email{mateus.araujo@uva.es}
\affiliation{Departamento de Física Teórica, Atómica y Óptica, Laboratory for Disruptive Interdisciplinary Science (LaDIS), Universidad de Valladolid, 47011 Valladolid, Spain}
\author{Igor Klep}
\email{igor.klep@fmf.uni-lj.si}
\affiliation{Faculty of Mathematics and Physics, University of Ljubljana 
\& Famnit, University of Primorska, Koper 
\& Institute of Mathematics, Physics and Mechanics,
Ljubljana, Slovenia}
\author{Andrew J. P. Garner}
\email{ajp.garner@gmail.com}
\affiliation{Institute for Quantum Optics and Quantum Information (IQOQI) Vienna\\ Austrian Academy of Sciences, Boltzmanngasse 3, Wien 1090, Austria}
\author{Tam\'as V\'ertesi}
\email{tvertesi@atomki.hu}
\affiliation{HUN-REN Institute for Nuclear Research, P.O. Box 51, H-4001 Debrecen, Hungary}
\author{Miguel Navascu\'es}
\email{Corresponding author. miguel.navascues@oeaw.ac.at}
\affiliation{Institute for Quantum Optics and Quantum Information (IQOQI) Vienna\\ Austrian Academy of Sciences, Boltzmanngasse 3, Wien 1090, Austria}

\begin{abstract}
We consider the problem of optimizing the state average of a polynomial of non-commuting variables, over all states and operators satisfying a number of polynomial constraints, and over all Hilbert spaces where such states and operators are defined. Such non-commutative polynomial optimization (NPO) problems are routinely solved through hierarchies of semidefinite programming (SDP) relaxations. 
By formulating the general NPO problem in Lagrangian terms, we heuristically derive first-order optimality conditions via small variations in the problem variables. Although the derivation is not rigorous, it gives rise to two types of optimality conditions---state and operator---which are rigorously analyzed in the paper. Both types of conditions can be enforced through additional positive semidefinite constraints in the SDP hierarchies. State optimality conditions are shown to be satisfied by all NPO problems. 
For NPO problems with optimal solutions (such as, e.g., Archimedean ones)
they allow enforcing a new type of constraints: namely, restricting the optimization over states to the set of common ground states of an arbitrary number of operators. Operator optimality conditions are the non-commutative analogs of the Karush--Kuhn--Tucker (KKT) conditions, which are known to hold in many classical optimization problems. In this regard, we prove that a weak form of operator optimality holds for all NPO problems; stronger versions require the problem constraints to satisfy some qualification criterion, just like in the classical case (e.g.: Mangasarian--Fromovitz constraint qualification). We test the power of the new optimality conditions by computing local properties of ground states of many-body spin systems and the maximum quantum violation of Bell inequalities.
\end{abstract}

\maketitle




\tableofcontents

\section{Introduction}
\label{intro}

Non-commutative polynomial optimization (NPO) studies the problem of minimizing the bottom of the spectrum of a polynomial of non-commuting variables, over all operator representations of these variables satisfying a number of polynomial equations and inequalities. As it turns out, in quantum mechanics many interesting physical quantities such as energy, spin and momentum are represented by operators satisfying polynomial constraints. Consequently, natural applications of NPO have been found in quantum information theory, quantum chemistry and condensed matter physics in the last decades. Examples of practical NPO problems include computing the maximal quantum violation of a Bell inequality~\cite{NPA2007,NPA2008}, the electronic energy of atoms and molecules~\cite{Nakata2001,Mazziotti2004,Mazziotti2023}, the ground state energies of spin systems~\cite{Barthel2012,Baumgratz2012,Haim2020,Requena2023}, or the ground state behavior of fermions at finite density~\cite{Lawrence2023}.

From the work of~\citet{PNA2010}, (see also Refs.~\cite{Helton2004,NPA2007,NPA2008,Doherty2008,Burgdorf2016}), we know that all NPO problems involving bounded operators can be solved through hierarchies of semidefinite programs (SDPs) \cite{Vandenberghe1996,nesterov} of increasing complexity. While the first levels of said hierarchies provide very good approximations for many NPO problems, sometimes there are considerable gaps between the lower bound provided by the SDP solver and the conjectured solution of the problem. That is, even though the SDP hierarchies converge for any problem, for some NPO problems they seem to converge too slowly. This leaves many important problems in quantum nonlocality and many-body physics unsolved, due to a lack of computational resources.

In this paper, we introduce an improved and stronger method to tackle NPO problems. The main idea is that NPOs, like classical optimization problems, very often obey a number of optimality relations, which in the classical case are dubbed the Karush--Kuhn--Tucker (KKT) conditions~\cite{Karush1939,Kuhn1951}. In the present work, we adapt and generalize these conditions to the non-commutative setting. We find that they come in two flavors: state and operator optimality conditions, both of which take the form of positive semidefinite constraints on top of the original SDP hierarchies \cite{PNA2010}. Optimality constraints on the solutions of specific NPO problems stemming from first-order differentiation have already been considered in the literature \cite{helton2023synchronous}, although not in the context of deriving new or improving existing numerical methods. Motivation aside, our contribution differs from this earlier work in the generality and scope of our optimality conditions, which we believe exhaust the set of first-order optimality constraints and can be applied to a large variety of NPO problems.

Our new optimality conditions come with two benefits: on one hand, they boost the speed of convergence of the original SDP hierarchy, often yielding convergence at a finite level. On the other hand, they allow us to enforce new types of constraints on NPO problems, such as demanding that the states over which the optimization takes place are the ground states of certain operators. We exploit this feature in Sections \ref{sec:stateopt}
and \ref{sec:many_body}, where we extract certified lower and upper bounds on local properties of the ground state of many-body spin systems. Remarkably, the ground state condition can be enforced in translation-invariant quantum systems featuring infinitely many particles. This allows us to make rigorous claims about the physics of quantum spin chains in the thermodynamic limit, thus solving an important open problem in condensed matter physics.

As in the classical, commutative case, careful study is required to justify exactly when the new non-commutative optimality conditions hold. While the state optimality conditions are easily seen to hold in all NPO problems, justifying the corresponding operator optimality conditions (the non-commutative Karush-Kuhn-Tucker conditions, or ncKKT conditions) requires more work. In this context, we show that: (a) \core ncKKT, the most relaxed variant of the operator optimality conditions, holds in all NPO problems; (b) normed ncKKT, a substantially strengthened version, holds in all ``bounded'' NPO problems that satisfy the non-commutative analog of the Mangasarian--Fromovitz conditions \cite{Nocedal2006}; (c) the even more restrictive strong ncKKT conditions hold if either the NPO problem is convex or its solution is achieved at a finite level of the original SDP hierarchy. 

Since the NPO formulation of quantum nonlocality only seems to satisfy \core ncKKT conditions, we also provide sufficient conditions that guarantee that either normed or strong ncKKT condition partially holds in a given NPO problem. Namely, in scenarios where the set of all variables can be partitioned into subsets that commute with each other and the remaining constraints and objective function for each of these parts are convex (satisfy the non-commutative Mangasarian--Fromovitz conditions), then a relaxed form of strong (normed) ncKKT holds. This result allows us to enforce more powerful optimality conditions on quantum nonlocality problems, with the resulting boost in convergence.

\textbf{Guide to the paper.}
The structure of this paper is as follows: 
In Section \ref{ssec:commutative} we recall optimization of commutative polynomials and then 
in Section \ref{sec:NPOprob} we move to the class of non-commutative optimization problems that we consider in this paper, and present their corresponding hierarchies of SDP relaxations. In Section \ref{sec:first_order}, 
we propose several generalizations of the first-order conditions for the non-commutative framework, which will allow us to incorporate extra constraints into our optimization problems. The necessity of the state optimality conditions will be already proven in Section \ref{sec:stateopt}. Sufficient criteria for the validity of the different forms of operator optimality are presented in Section \ref{sec:opop}. Section \ref{sec:partialop} investigates when it is legitimate to enforce the new optimality conditions partially. In Section \ref{sec:applications}, we will conduct numerical tests to see how the optimality conditions perform in practical problems. In this regard, we present two applications: the computation of the local properties of many-body quantum systems at zero temperature (Section \ref{sec:many_body}) and the maximum violation of bipartite Bell inequalities (Section \ref{sec:bell}). We then present our conclusions.

While conducting this research, we found that \citet{Fawzi2023} had independently arrived at the state optimality conditions \eqref{state_optimality}. In their interesting paper, the authors provide a sequence of convex optimization relaxations of the set of local averages of condensed matter systems at finite temperature. When the temperature parameter is set to zero, their convex optimization hierarchy turns into an SDP hierarchy, which coincides with the one presented in Section~\ref{sec:many_body} of this paper.

\section{Background on polynomial optimization theory}
\subsection{Classical polynomial optimization}\label{ssec:commutative}
In this section, we informally introduce the state-of-the-art on optimization of commutative polynomials through hierarchies of semidefinite programs. For a more detailed account, we refer the reader to \cite{Laurent2009,lasserrebook,niebook}.

Consider a classical polynomial optimization problem, i.e., a problem of the form:
\begin{equation}\label{classical_prob}
\begin{split}
p^\star:=&\min f(x)\\ 
\mbox{s.t.~}&g_i(x)\geq0,\quad i=1,\ldots ,m,\\ 
&h_j(x)=0, \quad j=1,\ldots ,m',
\end{split}
\end{equation}
where $x=(x_1,\ldots ,x_n)$ is a vector of real variables, and $f,g_i,h_j$ are real-valued polynomials thereof. Given a function $s(x)$, call $\partial_x s$ its gradient, i.e., $\partial_x s=\left(\frac{\partial s(x)}{\partial x_1},\ldots ,\frac{\partial s(x)}{\partial x_n}\right)$. Let $\mathbb{A}(x)$ denote the set of \emph{active} inequality constraints, i.e., the set of indices $i\in\{1,\ldots ,m\}$ for which $g_i(x)=0$. In this commutative scenario, the Karush--Kuhn--Tucker (KKT) conditions read: 
\beq
\begin{aligned}\label{KKT_classical}
&\exists \{\mu_i\}_i\subset\R_{\geq 0},\;\{\lambda_j\}_j\subset \R,\\
\mbox{such that }&\partial_xf\Bigr|_{x=x^\star}=\sum_{i\in\mathbb{A}(x^\star)}\mu_i\partial_x g_i\Bigr|_{x=x^\star}+\sum_j\lambda_j\partial_xh_j\Bigr|_{x=x^\star},\\
&\mu_ig_i(x^\star)=0,\quad\forall i, 
\end{aligned}
\eeq
where $x^\star$ is an optimizer of \eqref{classical_prob}.

The optimal solutions of many classical optimization problems, even those involving non-polynomial functions, are known to satisfy the KKT conditions, also known as first-order optimality conditions~\cite{Karush1939,Kuhn1951}. Sufficient criteria to ensure that the KKT conditions hold are dubbed `constraint qualifications'. A well-known instance of the latter is Mangasarian--Fromovitz constraint qualification MFCQ \cite{Nocedal2006}, which stipulates that any problem of the form \eqref{classical_prob} satisfies the (commutative) KKT conditions if
\begin{equation}
\{\partial_x h_j(x^\star)\}_j \mbox{ linearly independent},
\label{MF_classical_prelud}
\end{equation}
and there exists $\ket{z}\in\R^n$ such that\footnote{In the following, we use the physicist's ``braket'' notation, whereby a column vector $\psi$ is denoted by the ``ket'' $\ket{\psi}$; and its conjugate transpose, by the ``bra'' $\bra{\psi}$. In this notation, the scalar product between the vectors $\phi,\psi$ is represented by the ``braket'' $\braket{\phi}{\psi}$, which is anti-linear in its first argument and linear in the second.}
\begin{align}
&\braket{\partial_xg_i(x^\star)}{z}> 0,\quad \forall i\in\mathbb{A}(x^\star),\nonumber\\
&\braket{\partial_xh_j(x^\star)}z=0,\quad j=1,\ldots ,m'.
\label{MF_classical}
\end{align}

Polynomial optimization problems can be tackled via the Lasserre-Parrilo moment-SOS hierarchy of semidefinite programming (SDP) relaxations \cite{Lasserre2001,Parrilo}. This is a sequence of SDPs, with solutions $(p^k)_k$, such that $p^1\leq p^2\leq\cdots\leq p^\star$. Moreover, if the constraints of Problem \eqref{classical_prob} imply that any feasible point $x$ is bounded (more precisely, the constraints in 
\eqref{classical_prob} satisfy a
 condition known as Archimedeanity, see below), then it holds that $\lim_{k\to \infty}p^k=p^\star$.

In \cite{Nie2013}, Nie shows that generically, if we explicitly add the KKT constraints \eqref{KKT_classical} to Problem \eqref{classical_prob} and then we apply the Lasserre-Parrilo construction, we end up with a hierarchy of SDP relaxations that collapses at a finite order. Namely, there exists $\bar{k}\in\N$ such that $p^{\bar{k}}=p^\star$. Harrow, Natarajan and Wu go one step beyond in \cite{Harrow2017}, by upper bounding $\bar{k}$ for a large class of polynomial optimization problems with applications in quantum information theory.

The overarching goal of this paper is to generalize these methods to non-commutative polynomial optimization problems, which are the subject of the next subsection.

\subsection{Non-commutative polynomial optimization}
\label{sec:NPOprob}
In this section, we summarize current notions and methods in non-commutative polynomial optimization (NPO). For a deeper treatment of these matters, we recommend the reviews \cite{Navascues2012}, \cite{Burgdorf2016}.

First, let us set the notation to be used throughout the paper. In this work, we will be interested in polynomials with complex coefficients of $n$ non-commuting variables $x=(x_1,\ldots ,x_n)$. We will denote by $\C\langle x\rangle$ the set of all such polynomials; the set of monomials or words of the alphabet $x$ will be called $\langle x\rangle$. By $\C\langle x\rangle_d$ and $\langle x\rangle_d$ we will respectively mean the set of polynomials or words of degree smaller than or equal to $d$.

The set $\P$ comes with an involution $*$, which has the effect of conjugating the coefficients of the polynomial and inverting the order of the monomials. A polynomial $p(x)$ is called symmetric or Hermitian if $p(x)=p(x)^*$. Note that, according to this definition, the variables $x_1,\ldots,x_n$ themselves are Hermitian.

The variables $x_1,\ldots,x_n$ are meant to model Hermitian elements of a complex $C^*$-algebra $\A$, namely, a Banach algebra with an involution $*$ such that $\|aa^*\|=\|a\|^2$, for all $a\in\A$. Note that, for any Hilbert space $\H$, the set of bounded operators $B(\H)$ with the operator norm is a $C^*$-algebra. In fact, for every $C^*$-algebra, there exists a Hilbert space $\H$ such that $\A$ is isomorphic to some closed $*$-subalgebra of $B(\H)$. A \emph{state} $\psi$ in $\A$ is a linear functional $\psi:\A\to\C$ satisfying: $\psi(a^*)=\psi(a)^*$ for all $a\in \A$ (Hermiticity), $\psi(1)=1$ (normalization) and $\psi(aa^*)\geq 0$ (positivity). Given an algebra $\A$ with an operator representation in $B(\H)$, any normalized vector $\ket{\psi}\in \H$ defines a state $\psi$ through $\psi(\bullet)=\bra{\psi}\pi(\bullet)\ket{\psi}$, where $\pi$ is the $*$-homomorphism that maps elements of $\A$ into $B(\H)$. A state $\psi$ is \emph{normal} if there exist vectors $\{\ket{\psi_k}\}_k\subset \H$ such that $\psi(\bullet)=\sum_k\bra{\psi_k}\pi(\bullet)\ket{\psi_k}$. 

{For a tuple of Hermitian operators $X\in B(\H)^n$, $\A(X)$ will denote the unital $*$-subalgebra of $B(\H)$ generated by $X$, i.e., linear combinations of finite products of the $X_i$. Its closure in the norm is a $C^*$-subalgebra of $B(\H)$, denoted $C^*(X)$. Finally, its bicommutant (that is, the set of all elements of $B(\H)$ that commute with all elements that commute with all $X_i$) is the von Neumann algebra generated by $X$, denoted $W^*(X)$. By von Neumann's bicommutant theorem \cite[Theorem II.3.9]{takesaki1}, $\A(X)$ is dense in $W^*(X)$ in the weak or strong operator topology.}

We are ready to define NPO problems.
\begin{defin}
Let $x=(x_1,\ldots ,x_n)$ be a tuple of non-commuting variables, and let $f$,
$\{g_i:i=1,\ldots ,m\}$, $\{h_j:j=1,\ldots ,m'\}$ be symmetric polynomials on those variables. 
Then, the following program is a non-commutative polynomial optimization (NPO) problem:
\begin{equation}\label{nc_prob_hilbert}
\begin{split}
p^\star:=&\inf_{\H,X,\psi} \psi(f(X))\\ 
\mbox{s.t. }&g_i(X)\geq0,\quad i=1,\ldots ,m,\\ 
&h_j(X)=0, \quad j=1,\ldots ,m',
\end{split}
\end{equation}
where the minimization takes place over all Hilbert spaces $\H$, 
states $\psi:B(\H)\to \C$ and Hermitian operators $(X_1,\ldots ,X_n)\in B(\H)^{n}$. 
\end{defin}
A non-commutative polynomial optimization (NPO) problem~\cite{PNA2010,Burgdorf2016} is the natural analog of a polynomial optimization problem~\cite{Putinar1993,Lasserre2001,Laurent2009}.
Computing the maximal quantum violation of a Bell inequality~\cite{Bell1964,Tsirelson1987,Brunner2014} or the energy of a many-body quantum system~\cite{Anderson2018} are examples of NPO problems.

\begin{remark}
One might think of extending the definition of NPO problem to also allow equality constraints of the form $h(X)=0$, with $h$ not necessarily a Hermitian polynomial. However, such problems are also encompassed in the definition above. Indeed, any constraint of the form $h(X)=0$ is equivalent to the double constraint $\{h^1(X)=0,h^2(X)=0\}$, where
\begin{equation}
h^1:=h+h^*,h^2:=i(h-h^*)
\end{equation}
are Hermitian polynomials.
\end{remark}

\begin{remark}
Equality constraints of the form $h(x)=0$, with $h\in\P$ Hermitian, are equivalent to the inequality constraints $\{h(x)\geq 0,-h(x)\geq 0\}$. However, the relaxations that we next present to tackle Problem \eqref{nc_prob_hilbert} do make a distinction between both ways of modeling equality constraints. Ditto for the operator optimality conditions that we introduce in Sections \ref{sec:KKTcond} and \ref{sec:partialop}. In both cases, separately handling equality constraints using ideals as we do here is more efficient than handling two inequality constraints.
\end{remark}

The formulation \eqref{nc_prob_hilbert} of the NPO problem is very difficult to work with: since the optimal Hilbert space $\H^\star$ (if it exists) might well be infinite dimensional, it requires conducting an optimization over infinitely many variables!
Luckily, Problem \eqref{nc_prob_hilbert} can be relaxed to an (as we next see, more tractable) optimization over linear functionals on $\P$, i.e., elements of $\P^*$. In this regard, the following two definitions will be useful.
\begin{defin}
A linear functional $\sigma:\P\to\C$ is Hermitian if $\sigma(p)\in\R$ for all Hermitian $p\in \P$.
A (Hermitian) linear functional $\sigma:\P\to\C$ is positive if $\sigma(pp^*)\geq 0$, for all $p\in\P$.
These notions naturally extend to linear functionals $\A\to\C$ for any $*$-algebra $\A$.
\end{defin}
Given a feasible point $(\H,\psi,X)$ of Problem \eqref{nc_prob_hilbert}, the linear functional $\sigma:\P\to\C$, defined through
\begin{equation}
\sigma(p):=\psi(p(X))
\label{corresp_NPO_relax}
\end{equation}
is obviously Hermitian and positive. Moreover, it satisfies $\sigma(1)=1$.

To model the problem constraints, we will resort to the following two notions from algebraic geometry.
\begin{defin}
Given a vector of Hermitian polynomials $\h=(h_1,\ldots,h_{m'})$, the \emph{ideal generated by $\h$} is the set of polynomials of the form:
\begin{equation}
I(\h)=\left\{\sum_{j,k}p^+_{j,k}h_jp^{-}_{j,k}:\{p^+_{jk}\}_{jk}\cup\{p^{-}_{jk}\}_{jk}\subset \P\right\}.
\end{equation}
We will denote by $I(\h)_d$ the set of polynomials admitting a decomposition as above with $\deg(p^+_{j,k})+\deg(h_j)+\deg(p^{-}_{j,k})\leq d$, for all $j,k$.
\end{defin}
\begin{defin}
Given a vector of Hermitian polynomials $\g =(g_1,\ldots,g_m)$, the \emph{quadratic module generated by $\g$} is the set of polynomials of the form:
\begin{equation}
M(\g)=\left\{\sum_{k}p_{k}p^*_{k}+\sum_{i,k}p_{ik}g_ip^*_{ik}:\{p_j\}_j\cup\{p_{ik}\}_{ik}\subset \P\right\}.
\end{equation}    
We will denote by $M(\g)_d$ the set of polynomials admitting a decomposition as above with $2\deg(p_k),2\deg(p_{ik})+\deg(g_i) \leq d$, for all $i,k$.
\end{defin}
Note that both ideals and quadratic modules are convex cones of polynomials. Some such cones satisfy an important algebraic property that will appear over and over in this paper: Archimedeanity.
\begin{defin}
Let $C$ be a nonempty convex cone of polynomials.
The cone $C$ is Archimedean if, for any $p\in\P$, there exists $K\in\R^+$ such that $K-pp^*\in C$.
\end{defin}
\begin{remark}\label{rem:arch}
For $\g\in \P^m$, $\h\in \P^{m'}$, the set of polynomials $M(\g)+I(\h)$ defines a cone containing the identity. As proven in 
\cite[Lemma 4]{vidav} (or see \cite[Section 2]{cimpric}, \cite[Section 2]{Klep2007}), any such set is Archimedean iff there exists $K\in\R^+$ such that
\begin{equation}
K-\sum_i x_i^2\in M(\g)+I(\h).
\end{equation}
\end{remark}
Hermitian elements $s$ of $M(\g)+I(\h)$ are also called \emph{weighted sums of squares (SOS) polynomials}. The right-hand side of expressions of the form
\begin{equation}\label{eq:sos}
s=\sum_{k}p_{k}p^*_{k}+\sum_{i,k}p_{ik}g_ip^*_{ik}+\sum_{j,k}p^+_{j,k}h_jp^{-}_{j,k}
\end{equation}
is called an \emph{SOS decomposition for $s$}. Note that, since $s$ is Hermitian, we can rewrite the above as
\begin{equation}
s=\sum_{k}p_{k}p^*_{k}+\sum_{i,k}p_{ik}g_ip^*_{ik}+\frac{1}{2}\sum_{j,k}p^+_{j,k}h_jp^{-}_{j,k}+(p^-_{j,k})^*h_j(p^{+}_{j,k})^*.
\end{equation}

\begin{prop}
If $s\in M(\g)+I(\h)$, i.e., $s$ admits an SOS decomposition 
as in \eqref{eq:sos},
and the tuple of operators $\bar{X}\in B(\H)^{n}$ satisfies the constraints of Problem \eqref{nc_prob_hilbert}, then $s(\bar{X})$ is a positive semidefinite operator. 
\end{prop}

It is immediate that any linear functional $\sigma$, defined through Eq.~\eqref{corresp_NPO_relax} will satisfy the constraints $\sigma(s)\geq0,$ for all $s\in M(\g)+I(\h)$, a condition which we will denote by $\sigma(M(\g))\geq 0$, and $\sigma(I(\h))=0$. We arrive at the following relaxation of Problem \eqref{nc_prob_hilbert}:
\begin{equation}\label{nc_prob}
\begin{aligned}
q^\star:=&\inf_{\sigma:\P\to\C} \sigma(f)\\ 
\mbox{s.t. }&\sigma(1)=1,\\ 
&\sigma(M(\g))\geq0,\\ 
&\sigma(I(\h))=0,
\end{aligned}
\end{equation}
with $\g=(g_1,\ldots,g_m)$, $\h=(h_1,\ldots,h_{m'})$.
Clearly, $q^\star\leq p^\star$. Moreover, if the set $M(\g)+I(\h)$ is Archimedean, then $q^\star=p^\star$ and both problems admit an optimal solution \cite{Pironio2010}. This is partly a consequence of the Gelfand--Naimark--Segal (GNS) construction~\cite{Gelfand1943,Segal1947}; see \cite[Theorem I.9.14]{takesaki1} for a modern treatment. 

Given a positive linear functional $\sigma$ satisfying the conditions of Problem \eqref{nc_prob}, the GNS construction 
proceeds as follows: $\sigma$ defines a sesquilinear form 
$(a,b):=\sigma(a^*b)$ on $\P$. By modding out nullvectors $\cN$ , this form induces a scalar product, whence upon completion we arrive at a Hilbert space $\H$. 
Left multiplication by the variables $x_j$ on $\P$ induces
a linear operator $X_j$ on $\P/\cN$, and then
bounded operators $X_j$ on $\H$ (at this step the Archimedean condition enters). These operators satisfy the constraints of Problem \eqref{nc_prob_hilbert}. Finally, 
the quotient image modulo $\cN$ of the polynomial $1\in\P$ is a unit vector  $\phi\in\H$ such that
\begin{equation}\label{eq:GNS}
\bra{\phi}p(X)\ket{\phi}=\sigma(p(x)),\quad\forall p\in \P.
\end{equation}
Defining $\psi(\bullet):=\bra{\phi}\bullet\ket{\phi}$, we thus have that $(\H, X,\psi)$ is a feasible point of Problem \eqref{nc_prob_hilbert} with the same objective value as $\sigma$. 
This observation allows us to define a notion of boundedness for positive linear functionals in $\P$.
\begin{defin}
Let $\sigma:\P\to\C$ be Hermitian and positive. We say that $\sigma$ is \emph{bounded} if its GNS construction generates bounded operators $X_1,\ldots,X_n$.    
\end{defin}

Very conveniently, Problem \eqref{nc_prob} can be relaxed through hierarchies of semidefinite programs (SDP)~\cite{NPA2007,NPA2008,PNA2010}. Call $\HH_{k}$ the set of Hermitian linear functionals defined in $\P_k$, and, for $k\in \N$ large enough, consider the finite-dimensional optimization problem
\begin{equation}\label{k_relaxation}
\begin{aligned}
{p}^k:=&\inf_{\sigma^k\in \HH_{2k}} \sigma^k(f)\\ 
\mbox{s.t. }&\sigma^k(1)=1,\\
&\sigma^k(M(\g)_{2k})\geq0,
&\sigma^k(I(\h)_{2k})=0.
\end{aligned}
\end{equation}
The relaxation \eqref{k_relaxation} can be cast as a semidefinite program~\cite{Vandenberghe1996,Anjos2011} with $\mbox{dim}(\P_{2k})$ free complex variables. To implement the constraints, it suffices to consider bases of monomials. Let $\{o_a\}_a$ ($\{o^l_a\}_a$) be monomial bases of polynomials of degree at most $k$ or smaller (degree at most $k-\left\lceil\frac{\deg (g_l)}{2}\right\rceil$). Then, the matrices
\begin{equation}
\begin{aligned}
&\left(M^k(\sigma^k)\right)_{ab}:= \sigma^k(o_a^*o_b),\\ 
&\left(M_l^k(\sigma^k)\right)_{ab}:= \sigma^k((o_a^l)^*g_lo^l_b),
\end{aligned}
\label{def_mom_loc_matrices}
\end{equation}
are respectively called the \emph{$k^{th}$-order moment matrix of $\sigma^k$} and the \emph{$k^{th}$-order localizing matrix of $\sigma^k$ for constraint $g_l$} \cite{PNA2010}. Enforcing the second line of constraints in Problem \eqref{k_relaxation} boils down to demanding that the moment matrix $M^k(\sigma^k)$ and the localizing matrices $\{M^k_l(\sigma^k)\}_l$ are positive semidefinite. The last line can be similarly dealt with: it suffices to make sure, for $j=1,\ldots,m'$, that $\sigma^k(sh_js')=0$ for all monomials $s$, $s'$ with $\deg(s)+\deg(s')+\deg(h_j)\leq 2k$.

The SDP dual of \eqref{k_relaxation} can be formulated as:
\begin{equation}\label{k_dual}
q^k:=\sup\{\theta: f-\theta\in M(\g)_{2k}+I(\h)_{2k}\},
\end{equation}
with $\g=(g_1,\ldots,g_m)$, $\h=(h_1,\ldots,h_{m'})$. This problem is also an SDP~\cite{Helton2004}. Problem \eqref{k_dual} can be interpreted as finding the maximum real number $\theta$ such that the operator $f(X)-\theta$ can be proven 
(with an SOS certificate whose degree is bounded by $2k$)
positive semidefinite for all tuples of operators $X$ satisfying the constraints of Problem \eqref{nc_prob_hilbert}.

As the degree $k$ of the available polynomials  grows, one would expect the sequences of lower bounds $(q^k)_k$, $({p}^k)_k$ to better approximate the solution $q^\star$ of Problem \eqref{nc_prob}. Clearly, ${p}^1\leq {p}^2\leq\cdots\leq q^\star$, and similarly for the $q$'s. In fact, a sufficient condition for these hierarchies to be complete, in the sense that $\lim_{n\to\infty}p^n=\lim_{n\to\infty}q^n=q^\star$, is that the cone of polynomials $M(\mathbf{g})+I(\mathbf{h})$ satisfies Archimedeanity, in which case we will say that Problems \eqref{nc_prob_hilbert} and \eqref{nc_prob} are Archimedean~\cite{PNA2010}. The Archimedean property further implies that Problems \eqref{nc_prob_hilbert} and \eqref{nc_prob} are equivalent (and so $q^\star=p^\star$) and have optimal solutions (hence we can replace the $\inf$'s by $\min$'s in the problem definitions). Note that, if the Archimedean condition holds, then all the feasible points $X$ are bounded. Conversely, if the set of feasible operators is bounded, then a relation of the form $K-\sum_ix_i^2\geq 0$ can be added to the problem inequality constraints essentially without changing the problem.

\section{First-order optimality conditions}
\label{sec:first_order}
In this section, we present first-order optimality conditions for Problems \eqref{nc_prob_hilbert} or \eqref{nc_prob}. We arrived at them through a heuristic (namely, non-rigorous) argument, which we present in Appendix \ref{app:heuristic}. The argument involves mathematically ill-defined objects and it is not needed to understand the optimality conditions, or to find sufficient conditions to ensure that the conditions apply. We decided to include it in the paper nonetheless, as it could be a useful starting point to generalize our results (e.g: to NPO problems with state constraints \cite{statePoly}).

This said, first-order optimality constraints come in two types: state and operator optimality conditions. Let us start with the former.

\subsection{State optimality conditions}
\label{sec:stateopt}
We present two definitions, one for each formulation of the NPO problem.
\begin{defin}
We say that Problem \eqref{nc_prob_hilbert} satisfies \emph{state optimality conditions} if, for any  
solution $(\H^\star,\psi^\star,X^\star)$, it holds that
\begin{subequations}
\label{state_optimality_hilbert} 
\begin{flalign}
\psi^\star([f(X^\star),s(X^\star)]) = 0, \quad& \forall s\in \P,\\
\psi^\star\left(s^*(X^\star) f(X^\star) s(X^\star)-\frac{1}{2}\{f(X^\star),s^*(X^\star) s(X^\star)\}\right) \geq 0, \quad& \forall s\in \P,
\end{flalign}
\end{subequations}
where the expression $\{a,b\}$ denotes the anti-commutator of $a,b$, i.e., $\{a,b\}:=ab+ba$.
\end{defin}

\begin{defin}
We say that Problem \eqref{nc_prob} satisfies \emph{state optimality conditions} if, for any bounded solution $\sigma^\star$, it holds that
\begin{subequations}
\label{state_optimality} 
\begin{flalign}
\sigma^\star([f,s]) = 0, \quad& \forall s\in \P,\\
\sigma^\star\left(s^* f s-\frac{1}{2}\{f,s^* s\}\right) \geq 0, \quad& \forall s\in \P.
\end{flalign}
\end{subequations}
\end{defin}
\begin{remark}
{Note that, according to the definitions above, if Problem \eqref{nc_prob_hilbert} (Problem \eqref{nc_prob}) happens not to have optimal solutions (bounded optimal solutions),
 then it will satisfy the state optimality conditions. The same will hold true for the coming operator optimality conditions, introduced in Section \ref{sec:KKTcond}.}
\label{remark_ex_bound_sol}
\end{remark}

As we will soon see, state optimality conditions are easily shown to hold for all NPO problems. Before we provide the proof, though, let us introduce a few useful definitions.
\begin{defin}
Let $\H$ be a Hilbert space and let $H\in B(\H)$ be a Hermitian operator. \emph{The ground state energy of $H$}, denoted by $E_0(H)$, is the bottom of the spectrum of $H$. A \emph{ground state of $H$} is a state $\sigma\in B(\H)^*$ such that $\sigma(H)=E_0(H)$. We will denote by $\Gr (H)$ the set of ground states of $H$.
\end{defin}

The following proposition will play a role in the derivations to follow.
\begin{prop}\label{prop:eigen}
Let $\sigma$ be a ground state of $H$. Then, for any operator $Q$ it holds that 
\begin{equation}\sigma(HQ)=E_0(H)\sigma(Q)=\sigma(QH).
\label{prop:eigen0}
\end{equation}
\end{prop}
\begin{proof}
Define $\bar{H}:=H-E_0(H)$. Then we have that $\bar{H}\geq 0$ and that $\sigma(\bar{H})=0$. Since $\sigma$ is a positive functional, it holds that the matrix
\begin{equation}\label{eq:2x2mat}
\left(\begin{array}{cc}\sigma(Q^* \bar{H}Q)&\sigma(Q^* \bar{H} )\\ \sigma(\bar{H} Q)&\sigma(\bar{H})\end{array}\right)
\end{equation}
is positive semidefinite: otherwise, there would exist a vector $(a,b)$ such that $\sigma((aQ+b)^*\bar{H}(aQ+b))<0$. Since one of the diagonals in \eqref{eq:2x2mat} is zero, it follows that the off-diagonal terms are zero too. {Thus, $\sigma(\bar{H}Q)=\sigma(Q^*\bar{H})=0$, which implies $\sigma(HQ)=E_0(H)\sigma(Q)$, $\sigma(Q^*H)=E_0(H)\sigma(Q^*)$. The second equality in Eq.~\eqref{prop:eigen0} follows by replacing $Q^*$ by $Q$ in the last relation.}
\end{proof}

If $\sigma$ is the ground state of some operator $H$, then $\sigma$ can be shown to satisfy an analog of Eq.~\eqref{state_optimality_hilbert}.
\begin{prop}\label{prop:gs}
Let $\sigma\in B(\H)^*$ be the ground state of the Hermitian operator $H\in B(\H)$. Then, for any $Q\in B(\H)$, it holds that
\begin{equation}
\sigma([H,Q])=0, \quad \sigma\left(QHQ^*-\frac{1}{2}\{H,QQ^*\}\right)\geq0.
\label{state_opt_ground}
\end{equation}
\label{prop:state_opt_ground_states}
\end{prop}
\begin{proof}
Let $Q\in B(\H)$. By Proposition \ref{prop:eigen}, we have, on one hand, that
\begin{equation}
\sigma([H,Q])=\sigma(HQ)-\sigma(QH)=E_0(H)(\sigma(Q)-\sigma(Q))=0.
\end{equation}
On the other hand,
\begin{subequations}
\begin{align}
\sigma\left(QHQ^*-\frac{1}{2}\{H,QQ^*\}\right)&  =\sigma\left(QHQ^*-E_0(H)QQ^*\right)\\
&=\sigma\left(Q(H-E_0(H))Q^*\right)\\ &\geq 0,
\end{align}
\end{subequations}
where the last inequality stems from the relation $H-E_0(H)\geq 0$.
\end{proof}

We are ready to prove the universal validity of the state optimality conditions.
\begin{theo}\label{thm:stateopt}
Both Problem \eqref{nc_prob_hilbert} and Problem \eqref{nc_prob} satisfy the state optimality conditions.
\end{theo}

\begin{proof}
We will first prove the statement for Problem \eqref{nc_prob_hilbert}. Let $(\H^\star,\Psi^\star,X^\star)$ be a solution of \eqref{nc_prob_hilbert}. Then, $p^\star\geq E_0(f(X^\star))$, since $f(X^\star)-E_0(f(X^\star))\geq 0$. Further, $E_0(f(X^\star))\geq p^\star$: otherwise, $\psi(f(X^\star))<p^\star$, for any $\psi\in \mbox{Gr}(f(X^\star))$, so $\Psi^\star$ would not be an optimal solution of \eqref{nc_prob_hilbert}.

Thus, $E_0(f(X^\star))=p^\star$, and so $\Psi^\star$ is a ground state of the operator $f(X^\star)$. It follows, by Proposition \ref{prop:state_opt_ground_states}, that $\Psi^\star$ satisfies Eqs.~\eqref{state_opt_ground} for $H=f(X^\star)$ and any operator of the form $Q=s(X^\star)$, with $s\in \P$.

We next prove that Problem \eqref{nc_prob} also satisfies the state optimality conditions. Let $\sigma^\star$ be a bounded solution. Then we can use the GNS construction 
(see the paragraph leading up to \eqref{eq:GNS})
to construct a representation $(\H^\star,\Psi^\star,X^\star)$ with $\sigma^\star(s)=\Psi^\star(s(X^\star))$ for all $s\in\P$. From here on, the proof goes exactly as for Problem \eqref{nc_prob_hilbert}.
\end{proof}

\begin{remark}
Note that, if Problem \eqref{nc_prob} is Archimedean, then: (a) it has optimal solutions; and (b) all of them are bounded. Thus, one can add the extra SDP constraints
\begin{subequations}
\begin{flalign}
\sigma^k([f,s]) = 0, \quad& \forall s\in \P,\,\deg(s)+\deg(f)\leq 2k,\\
\sigma^k\left(p^* f p-\frac{1}{2}\{f,p^* p\}\right) \geq 0, \quad& \forall p\in \P,\,2\deg(p)+\deg(f)\leq 2k.
\end{flalign}
\end{subequations}
to the SDP hierarchy \eqref{k_relaxation}. This has sometimes the effect of boosting its speed of convergence, as we show in Sections \ref{sec:numerical}, \ref{sec:many_body}.

\end{remark}

By Proposition \ref{prop:state_opt_ground_states}, conditions analog to \eqref{state_optimality} allow us to enforce a new type of constraints in non-commutative optimization problems: namely, restricting the state optimization to the set of ground states of a number of operators defined through Hermitian polynomials of the problem variables. Indeed, given any set of Hermitian polynomials $b_1(x), \ldots,b_{m''}(x)$, consider a problem of the form
\begin{equation}\label{nc_prob_extended}
\begin{aligned}
p_G^\star:=&\min_{\H,X,\psi} \psi(f(X)) &&\\ 
\mbox{s.t. }&g_i(X)\geq0, && i=1,\ldots ,m,\\ 
&h_j(X)=0,&& j=1,\ldots ,m',\\ 
&\psi\in\Gr (b_l(X)), && l=1,\ldots ,m''.
\end{aligned}
\end{equation}
By Proposition \ref{prop:state_opt_ground_states}, one can relax this problem by adding to the relaxation \eqref{k_relaxation} the SDP constraints:
\begin{subequations}\label{sdp_constr_ground}
\begin{flalign}
\sigma^k([b_l,s]) = 0, \quad& \forall s\in \P,\,\deg(s)+\deg(b_l)\leq 2k,\mbox{ for }l=1,\ldots,m'',\\
\sigma^k\left(p^* b_lp-\frac{1}{2}\{b_l,p^* p\}\right) \geq 0, \quad& \forall p\in \P,\,2\deg(p)+\deg(b_l)\leq 2k,\mbox{ for }l=1,\ldots,m''.
\end{flalign}
\end{subequations}
Call $p_G^k$ the result of running program \eqref{k_relaxation} together with the SDP constraints \eqref{sdp_constr_ground}. Then, we trivially have that $p_G^1\leq p_G^2\leq\cdots\leq p_G^\star$. The following theorem provides a sufficient condition for the completeness of this hierarchy of relaxations.
\begin{theo}
\label{theo:suff_state_optimality}
Let $M(\mathbf{g})+I(\mathbf{h})$ be Archimedean, and suppose that all tuples of Hermitian operators $X\in B(\H)^n$ satisfying the problem constraints $\{g_i(X)\geq0\}_i$ and $\{h_j(X)=0\}_j$ have these two properties:
\begin{enumerate}[\rm(1)]
    \item For some unitary $U\in B(\H)$,
    \begin{equation}
    U\A(X)U^*=\bigoplus_l \A_l\otimes\id_{\bar{H}_l},
    \end{equation}
    where the direct sum is taken to be countable and, for each $l$, $\A_l$ is an irreducible 
    $*$-algebra.
    \item There exist projectors $\{P_l\in W^*(X)\}_l$ such that, for all $l$,
    \begin{equation}\label{eq:projthm}
    UP_l\A(X)P_lU^*=0\oplus (\A_l\otimes\id_{\bar{H}_l})\oplus 0.
    \end{equation}    
\end{enumerate}
Then, $\lim_{k\to\infty} p_G^k=p_G^\star$.
\end{theo}
\begin{remark}
\label{remark_suff_conds_state_opt}
The conditions of the theorem are satisfied if either of the conditions below hold:
\begin{enumerate}[\rm(1)]
\item For any $X$ such that $\{g_i(X)\geq0\}_i$ and $\{h_j(X)=0\}_j$, $W^*(X)$ is \emph{purely atomic}, i.e., it decomposes as a (at most countable) direct sum of irreducible von Neumann algebras. This could be the case, for instance, if the problem's equality constraints are such that, for some Hermitian polynomials $\{p_k(x)\}_k$, it holds that
\begin{align}
&[x_i,p_k], p_kp_j-\delta_{jk}p_k\in I(\mathbf{h}),
\;\forall j,k, \nonumber\\
&1-\sum_k p_k\in I(\mathbf{h}),\nonumber
\end{align}
and, in addition, for each $k$, the ideal $I_k$ generated by $\{p_kx_l\}_l$ is such that $I_k/\big(I(\mathbf{h})\cap I_k\big)$ admits a single irreducible representation $\pi$ with $\pi(p_k)=1$. For any feasible operator tuple $X\in B(\H)^n$, the operators $\{p_k(X)\}_k$ would thus be central projections in $W^*(X)$, yielding a block decomposition
\[
\H=\bigoplus_{k=1}^r p_k(X)\H,
\qquad
W^*(X)=\bigoplus_{k=1}^r W^*\Bigl(X\big|_{p_k(X)\H}\Bigr),
\]
with each summand $W^*(X|_{p_k(X)\H})$ irreducible.

\item For any $X$ such that $\{g_i(X)\geq0\}_i$ and $\{h_j(X)=0\}_j$, $\A(X)$ is finite dimensional. This is often the most straightforward condition to verify in applications, and is of importance since it covers, e.g., the Pauli constraints, see Section \ref{sec:many_body}. If the equality constraints $h_j(x)=0$ imply that the quotient $*$-algebra
$
\mathbb{C}\langle x\rangle / I(\mathbf{h})
$
is finite-dimensional (equivalently: for every feasible tuple $X$, the $*$-algebra $\A(X)$ generated by $X$ is finite-dimensional), then $\A(X)$ is isomorphic to a finite direct sum of full matrix algebras,
\[
\A(X) \cong \bigoplus_{\ell=1}^r M_{d_\ell}(\CC),
\]
and the decompositions  in Theorem \ref{theo:suff_state_optimality} are obtained by taking $P_\ell$ to be the (central) block projections corresponding to these simple summands.
\end{enumerate}
\end{remark}

To prove Theorem \ref{theo:suff_state_optimality}, we need the following lemma.
\begin{lemma}
\label{lemma:ground}
Let $\H$ be a Hilbert space, let $\psi$ be a normal state and let $X\in B(\H)^n$ be Hermitian operators. Let $X$ have no non-trivial invariant subspaces, i.e., for any non-zero $\ket{\phi}\in \H$, it holds that $\A(X)\ket{\phi}$ is dense in $\H$. 
Then, for any Hermitian $b\in W^*(X)$, the conditions
\begin{subequations}\label{condi_state}
\begin{flalign}
\psi\left([b,s]\right)=0,& \quad\forall s\in \A(X)\label{eq:acondi_state}\\
\psi\left(sbs^*-\frac{1}{2}\{ss^*,b\}\right)\geq 0,& \quad\forall s\in \A(X)
\end{flalign}
\end{subequations}
imply that $\psi$ is a ground state of $b$.
\end{lemma}
\begin{proof}
Note that the first condition of \eqref{condi_state} implies that $\psi([b^n,s])=0$ for all $n\in\N$, $s\in \A(X)$.
Indeed, consider the following chain of equalities:
\[
\begin{split}
\psi(b^ns) & = \psi(b(b^{n-1}s))=\psi((b^{n-1}s)b)=\psi(b(b^{n-2}sb))= \psi(b^{n-2}sb^2) = \cdots = \psi(sb^n).
\end{split}
\]
 Since $\A(X)$ is dense in $W^*(X)$ in the weak topology, the relation \eqref{eq:acondi_state} also holds for all $s\in W^*(X)$. Invoking again the density of $\A(X)$ and the equation above, it follows that, for any measurable function $r:\R\to\R$,
\begin{equation}
\psi([r(b),s])=0,\quad 
{\forall s\in W^*(X).}
\label{super_comm}
\end{equation}
For any measurable subset $A$ of $[E_0(b),-E_0(-b)]$, define $\Pi(A):=\chi_A(b)$, i.e., $\Pi(A)$ is the projector onto the generalized eigenstates of $b$ with eigenvalues within $A$. Note that, for any set of disjoint intervals $\{A_j\}\subset [E_0(b),-E_0(-b)]$ with $\cup_jA_j=[E_0(b),-E_0(-b)]$ it holds that
\begin{equation}
\psi(\bullet)= \psi(\sum_j\Pi(A_j)^2\bullet)=\psi(\sum_j\Pi(A_j)\bullet \Pi(A_j)),
\end{equation}
where the first equality follows from $\sum_j\Pi(A_j)=1$ and $\Pi(A_j)^2=\Pi(A_j)$, and the second one is a consequence of Eq.~\eqref{super_comm}, with $r(b)=\Pi(A_j)$, $s=\Pi(A_j)\bullet$.

For $\epsilon>0$ sufficiently small  consider the projectors $\Pi_G:=\Pi([E_0(b),E_0(b)+\epsilon])$, $\Pi_E:=\Pi([E_0(b)+2\epsilon,-E_0(-b)])$. Let $N\in \N$ and, for $j=0,\ldots,N-1$, define $\Pi_j:=\Pi([E_j,E_j+\delta)$, with $E_j:=E_0(h)+2\epsilon+ j\delta$ and $\delta:=\frac{-E_0(-b)-E_0(b)-2\epsilon}{N}$. Note that $\{\Pi_j\}_j$ are orthogonal projectors and  they add up to $\Pi_E$.

Setting $s=\Pi_E\tilde{s}\Pi_G$, in Eq.~\eqref{condi_state}, we have that
\begin{align}\label{eq:longcalc}
&\psi\left(sbs^*-\frac{1}{2}\{ss^*,b\}\right)\nonumber\\
&=\psi\left(\Pi_E \tilde{s}\Pi_G b \Pi_G\tilde{s}^*\Pi_E\right)-\psi\left(\Pi_E \tilde{s}\Pi_G\tilde{s}^*\Pi_Eb\right)\nonumber\\
&\leq (E_0(b)+\epsilon)\psi\left(\Pi_E \tilde{s}\Pi_G\tilde{s}^*\Pi_E\right)-\sum_j\psi\left(\Pi_j\Pi_E \tilde{s}\Pi_G\tilde{s}^*\Pi_Eb\Pi_j\right)\nonumber\\
&= (E_0(b)+\epsilon)\psi\left(\Pi_E \tilde{s}\Pi_G\tilde{s}^*\Pi_E\right)-\sum_jE_j\psi\left(\Pi_j\tilde{s}\Pi_G\tilde{s}^*\Pi_j\right)-\sum_j\psi\left(\Pi_j\tilde{s}\Pi_G\tilde{s}^*\Pi_E(b\Pi_j-E_j\Pi_j)\right)\nonumber\\
&\leq \sum_j(E_0(b)+\epsilon-E_j)\psi\left(\Pi_j \tilde{s} \Pi_G\tilde{s}^*\Pi_j\right)+\sum_j\psi\left(\Pi_j\right)\|\tilde{s}\|^2\delta\nonumber\\
&\leq -\epsilon\sum_j\psi\left(\Pi_j \tilde{s}\Pi_G\tilde{s}^*\Pi_j\right)+\|\tilde{s}\|^2\delta\nonumber\\
&= -\epsilon\psi\left(\Pi_E \tilde{s}\Pi_G \tilde{s}^*\Pi_E\right)+\|\tilde{s}\|^2\delta.
\end{align}
We can make $\delta$ as small as we wish while keeping $\epsilon$ the same. Since the first expression of Eq.~\eqref{eq:longcalc} above is non-negative, it follows that 
\begin{equation}
\psi\left(\Pi_E \tilde{s}\Pi_G \tilde{s}^*\Pi_E\right)=0.
\label{null_space}
\end{equation}
Now, $\psi$ is a normal state, and thus $\psi(\bullet)=\sum_k\bra{\psi_k}\bullet\ket{\psi_k}$, for some non-zero vectors $\ket{\psi_k}$. Eq.~\eqref{null_space} implies that 
\begin{equation}
\Pi_G\tilde{s}^*\Pi_E\ket{\psi_k}=0,
\end{equation}
for all $\tilde{s}\in \A(X)$. Since by assumption there are no invariant subspaces, this implies that $\Pi_E\ket{\psi_k}=0$ for all $k$. Hence, $\psi(\Pi_E)=0$, which implies that
\begin{equation}
\psi(b)=\psi(\Pi_E b\Pi_E)+\psi((1-\Pi_E)b(1-\Pi_E))=\psi((1-\Pi_E)b(1-\Pi_E))\leq E_0(b)+2\epsilon.
\end{equation}
Now let $\epsilon\searrow0$ to deduce $\psi(b)\leq E_0(b)$. On the other hand, $\psi(b)\geq E_0(b)$, since $\psi$ is a  state. It follows that $\psi(b)=E_0(b)$, and so $\psi$ is a ground state of $b$.
\end{proof}

\begin{proof}[Proof of Theorem \ref{theo:suff_state_optimality}]
If Problem \eqref{nc_prob_extended} is feasible and Archimedean, then one can follow the proof of Theorem $1$ in \cite{PNA2010} to conclude that there exist $(\H^\star,\psi^\star,X^\star)$ such that:
\begin{enumerate}
\item $\psi^\star(\bullet)=\bra{\phi^\star}\bullet\ket{\phi^\star}$, for some unit vector $\ket{\phi^\star}\in\H^\star$.    
\item $X^\star$ satisfies $g_i(X^\star)\geq 0$, for $i=1,\ldots,m$ and $h_j(X^\star)=0$, for $j=1,\ldots,m'$.
\item $\psi^\star(f(X^\star))=\lim_{k\to \infty}p_G^k$.
\item For $l=1,\ldots,m''$, it holds that
\begin{subequations}\label{cosa_fin}
\begin{flalign}
\psi^\star\left([b_l(X^\star),s]\right)=0,&\quad \forall s\in \A(X^\star)\\
\psi^\star\left(sb_l(X^\star)s^*-\frac{1}{2}\{ss^*,b_l(X^\star)\}\right)\geq 0,&\quad\forall s\in \A(X^\star).
\end{flalign}
\end{subequations}
\end{enumerate}
Since the operators $X^\star$ satisfy the problem constraints, by assumption we have that there exists a unitary $U\in B(\H)$ such that
\begin{equation}
UX_i^\star U^*=\bigoplus_{j}X^\star_{i,j}\otimes \id_{\bar{\H}_j},    
\label{decomp}
\end{equation}
where $\{X^\star_{i,j}:i\}\subset B(\H_j)$ generate an irreducible $*$-algebra in $B(\H_j)$.

By assumption, for any $j$, there exists a projector $P_j\in W^*(X^\star)$ onto the $j^{th}$ block of \eqref{decomp}. Then, $\pi_j(\bullet)= P_j\bullet P_j$ is a homomorphism on $\A(X^*)$, and thus $g_i(\pi_j(X^\star))\geq 0$, for all $i$ and $h_l(\pi_j(X^\star))= 0$, for all $l$.

Let $W_l(X^\star)$ be the algebra generated by $\{X^\star_{i,l}\}_i$, and let $\psi^\star_j:W_j(X^\star)\to\C$ be defined as 
\begin{equation}
 \psi^\star_j(\bullet):=\bra{\phi_j^\star}(0\oplus\bullet \otimes\id_{\bar{H}_j}\oplus 0)\ket{\phi_j^\star},
\end{equation}
with $\ket{\phi^\star_j}=U^*P_j\ket{\phi^\star}$. Then, for $\psi^\star_j\not=0$, $\frac{1}{\psi^\star_j(1)}\psi^\star_j$ is a normal state. 
In addition,
\begin{equation}
\psi^\star(\bullet)=\sum_j\psi^\star(P_j\bullet P_j)=\sum_j\psi^\star_j(\pi_j(\bullet)).
\end{equation}
Now, for all $j$ such that $\psi^\star_j\not=0$, 
Eq.~\eqref{cosa_fin} is satisfied if we replace $\psi^\star$ by $\psi^\star_j$ and all operators $\bullet$ by $\pi_j(\bullet)$. Indeed, it suffices to take $s=P_jsP_j$ in Eq.~\eqref{cosa_fin}. By Lemma \ref{lemma:ground}, it thus follows that $\psi^\star_j$ is a ground state of $\pi_j(b_l(X^\star))$ for $l=1,\ldots,m''$.

Hence, $\psi^\star(f(X^\star))$ can be expressed as a convex combination of the objective values of the feasible points $(\H_j,\frac{1}{\psi^\star_j(1)}\psi^\star_j,\pi_j(X^\star))$ of Problem \eqref{nc_prob_extended} (for $j$ such that $\psi^\star_j\not=0$). Thus, $\lim_{k\to\infty}p_G^k\geq p_G^\star$. Since $p_G^k\leq p_G^\star$ for all $k$, we arrive at the statement of the theorem.
\end{proof}

As we will see in Section \ref{sec:many_body}, the ability to conduct optimizations over ground states has important applications in many-body quantum physics.

\subsection{Operator optimality conditions}
\label{sec:KKTcond}
Operator optimality conditions are the non-commutative analogs of the Karush--Kuhn--Tucker (KKT) conditions \cite{Nocedal2006}. Such non-commutative Karush--Kuhn--Tucker (ncKKT) conditions come in several variants of different strength. To formulate them, we must first introduce an analog of the gradient for non-commutative polynomials. Closely related notions have been studied in free function theory \cite{vinnikov} and non-commutative real algebraic geometry \cite{HKMsurvey}.
\begin{defin}
The gradient of a (non-commutative) polynomial $p(x)$ is a polynomial of the original non-commuting variables $x$ and their `variations' $y=(y_1,\ldots ,y_n)$, linear in $y$. This polynomial, denoted as $\nabla_x p(y)$, is obtained from $p(x)$ by evaluating $p(x+\epsilon y)$ and keeping only the terms linear in $\epsilon$. Informally:
\begin{equation}
\nabla_xp(y)
:=\lim_{\epsilon\to 0}\frac{p(x+\epsilon y)-p(x)}{\epsilon}.
\end{equation}
The evaluation of the $x$ variables in the expression $\nabla_x p(y)$, i.e., any replacement of the form $x\to z$ will be denoted as $\nabla_xp(y)\Big|_{x=z}$.
\end{defin}

Note that this definition of a gradient is the non-commutative analog of directional derivatives.

\begin{eg}
Let $p(x)=x_1^3$. Then, $\nabla_xp(y)=y_1x_1^2+x_1y_1x_1+x_1^2y_1$.
\end{eg}

Armed with this definition, we can state the strongest of the ncKKT conditions:
\begin{defin}[Strong ncKKT]
\label{def:strong_ncKKT}
We say that the NPO Problem \eqref{nc_prob_hilbert} satisfies the \emph{strong non-commutative Karush--Kuhn--Tucker conditions} if, for any solution $(\H^\star, X^\star,\psi^\star)$ of Problem \eqref{nc_prob_hilbert}, there exist positive linear functionals $\{\mu_i:{\A}(X^\star)\to\C\}_{i=1}^m$ and Hermitian linear functionals $\{\lambda_j:{\A}(X^\star)\to\C\}_{j=1}^{m'}$ such that
\begin{subequations}
\label{strong_conditions}
\begin{align}
&\mu_i(g_i(X^\star))=0,\quad i=1,\ldots ,m,\label{strong_comp_slackness}\\
&\psi^\star\left(\nabla_xf(p)\Bigr|_{x=X^\star}\right)-\sum_i\mu_i\left(\nabla_x g_i(p)\Bigr|_{x=X^\star}\right)-\sum_j\lambda_j\left(\nabla_x h_j(p)\Bigr|_{x=X^\star}\right)=0,\quad \forall p\in\A(X^\star)^{n}.\label{strong_operator_optimality}   
\end{align}
\end{subequations}

\end{defin}
Note that, if Problem \eqref{nc_prob_hilbert} does not have optimal solutions, then it satisfies strong ncKKT, see Remark \ref{remark_ex_bound_sol}. 

The problem with the above definition is that we do not know how to enforce conditions (\ref{strong_comp_slackness}) fully. Indeed, current NPO techniques do not allow optimizing over both $\star$-algebras $\A(X^\star)$ and several (more than one) independent linear functionals defined on $\A(X^\star)$, unless $\A(X^\star)$ is uniquely defined, modulo isometries, by the problem constraints. Replacing ${\A}(X^\star)$ by $\P$ in the definition above, we arrive at a much more convenient set of constraints, phrased in terms of the solutions of Problem \eqref{nc_prob}, rather than Problem \eqref{nc_prob_hilbert}.
\begin{defin}[Normed ncKKT]
\label{def:normed_ncKKT}
We say that the NPO Problem \eqref{nc_prob} satisfies the \emph{normed non-commutative Karush--Kuhn--Tucker  conditions} if, for any bounded solution $\sigma^\star$ of Problem \eqref{nc_prob}, it holds that 
\begin{subequations}\label{eq:normed_ncKKT}
\begin{align}
&\forall i,j,\; \exists\;\mbox{\rm Hermitian }\mu_i,\lambda^{\pm}_j:\P\to\C, \quad \mu_i(M(\mathbf{g})+I(\mathbf{h})),\;\lambda_j^{\pm}(M(\mathbf{g})+I(\mathbf{h}))\geq 0,\label{normed_operator_duality}\\
\mbox{\rm{}s.t. }&\mu_i(g_i)=0,\quad i=1,\ldots ,m,\label{normed_comp_slackness}\\
&\sigma^\star\left(\nabla_x f(p)\right) - \sum_i\mu_i\left(\nabla_x g_i(p)\right)-\sum_j\lambda_j\left(\nabla_x h_j(p)\right)=0,\quad \forall p\in\P^{n}, \label{normed_operator_optimality}\\
&\mbox{\rm with }\lambda_j:=\lambda^+_j-\lambda^-_j,\quad j=1,\ldots ,m'.
\end{align}
\end{subequations}
\end{defin}
Note that we demand each functional $\lambda_j$ to be the difference of two positive functionals $\lambda^{\pm}_j$. This is a relaxation of the property that the functional $\lambda_j:\A(X^\star)\to\C$ appearing in the definition of strong ncKKT is bounded (or normed). Indeed, the Jordan decomposition 
\cite[Proposition III.2.1]{takesaki1}
shows that any Hermitian linear functional $\lambda:\A(X^\star)\to\C$ of bounded norm can be expressed as the difference between two positive linear functionals. 

If we are willing to drop the boundedness assumption on $\{\lambda_j\}_j$, then we can simply define the latter to be Hermitian functionals of $\P$. Without the variables $\lambda^{\pm}_j$, though, it is no longer possible to enforce that $\lambda_j$ is compatible with the inequality constraints $\{g_i(x)\geq 0\}_i$ 
in any meaningful way. We arrive at a weaker variant of the ncKKT conditions.
\begin{defin}[Weak ncKKT]
\label{def:weak_ncKKT}
We say that the NPO Problem \eqref{nc_prob} satisfies the \emph{weak non-commutative Karush--Kuhn--Tucker conditions} if, for any bounded solution $\sigma^\star$ of Problem \eqref{nc_prob}, it holds that
\begin{subequations}\label{operator_optimality}
\begin{align}
&\forall i,j,\; \exists\;\mbox{\rm Hermitian }\mu_i,\lambda_j:\P\to\C,\quad \mu_i(M(\mathbf{g})+I(\mathbf{h}))\geq0,\;\lambda_j(I(\mathbf{h}))= 0,\label{operator_duality}\\
\mbox{\rm{}s.t. }&\mu_i(g_i)=0,\quad i=1,\ldots ,m,\label{_comp_slackness}\\
&\sigma^\star\left(\nabla_x f(p)\right) - \sum_i\mu_i\left(\nabla_x g_i(p)\right)-\sum_j\lambda_j\left(\nabla_x h_j(p)\right)=0,\quad \forall p\in\P^{n}.
\end{align}
\end{subequations}
\end{defin}

If we were guaranteed that Problem \eqref{nc_prob} had bounded solutions and that the weak ncKKT conditions held, then we could add some further constraints to our SDP relaxation \eqref{k_relaxation}, namely:
\begin{subequations}
\label{KKT_SDP}
\begin{align}
\mathrm{For~each~}i=1,\ldots,m: &&& \nonumber\\
\exists\,\mbox{ Hermitian } \mu^k_i:\P_{2k}\to\C,\quad
& \mu^k_i(M(\mathbf{g})_{2k}+I(\mathbf{h})_{2k}) \geq 0, 
& \mu^k_i(g_i)=0, \label{KKT_SDP:mu} & \\ 
\mathrm{For~each~}j=1,\ldots, m': &&&\nonumber\\
\exists\,\mbox{ Hermitian }\lambda_j^k:\P_{2k}\to\C,\quad
&  \lambda^k_j(I(\mathbf{h})_{2k})=0,    
\end{align}
and
\begin{multline}
\sigma^k\!\left(\nabla_x f(p(x))\right)- \sum_i\mu^k_i\!\left(\nabla_x g_i(p(x))\right)+\sum_j\lambda^k_j\!\left(\nabla_x h_j(p(x))\right) = 0, \hspace{20em} \\
\forall p\in \P^{n},
 \; \deg \left(\nabla_x f\right),
 \; \deg \left(\nabla_x g_i\right),
 \; \deg \left(\nabla_x h_j\right)\leq 2k-\deg(p),\label{KKT_SDPlast} 
\end{multline}
\end{subequations}
where, for any polynomial $s$ with $\nabla_x s(p)=\sum_{i,k}s_{ik}^+p_k s_{ik}^-$, the expression $\deg(\nabla_x s)$ denotes $\max_{i,k}\deg(s_{ik}^+)+\deg(s_{ik}^-)$, and, for any $n$-tuple of polynomials $p$, $\deg(p)=\max_k\deg(p_k)$. Note that, in order to enforce the last condition \eqref{KKT_SDPlast}, it is enough to consider tuples of monomials $p$ with just one non-zero entry.

Similarly, to enforce normed ncKKT, one just needs to add to Eq.~\eqref{KKT_SDP}, for each $j = 1,\ldots m'$, the positive semidefinite constraints:
\begin{subequations}
\begin{align}
\exists \mbox{ Hermitian } \lambda^{\pm, k}_j:\P_{2k}\to\C,\hspace{1em}& \lambda^{\pm,k}_j(M(\mathbf{g})_{2k}+I(\mathbf{h})_{2k})\geq 0,\\
&\lambda^k_j=\lambda^{+,k}_j-\lambda^{-,k}_j. &&
\end{align}
\end{subequations}

As we will see, there are Archimedean problems that do not satisfy weak ncKKT conditions. This led us to introduce an even weaker variant of the ncKKT conditions, the so-called \core ncKKT conditions. In Theorem \ref{essential_theo} below we show that these conditions are satisfied by all NPO problems.

\begin{defin}[\Core ncKKT]
\label{def:core}
An NPO Problem \eqref{nc_prob} satisfies the \emph{\core ncKKT conditions} if, for any 
bounded minimizer $\sigma^\star$ of Problem \eqref{nc_prob} and for all $k\in\N$, $\{\epsilon_i\}_i \subset \R^+$ there exist Hermitian linear functionals $\mu_i^k:\P_{2k}\to\C$, $i=1,\ldots ,m$, $\lambda_j^k:\P_{2k}\to\C$, $j=1,\ldots ,m'$, such that 
{\begin{subequations}
\label{essential_KKT}
\begin{align}
\text{For each }i=1,\ldots,m:\hspace{-3em}&&&\nonumber\\
&\mu^k_i(pp^*)+\epsilon_i\|p\|^2_2\geq 0, \quad && \forall p\in\P,\; \deg (p)\leq k, \\ 
&\mu^k_i(pg_lp^*)+\epsilon_i\|p\|^2_2\geq 0, \quad && \forall p\in\P,\; \deg (p)\leq k-\left\lceil\frac{\deg(g_l)}{2}\right\rceil,\; l=1,\ldots ,m,\\ 
&\mu^k_i(I(\mathbf{h})_{2k})=0, \\ 
&\mu^k_i(sg_i)=\mu^k_i(g_is)=0, \quad && \forall s\in\P,\;\deg(s)\leq 2k-\deg(g_i), 
\label{essential_KKT:mu}
\\ 
\text{For each }j=1,\ldots,m':\hspace{-3em}&&\nonumber\\
&\lambda_j^k(I(\mathbf{h})_{2k})=0,
\end{align}
and
\begin{multline}
\sigma^\star\left(\nabla_x f(p)\right)- \sum_i\mu^k_i\left(\nabla_x g_i(p)\right)-\sum_j\lambda^k_j\left(\nabla_x h_j(p)\right)=0,\quad \\
\forall p\in \P^{n}, \deg \left(\nabla_x f\right),\deg \left(\nabla_x g_i\right),\deg \left(\nabla_x h_j\right)\leq 2k-\deg(p), %
\label{essential_operator_optimality}
\end{multline}
\end{subequations}}
where $\|p\|_2$ denotes the $2$-norm of the vector of coefficients of $p$.
\end{defin}

\begin{remark}
\Core ncKKT is a relaxation of weak ncKKT, whereby we relax each positive semidefiniteness condition in the SDP \eqref{KKT_SDP}, for each $k$. Indeed, enforcing the first two lines of Eq.~\eqref{essential_KKT} amounts to demanding that the $k^{th}$-order moment and localizing matrices (see Eq.~\eqref{def_mom_loc_matrices} for a definition) $M^k(\mu_i), M^k_l(\mu_i)$, rather than being positive semidefinite, satisfy 
\begin{equation}
M^k(\mu_i)+\epsilon_i\id\geq 0,\; M^k_l(\mu_i)+\epsilon_i\id\geq 0,\quad l=1,\ldots ,m.   
\end{equation}
The \core ncKKT conditions thus translate to non-trivial positive semidefinite constraints, which one can use to boost the speed of convergence of the SDP relaxation \eqref{k_relaxation}. 

\end{remark}

\begin{remark}
Rather than demanding $\mu_i(g_i)=0$ (as we did in Eq.~\eqref{KKT_SDP:mu}), in Eq.~\eqref{essential_KKT:mu} we demand the stronger condition $\mu_i(sg_i)=\mu_i(g_is)=0$ for all $s\in\P$.
The explanation is as follows: if $\epsilon_i$ were to vanish, then this condition would follow from $\mu_i(g_i)=0$ and the positivity of $\mu_i$ and $g_i$. For $\epsilon_i \in \mathbb R^+$, this is no longer the case, so we impose the condition by hand.    
\end{remark}

\begin{prop}
Call $\mbox{\textbf{\Core}},\ \mbox{\textbf{Weak}},\ \mbox{\textbf{Normed}},\ \mbox{\textbf{Strong}}$ the set of NPO problems respectively satisfying the \core, weak, normed and strong ncKKT conditions. Then we have that 
\begin{equation}
\mbox{\textbf{\Core}}\supsetneq \mbox{\textbf{Weak}}\supsetneq \mbox{\textbf{Normed}}\supseteq \mbox{\textbf{Strong}}.    
\end{equation}
The above relations also hold if we restrict to Archimedean NPO problems.
\end{prop}
\begin{proof}
$\mbox{\textbf{\Core}}$ is a relaxation of $\mbox{\textbf{Weak}}$, which is a relaxation of $\mbox{\textbf{Normed}}$, which is a relaxation of $\mbox{\textbf{Strong}}$. Thus, 
\begin{equation}
\mbox{\textbf{\Core}}\supseteq \mbox{\textbf{Weak}}\supseteq \mbox{\textbf{Normed}}\supseteq \mbox{\textbf{Strong}}.
\label{inclusions_proof}
\end{equation}
To prove the proposition, we need to show that there exists an Archimedean problem that satisfies \core ncKKT, but not weak ncKKT, and another one that satisfies weak ncKKT, but not normed ncKKT.

Consider thus the following instance of Problem \eqref{nc_prob}:
\begin{equation}
\begin{aligned}
\min\ & \sigma(x)\\
\text{s.t. } & \sigma(1)=1,\,\sigma(M(-x^2))\geq 0.
\end{aligned}
\label{counterex_KKT}
\end{equation}
By Remark \ref{rem:arch}, this problem is clearly Archimedean. The solution is, obviously, $\sigma^\star(x^n)=0$, for $n\geq 1$. The weak ncKKT conditions would demand the existence of a positive linear functional $\mu$, with $\mu(M(-x^2))\geq 0$, and such that
\begin{equation}
-\mu(\{x,p\})=\sigma^\star(p),\quad \forall p\in\P.
\label{rel_opt}
\end{equation}
Taking $p=x^n$, this implies that
\begin{equation}
\mu(x)=-\frac{1}{2},\; \mu(x^n)=0,\,\forall n\geq 2.
\end{equation}
The value of $\mu(1)$ is not fixed by relation \eqref{rel_opt}. This means that, for every order $k\geq 2$, the moment matrix $M_k(\mu)$ is of the form
\begin{equation}
M_k(\mu)=\left(\begin{array}{cccc}\mu(1)&-\frac{1}{2}&0&\cdots\\-\frac{1}{2}&0&0&\cdots\\0&0&0&\cdots\\\vdots&\vdots&\vdots&\ddots\end{array}\right)
\label{matrix22_counter}
\end{equation}
This matrix cannot be made positive semidefinite, no matter the choice of $\mu(1)$. Thus, Problem \eqref{counterex_KKT} does not satisfy weak ncKKT. Note, however, that, for any $\epsilon>0$, $M_k(\mu)+\epsilon\id$ is positive semidefinite if one takes $\mu(1)=\frac{1}{\epsilon}$. In addition, $-\mu(p^*x^2p)=0$, for all $p\in\P$, and so the localizing matrix associated to the constraint $-x^2\geq 0$ is positive semidefinite. Thus, Problem \eqref{counterex_KKT} satisfies \core ncKKT and the first inclusion relation of \eqref{inclusions_proof} is strict.

Next, consider the NPO problem
\begin{equation}
\begin{aligned}
\min\ & \sigma(x) \\
\text{s.t. } &\sigma(1)=1,\, 
\sigma(I(x^2))=0.
\label{example_prob_nc}
\end{aligned}
\end{equation}
This problem is also Archimedean. Its only solution is $\sigma^\star(x^n)=0$, for all $n\geq 1$, which is bounded. Weak ncKKT demands that there exists a Hermitian $\lambda:\P\to\C$ such that
\begin{equation}
\label{the_thing}
\begin{split}
\lambda(\{x,p\})&=\sigma^\star(p), \quad \forall p\in\P,   \\
\lambda(I(x^2))&=0.    
\end{split}
\end{equation}
The above is equivalent to
\begin{equation}\label{unbounded}
\begin{split}
\lambda(x)&=\frac{1}{2}, \\
\lambda(x^n)&=0,\quad \forall n\geq 2.
\end{split}
\end{equation}
These relations clearly give rise to a Hermitian linear functional $\lambda:\P\to\C$ satisfying \eqref{the_thing}
so Problem \eqref{example_prob_nc} satisfies weak ncKKT. However, Problem \eqref{example_prob_nc} does not satisfy normed ncKKT. Indeed, the requirement that the $2\times 2$ matrices $M^{\pm}_{kl}:=\lambda^{\pm}(p_kp_l^*)$, with $p_1=1,p_2=x$, be positive semidefinite implies that the off-diagonal elements $\lambda^{\pm}(x)$ vanish, and so $\lambda(x)=0$.
\end{proof}

In Section \ref{sec:opop} we derive constraint qualification criteria which, similarly to the classical case, ensure that each of the ncKKT conditions holds for a given NPO problem. These are the main conclusions:
\begin{enumerate}
    \item All NPO problems satisfy the \core ncKKT conditions (Theorem \ref{essential_theo}).
    \item 
    Under Archimedeanity and the (verifiable) algebraic condition \eqref{MF}, all the SDP relaxations derived from weak ncKKT hold (Theorem \ref{theo_null_epsilon}). That is, for any $k\in\N$, it is sound to add the SDP constraints \eqref{KKT_SDP} to the SDP relaxation \eqref{k_relaxation}.

    \item Normed ncKKT holds for Archimedean problems if the operator constraints of the NPO problem satisfy a non-commutative, algebraic version of Mangasarian--Fromovitz constraint qualification \cite{Nocedal2006} (Theorem \ref{theo_MF}).
    \item NPO problems admitting an exact sum-of-squares resolution satisfy strong ncKKT (Theorem \ref{thm:k->strong}).
    
\end{enumerate}

As it turns out, some interesting NPO problems in quantum information theory seem to satisfy neither normed nor strong ncKKT. Such NPO problems have a very peculiar structure, namely, the set of variables is partitioned in such a way that: (a) all variables within a partition commute with all the variables within any other partition; (b) there are no further operator equalities or inequalities relating variables from different partitions. In view of this, in Section \ref{sec:partialop} we further define variants of strong and normed ncKKT where Eqs.~\eqref{strong_operator_optimality} or \eqref{normed_operator_optimality} only apply to polynomials of non-commuting variables within a partition. We identify sufficient (and very mild) criteria for those `partial' operator optimality conditions to hold (Theorems \ref{theo_partial_MF} and \ref{theo_partial_conv}). In Section \ref{sec:bell}, we show that the strongest of them, partial strong ncKKT, is satisfied by the NPO formulation of the quantum nonlocality problem.

\subsection{A numerical example}
\label{sec:numerical}
We next apply the numerical tools presented so far to a problem adapted from Nie's \cite[Example 5.4]{Nie2013}. Let $x=(x_1, x_2, x_3)$, and, for $\delta, r_{\rm min}, r_{\rm max} \in \R^+$, with $r_{\rm min} < r_{\rm max}$, consider the following NPO problem:
\begin{equation}
\begin{aligned}
\min\ & \sigma\left(f(x)\right),\\
\text{s.t. } & r_{\rm min}  \leq x_1^2 + x_2^2 + x_3^2 \leq r_{\rm max}, \\
&\delta\pm i[x_k,x_l]\geq 0,\quad (k,l) = (1,2),\;(2,3),\;(3,1),
\label{test_problem}
\end{aligned}
\end{equation}
where the objective function
\begin{equation}
f(x):= \frac{1}{2} \{x_1^2, x_2^4\} + \frac{1}{2}\{x_1^4, x_2^2\} + x_3^6 
  - \frac{3}{2}\left(x_1^2 x_2^2 x_3^2 + x_3^2 x_2^2 x_1^2\right)
\end{equation}
is a symmetrized version of the Motzkin polynomial \cite{Motzkin}.

The constraint involving $r_{\rm max}$ ensures that this problem is Archimedean.
Hence, it satisfies both \core ncKKT (see Theorem \ref{essential_theo} below) and the state optimality conditions \eqref{state_optimality}.

Let us see what enforcing the \core ncKKT conditions entails. To begin with, we need to invoke the existence of eight Lagrange multipliers: $\mu^{\max},\mu^{\min}, \{\mu^{\pm}_{kl}\}_{kl}$, one for each inequality constraint. Each such multiplier is required to have moment and localizing matrices $\epsilon$-away from positive semidefinite matrices, see the first two lines of Eq.~\eqref{essential_KKT}.

Complementary slackness --Eq.~\eqref{essential_KKT:mu}-- translates to 
\begin{subequations}
\begin{align}
\mu^{\max}\left(\left(r_{\rm max}-\sum_kx_k^2\right)s\right)&=\mu^{\max}\left(s\left(r_{\rm max}-\sum_kx_k^2\right)\right)=0,\\
\mu^{\min}\left(\left(\sum_kx_k^2-r_{\rm min}\right)s\right)&=\mu^{\min}\left(s\left(\sum_kx_k^2-r_{\rm min}\right)\right)=0,\\
\mu^\pm_{k,k+1}\left(\left(\delta\pm i[x_k,x_{k+1}]\right)s\right)& =\mu^\pm_{k,k+1}\left(s\left(\delta\pm i[x_k,x_{k+1}]\right)\right)=0,\quad k=1,2,3,
\end{align}
\end{subequations}
for all monomials $s\in\P$. In the equation above, we respectively identify $\mu_{3,4}^\pm, x_4$ with $\mu_{3,1}^\pm, x_1$.

Finally, Eq.~\eqref{essential_operator_optimality} reduces to:
\begin{equation}
\begin{split}
\sigma^\star\left(\nabla_xf(p)\right)&=-\mu^{\max}\left(\sum_k\{x_k,p_k\}\right)+\mu^{\min}\left(\sum_k\{x_k,p_k\}\right)\\
&\phantom{{}={}}
+i\sum_{k}\left\{\mu_{kl}^{+}\left([x_k,p_{k+1}]+[p_k,x_{k+1}]\right)-\mu_{kl}^{-}\left([x_k,p_{k+1}]+[p_k,x_{k+1}]\right)\right\},
\end{split}
\end{equation}
for any tuple of polynomials $p=(p_1,p_2,p_3)$ (again, $x_4,p_4$ are to be understood as $x_1, p_1$).

Setting $\delta=0.01$, $r_{\rm min} =1$, and $r_{\rm max}=7.5$, we ran the SDP relaxation \eqref{k_relaxation} for $k=4$, with and without \core ncKKT and state optimality conditions. Table~\ref{table:SimpleExample} shows the resulting lower bounds on the solution of Problem \eqref{test_problem}\footnote{
A MATLAB implementation of this is available at \url{https://github.com/ajpgarner/nckkt-nie-motzkin}.
}.

 \begin{table}[h!]
 \centering
 \begin{tabular}{ r| c| c }
 &~without ncKKT~&~\core ncKKT~\\\hline 
 without state optimality & $-1.9774$ & $-1.8179$   \\  
 with state optimality & $-1.8522$ & $-1.4266$
 \end{tabular}
 \caption{Lower bounds on the solution of Problem \eqref{test_problem}, computed at level $k=4$ of the hierarchy of SDP relaxations \eqref{k_relaxation}, under different optimality constraints. \Core ncKKT was enforced with $\epsilon_i=0.0001$ for all $i$.}
 \label{table:SimpleExample}
 \end{table}

This example illustrates that even the very general \core ncKKT conditions can greatly boost the speed of convergence of the hierarchy of SDPs \eqref{k_relaxation}, especially when combined with the state optimality conditions \eqref{state_optimality}.

\section{Non-commutative constraint qualification}
\label{sec:opop}
In this section, we provide sufficient criteria for the different variants of the ncKKT conditions to hold. The roadmap from \core ncKKT to strong ncKKT is as follows:

\begin{itemize}
    \item 
In section \ref{sec:essential}, we prove that all NPO problems satisfy the \core ncKKT conditions. Moreover, we will further identify, in subsection \ref{sec:epsilon_zero}, `boundedness conditions' that allow replacing the variables $\mu_i^k$ in Definition \ref{def:core} by bounded, positive linear functionals $\{\mu_i\}_i$ independent of $k$ and satisfying (\ref{operator_duality}) and (\ref{_comp_slackness}). This is not the same as having weak ncKKT conditions: those would in addition demand the existence of $k$-independent linear functionals $\{\lambda_j\}_j$. NPO problems satisfying the boundedness condition admit, however, the same SDP relaxations as problems satisfying weak ncKKT. 

\item
To prove the existence of $k$-independent $\{\lambda_j\}_j$, we introduce in section \ref{sec:linear_indep} the notion of linear independence of non-commutative gradients. If both the boundedness condition and gradient linear independence hold, then one can prove that there exist bounded, $k$-independent linear functionals $\{\mu_j\}_i$, $\{\lambda_j\}_j$ complying with the conditions of normed ncKKT, see section \ref{sec:MF}. 

\item 
We were not able to carry this argument any further to find even more stringent conditions that would guarantee strong ncKKT. Thus, we switched to a completely different route: in section \ref{sec:strong}, we prove that NPO problems admitting an exact SOS resolution satisfy strong ncKKT.
\end{itemize}

\subsection{\Core KKT conditions}
\label{sec:essential}
In this section we will prove the universal validity of the \core ncKKT conditions.
\begin{theo}
\label{essential_theo}
Any \emph{bounded} solution $\sigma^\star$ of the NPO problem \eqref{nc_prob} satisfies Eq.~\eqref{essential_KKT}. That is, all NPO problems satisfy the \core ncKKT conditions.
\end{theo}

\begin{remark}
\label{epsilon_zero}
Setting $\epsilon_i=0$ for all $i$, \core ncKKT reduces to the SDP constraints conditions \eqref{KKT_SDP} implied by weak ncKKT on the SDP relaxation \eqref{k_relaxation}. The reason why Theorem \ref{essential_theo} demands non-zero $\epsilon_i$ is to prevent that $\mu^k_i(p)=\infty$ for some polynomials $p\in\P_{2k}$. See, e.g., the NPO problem \eqref{counterex_KKT}: the requirement that the moment and localizing matrices of $\mu^k$ be $\epsilon$-close to positive implies that $\mu(1)\geq O\left(\frac{1}{\epsilon}\right)$.
\end{remark}

To prove Theorem \ref{essential_theo}, we will need the following technical result.
\begin{lemma}
\label{lemma_ODE}
Let the $n$-tuple of symmetric polynomials $q(x)$ satisfy 
\begin{equation}
 \nabla_xh_j(q(x))\in I(\mathbf{h}),\quad j=1,\ldots ,m',
 \label{cond_zeros}
\end{equation}
and let $X^\star\in B(\H^\star)^{n}$ be a tuple of Hermitian operators acting on the Hilbert space $\H^\star$, such that 
\begin{equation}
h_j(X^\star)=0,\quad j=1,\ldots ,m'.
\label{cond_ann}
\end{equation}
Then, there exists $\epsilon>0$ and an analytic trajectory $\{X(t):t\in [-\epsilon,\epsilon]\}\subset B(\H^\star)^{n}$ of Hermitian operators such that 
\begin{subequations}\label{boundary}
\begin{flalign}
h_j(X(t))&=0,\quad j=1,\ldots ,m',\; t\in[-\epsilon,\epsilon],\label{boundary:1}\\ 
X(0)&=X^\star,\label{boundary:2}\\ 
\frac{dX(t)}{dt}\Bigr|_{t=0}&=q(X^\star).\label{boundary:3}
\end{flalign}
\end{subequations}
If, in addition, $g_i(X^\star)\geq0$ for all $i$, and, for some $\{s_i\in\P\}_i$, it holds that
\begin{equation}
s_i(x)g_i(x)+g_i(x)s_i(x)^*+\nabla_xg_i(q(x))
\in M(\g)+I(\h)
,\quad i=1,\ldots ,m,
\label{SOS_boundary}
\end{equation}
then the trajectory $X(t)$ also satisfies:
\begin{equation}
g_i(X(t))\geq 0,\quad i=1,\ldots ,m,
\label{positivity_diff}
\end{equation}
for $t\in[0,\epsilon]$.

\end{lemma}

\begin{proof}
Consider the system of ordinary differential equations (ODEs)
\beq
\begin{aligned}
\frac{dX(t)}{dt}&=q(X),\\ 
X(0)&=X^\star.
\end{aligned}
\eeq
Since $X_1^\star,\ldots ,X^\star_n$ are bounded, we can apply the Cauchy-Kovalevskaya theorem (in Banach spaces, see, e.g., \cite{rosenbloom61} or \cite{SW76}\footnote{The basic idea of the proof of the Cauchy-Kovalevskaya theorem is fairly simple. Differentiate the ODE to compute a formal power series for $X(t)$ about $0$. The method of majorants was developed by Cauchy and Kovalevskaya to prove that this power series has a positive radius of convergence and thus determines an analytic function that solves the ODE.\label{foot:1}}) and conclude that there exists a ball in the complex plane of radius $\epsilon>0$ and with center at $0$ where the solution of this differential equation is analytic. In particular, for $t\in[-\epsilon,\epsilon]$, $X(t)$ exists and, from the equation above, it satisfies the boundary conditions \eqref{boundary:2} and \eqref{boundary:3}.

Being polynomials of a tuple of analytic operators, $\{h_j(X(t))\}_j$ are also analytic in the region $\{t\in\C:|t|\leq\epsilon\}$. {Therefore, for $t\in[-\epsilon,\epsilon]$,
\begin{equation}
h_j(X(t))=\sum_{k=0}^\infty\frac{t^k}{k!}\frac{d^kh_j(X(\tau))}{d\tau^k}\Bigr|_{\tau=0}.
\label{taylor_exp}
\end{equation}}

{Due to relation \eqref{cond_zeros}, we have that
\begin{equation}\label{induct_hypo}
\begin{aligned}
\frac{dh_j(X(t))}{dt}&=\nabla_xh_j\left(\frac{dX(t)}{dt}\right)
=\nabla_xh_j(q(X(t)))\\ 
&=\sum_{j',l}r^+_{jj'l}(X(t))\,h_{j'}(X(t))\,r^-_{jj'l}(X(t)).
\end{aligned}
\end{equation}
By induction, it therefore holds that 
\begin{equation}
\frac{d^kh_j(X(t))}{dt^k}=\sum_{j',l}p^+_{jj'kl}(X(t))\,h_{j'}(X(t))\,p^-_{jj'kl}(X(t)),
\end{equation}
for some polynomials $\{p^{\pm}_{jj'kl}:j,j',k,l\}\subset\P$. Since $h_{j'}(X(0))=0$ for all $j'$, we have that
\begin{equation}
\frac{d^kh_j(X(\tau))}{d\tau^k}\Bigr|_{\tau=0}=0,\quad\forall k.
\end{equation}
Consequently, by Eq.~\eqref{taylor_exp}, $h_j(X(t))=0$, for $t\in[-\epsilon,\epsilon]$.}

{Let us now assume that Eq.~\eqref{positivity_diff} holds. Define the function $G_i(t):=\omega_i(t)g_i(X(t))\omega_i(t)^*$, where $\omega_i(t)$ is the solution of the differential equation:
\begin{equation}
\frac{d\omega_i(t)}{dt}=\omega_i(t)s_i(X(t)),\quad \omega_i(0)=\id,
\end{equation}
which, by the Cauchy-Kovalevskaya theorem, is analytic for short times $t$. For short times, $\omega(t)$ is also close to the identity, and thus invertible. Let $\epsilon'>0$ be such that $\epsilon\geq \epsilon'$ and $\{\omega_i(t)\}_i$ are analytic and invertible for $t\in [0,\epsilon']$.}

Being the product of three analytic functions, the function $G_i(t)$ is also analytic in $t\in[0,\epsilon']$. Taking differentials, we find
\begin{align}\label{eq:dGi}
&\frac{dG_i(t)}{dt}=\omega_i(t)(s_i(X(t))g_i(X(t))+g_i(X(t))s^*_i(X(t))+\nabla_xg_i(q(X(t)))\omega_i(t)^*\nonumber\\
&\overset{\eqref{SOS_boundary}}{=}\omega_i(t)\left(\sum_ls_{il}^*(X(t))s_{il}(X(t))+\sum_{i',l}s_{ii'l}^*(X(t)) g_{i'}(X(t))s_{ii'l}(X(t))\right.\nonumber\\
&\phantom{{}={}}\left.\; +\sum_{j,l}s^+_{ijl}(X(t)) h_j(X(t))s^-_{ijl}(X(t))\right)\omega_i(t)^*\nonumber\\
&\overset{\eqref{boundary:1}}{=}\omega_i(t)\left(\sum_ls_{il}^*(X(t))s_{il}(X(t))+\sum_{i',l}s_{ii'l}^*(X(t))\omega_{i'}(t)^{-1} G_{i'}(t)(\omega_{i'}(t)^{-1})^*s_{ii'l}(X(t))\right)\omega_i(t)^*.
\end{align}

{Since all the operators are bounded, this equation might be solved through any standard numerical method, e.g.: Euler's explicit method. Take $\Delta>0$ and consider the following time discretization:
\begin{equation}
\begin{split}
&G_i(k+1;\Delta)=G_i(k;\Delta)+\Delta\omega_i(\Delta k)\left(\sum_ls_{il}^*(X(\Delta k))s_{il}(X(\Delta k))\right)\omega_i(\Delta k)^*\\
& +\Delta\omega_i(\Delta k)\left(\sum_{i'l}s_{ii'l}^*(X(\Delta k))\omega_{i'}(\Delta k)^{-1} G_{i'}(k;\Delta))(\omega_{i'}(\Delta k)^{-1})^*s_{ii'l}(X(\Delta k))\right)\omega_i(\Delta k)^*.
\end{split}
\end{equation}
Clearly, starting from positive semidefinite operators
\begin{equation}
G_i(0;\Delta)=g_i(X^\star)\geq 0, \quad i=1,\ldots ,m,
\end{equation}
it holds that $G_i(k;\Delta)\geq 0$, for all $k$. }

{Suppose that the Euler explicit method converged: namely, for any $t\in [0,\epsilon']$, $\lim_{k\to \infty}G_i(k;\frac{t}{k})=G_i(t)$, for all $i$. Then, the observation above would imply that $\omega_i(t)g_i(X(t))\omega_i(t)^*\geq 0$, for $t\in[0,\epsilon']$ and for all $i$. Since $\omega_i(t)$ is invertible, this would imply that $g_i(X(t))\geq 0$ and thus Eq.~\eqref{positivity_diff} would hold.}

{To finish the proof, it thus suffices to show that the Euler method converges. This is not difficult: we just need to follow the proof for ordinary differential equations. First, for fixed $\Delta$, define $t_k:=k\Delta$ and the error
\begin{equation}
\delta_k:=\max_i\|G_i(t_k)-G_i(k;\Delta)\|.
\end{equation}
Also, rewrite the right-hand side of Eq.~\eqref{eq:dGi} as
\begin{equation}
\frac{dG(t)}{dt}=\rho(t,G(t)).
\label{diff_gen}
\end{equation}
Note that $\rho(t,G(t))$ is affine linear  in $G(t)$. This implies that it is Lipschitz, i.e., for all $t\in [0,\epsilon']$, it holds that $\|\rho(t,G)-\rho(t,G')\|\leq K\|G-G'\|$, for some $K\in\R^+$.
Eq.~\eqref{diff_gen} implies that
\begin{equation}
G(t)=\int_{0}^td\tau \rho(\tau,G(\tau)). \footnote{Here and below we only use the Riemann integral in Banach spaces, where 
the basic properties are essentially the same as those in Euclidean spaces. For details and a deeper study in integration in Banach spaces we refer to one the standard works, e.g.,
\cite{dieudonne,cartan,ambrosetti}}   
\end{equation}
Integrating by parts, we have that
\begin{equation}
G(t_{j+1})=G(t_j)+\rho(t_j,G(t_j))\Delta+\int_{t_j}^{t_{j+1}}d\tau(t_{j+1}-\tau) \frac{d}{d\tau}\rho(\tau,G(\tau)).
\end{equation}
Subtracting on both sides by $G(j+1;\Delta)=G(j;\Delta)+\Delta \rho(t_j,G(j;\Delta))$, we arrive at
\begin{equation}
\begin{split}
\delta_{j+1} &\leq (1+ K\Delta)\delta_j+\left\|\int_{t_j}^{t_{j+1}}d\tau(t_{j+1}-\tau) \frac{d}{d\tau}\rho(\tau,G(\tau))\right\|\\
&\leq (1+ K\Delta)\delta_j+\left(\int_{t_j}^{t_{j+1}}(t_{j+1}-\tau)^2d\tau\right)^{1/2} \left(\int_{t_j}^{t_{j+1}}d\tau\left\|\frac{d}{d\tau}\rho(\tau,G(\tau))\right\|^2\right)^{1/2}\\
&\leq (1+ K\Delta)\delta_j +\gamma \Delta^{3/2},
\label{recursion_deltas}
\end{split}
\end{equation}
with
\begin{equation}
\gamma=\frac{1}{\sqrt{3}}\left(\int_{0}^{\epsilon'}d\tau\left\|\frac{d}{d\tau}\rho(\tau,G(\tau))\right\|^2\right)^{1/2}.
\end{equation}
It can be verified, by induction, that 
\begin{equation}\label{eq:above}
\delta_{j}\leq e^{(\gamma+K)t_j}\Delta^{1/2}.    
\end{equation}
Indeed, $\delta_0=0\leq \Delta^{1/2}$. Now, suppose that the equation \eqref{eq:above} above holds for $j$. Then, by Eq.~\eqref{recursion_deltas}, we have that
\begin{align}
\delta_{j+1}&\leq (1+K\Delta)e^{(\gamma+K)t_j}\Delta^{1/2}+\gamma\Delta^{3/2}\nonumber\\
&\leq(1+K\Delta)e^{(\gamma+K)t_j}\Delta^{1/2}+e^{(\gamma+K)t_j}\gamma\Delta^{3/2}\nonumber\\
&=(1+(K+\gamma)\Delta)e^{(\gamma+K)t_j}\Delta^{1/2}\\
&\leq e^{(K+\gamma)\Delta}e^{(\gamma+K)t_j}\Delta^{1/2}\nonumber\\
&=e^{(\gamma+K)t_{j+1}}\Delta^{1/2}\nonumber.
\end{align}
For any $j\in\{0,\ldots,\left\lfloor\frac{\epsilon'}{\Delta}\right\rfloor\}$, it thus follows that $\delta_j\leq e^{(\gamma+K)\epsilon'}\Delta^{1/2}$, which tends to zero in the limit $\Delta\to 0$. Thus, the Euler explicit method converges.}
\end{proof}

\begin{proof}[Proof of Theorem \ref{essential_theo}.]
It is enough to prove the statement for $\epsilon_i=\epsilon>0$, for all $i$. Let $\sigma^\star$ be a bounded minimizer of Problem \eqref{nc_prob}. For fixed $k\in\N$, consider the semidefinite program
\begin{subequations}\label{eq:SDPP}
\begin{align}
\mathbb{P}:=&\min_{\epsilon,\mu,\lambda} \epsilon\nonumber\\
\mbox{s.t. }&\epsilon\geq 0,\nonumber\\
&\mu^k_i(pp^*)+\epsilon\|p\|^2_2\geq 0,\quad \forall p\in\P,\; \deg (p)\leq k,\label{moment_proof}\\
&\mu^k_i(pg_lp^*)+\epsilon\|p\|^2_2\geq 0, \quad \forall p\in\P,\; \deg (p)\leq k-\left\lceil\frac{\deg(g_l)}{2}\right\rceil,\;i=1,\ldots ,m,\;l=1,\ldots ,m,\label{localizing_proof}\\ 
&\mu^k_i(I(\mathbf{h})_{2k})=0, i=1,\ldots ,m,\label{feas_proof}\\ 
&\mu^k_i(sg_i)=\mu^k_i(g_is)=0,\quad \forall s\in\P,\;\deg(s)\leq 2k-\deg(g_i)\quad i=1,\ldots ,m,\label{slackness_proof}\\ 
&\lambda_j^k(I(\mathbf{h})_{2k})=0,j=1,\ldots ,m',\label{feas_proof2}\\
&\sigma^\star\left(\nabla_x f(p(x))\right)-\sum_i\mu^k_i\left(\nabla_x g_i(p(x))\right)-\sum_j\lambda^k_j\left(\nabla_x h_j(p(x))\right)=0,\quad \forall p\in \P^{n}, \label{balance_proof}\\ 
&\hspace{1cm}\deg \left(\nabla_x f\right),\deg \left(\nabla_x g_i\right),\deg \left(\nabla_x h_j\right)\leq 2k-\deg(p),\label{degrees_SDP}
\end{align}
\end{subequations}
where the minimization takes place over Hermitian functionals $\{\mu_i^k\}_i,\{\lambda_j^k\}_j$.
We next prove that this problem has feasible points. 

{First, let us focus on Eqs.~\eqref{feas_proof}--\eqref{balance_proof}. Let $\M$ be the set of polynomials of the form $r+r^*$ and $i(r-r^*)$, where $r\in\langle x\rangle$. Note that all the elements of $\M$ are Hermitian. Moreover, $\M_{k}$, the set of elements of $\M$ with degree smaller than or equal to $k$, satisfies $\mbox{span}(\M_k)=\P_k$. With this notation, satisfying Eqs.~\eqref{feas_proof}--\eqref{balance_proof} with Hermitian functionals can be seen equivalent to the feasibility problem:
\begin{align}
&\exists \{\mu_i^k(p):p\in \M_{2k}\}_i, \{\lambda_j^k(p):p\in \M_{2k}\}_j\subset \R,\nonumber\\
\mbox{s.t. }&\mu^k_i\left(r^+h_jr^-+(r^-)^*h_j(r^+)^*\right)=0,\;\forall r^+,r^-\in \langle x\rangle,\, \deg(r^-)+\deg(r^+)\leq \deg(h_j), \, \forall i,j\nonumber\\
&\mu^k_l\left(i(r^+h_jr^--(r^-)^*h_j(r^+)^*)\right)=0,\;\forall r^+,r^-\in \langle x\rangle,\, \deg(r^-)+\deg(r^+)\leq \deg(h_j),\, \forall j,l,\nonumber\\
&\mu^k_l(sg_l+g_ls^*)=0,\quad \forall s\in \langle x\rangle,\,\deg(s)\leq 2k-\deg(g_l),\,\forall l,\nonumber\\
&\mu^k_l(i(sg_l-g_ls^*))=0,\quad \forall s\in \langle x\rangle,\,\deg(s)\leq 2k-\deg(g_l),\,\forall l,\nonumber\\
&\lambda^k_l\left(r^+h_jr^-+(r^-)^*h_j(r^+)^*\right)=0,\;\forall r^+,r^-\in \langle x\rangle, \,\deg(r^-)+\deg(r^+)\leq \deg(h_j), \, \forall j,l,\nonumber\\
&\lambda^k_l\left(i(r^+h_jr^--(r^-)^*h_j(r^+)^*)\right)=0,\;\forall r^+,r^-\in \langle x\rangle,\, \deg(r^-)+\deg(r^+)\leq \deg(h_l),\, \forall j,l,\nonumber\\
&\sum_i\mu^k_i\left(\nabla_x g_i(r)\right)+\sum_j\lambda^k_j\left(\nabla_x h_j(r)\right)=\sigma^\star\left(\nabla_x f(r)\right),\quad \forall r\in \M^n, \nonumber\\ 
&\hspace{1cm}\deg \left(\nabla_x f\right),\deg \left(\nabla_x g_i\right),\deg \left(\nabla_x h_j\right)\leq 2k-\deg(r).
\label{linear_sys}
\end{align}
The left-hand side of each constraint can be expressed as a real linear combination of the real variables defined in the first line. The feasibility problem \eqref{linear_sys} is therefore a real system of linear equations. Note that the first conditions are homogeneous (namely, satisfied by taking all variables equal to zero); the only non-homogeneous conditions can be found in the last equality in the line next to the last of Eq.~\eqref{linear_sys}.

An arbitrary real linear combination of the equalities of the last line results in an expression of the form
\begin{equation}
\sum_i\mu^k_i\left(\nabla_x g_i(p)\right)+\sum_j\lambda^k_j\left(\nabla_x h_j(p)\right)=\sigma^\star\left(\nabla_x f(p)\right),
\end{equation}
where the vector of polynomials $p\in\P^n$ has Hermitian entries. By basic linear algebra, the system is not solvable iff, for some $p\in\P^n$, the right-hand side of the equation above is non-zero, whereas the left-hand side can be expressed a real linear combination of the homogeneous constraints. Since the homogeneous constraints do not relate the different variable types (namely, those that go with each functional), this would imply that the polynomial evaluated by $\lambda^k_j$ belongs to $I(\mathbf{h})_{2k}$, and the polynomial evaluated by $\mu_i^k$ is of the form $s_ig_i+g_is_i^*+ r_i$, for some polynomials $s_i\in\P$, $r_i\in I(\mathbf{h})_{2k}$.

We have just proven that there exist Hermitian $\lambda,\mu$ satisfying Eqs.~\eqref{feas_proof}--\eqref{balance_proof} iff, for all Hermitian $p\in\P^{n}$ and (not necessarily Hermitian) $\{s_i\}_i\subset\P$, the conditions 
\begin{equation}
\begin{split}
&\deg \left(\nabla_x f\right),\deg \left(\nabla_x g_i\right),\deg \left(\nabla_x h_j\right)\leq 2k-\deg(p),\\
&\mbox{deg}(s_i)\leq 2k-\mbox{deg}(g_i),i=1,\ldots ,m,
\label{degrees_p}
\end{split}
\end{equation}
and
\begin{equation}
\begin{split}
&\nabla_xg_i(p)-s_ig_i-g_is_i^*\in I(\mathbf{h})_{2k},i=1,\ldots ,m,\\
&\nabla_x h_j(p)\in I(\mathbf{h})_{2k}, j=1,\ldots ,m',
\label{feas_p}
\end{split}
\end{equation}
imply that
\begin{equation}
\sigma^\star(\nabla_x f(p))=0.
\label{null_grad}
\end{equation}}

Let $(\H^\star, X^\star, \psi^\star)$ be the result of applying the GNS construction on $\sigma^\star$, and suppose that $p\in\P^{n}$ satisfies Eqs.~\eqref{degrees_p} and \eqref{feas_p}. Then, $X^\star, q:=\pm p$ satisfy the conditions of Lemma \ref{lemma_ODE}, and thus there exist two feasible operator trajectories $\{X^{\pm}(t):t\in[0,\delta]\}\subset B(\H^\star)$ such that
\begin{equation}
X^{\pm}(0)=X^\star,\; \frac{dX^{\pm}}{dt}\Bigr|_{t=0}=\pm p(X^\star).
\end{equation}
Hence,
\begin{equation}
\pm\sigma^\star(\nabla_xf(p))=\sigma^\star(\nabla_xf(\pm p))=\psi^\star\left(\nabla_xf(\pm p(X^\star))\Bigr|_{X=X^\star}\right)=\frac{d\psi^\star\left(f(X^\pm(t))\right)}{dt}\Bigr|_{t=0}\geq 0,
\end{equation}
where the last inequality follows from the fact that $X=X^\star$ minimizes the expression $\psi^\star\left(f(X)\right)$ over tuples $X\in B(\H^\star)^{n}$ of feasible operators. Thus, Eq.~\eqref{null_grad} holds whenever $p,s$ satisfy Eqs.~\eqref{degrees_p}, \eqref{feas_p}.

It follows that there exist $\mu^k,\lambda^k$ satisfying conditions \eqref{feas_proof}--\eqref{balance_proof}. For each $i$, the corresponding moment and localizing matrices of $\mu^k_i$ might not be positive semidefinite, but, 
one can find $\epsilon>0$ large enough such that conditions \eqref{moment_proof}, \eqref{localizing_proof} hold. We therefore conclude that the SDP \eqref{eq:SDPP} is feasible.

The dual of \eqref{eq:SDPP} is the following SDP:
\beq
\begin{aligned}\label{eq:SDPQ}
\mathbb{P}^*:=&\max_{q,s,Z} -\sigma^\star(\nabla_xf(q))\\
\mbox{s.t. }&\nabla_xg_i(q)+s_ig_i+g_is^*_i-\sum_{a,b}(Z_i)_{ab}o_ag_io_b^*-\sum_{l,a,b}(Z_{i,l})_{ab}o^l_ag_l(o^l_b)^*\in I(\mathbf{h})_{2k},\\
&\nabla_xh_j(q)\in I(\mathbf{h})_{2k},\quad j=1,\ldots ,m',\\
&Z_i\geq 0,\; Z_{i,l}\geq 0,\quad i,k=1,\ldots ,m,\\
&\sum_{i,l}\tr(Z_{i,l})\leq 1.
\end{aligned}
\eeq
where $\{o_k\}_k$, $\{o^{l}_k\}_k$ are, respectively, monomial bases for polynomials of degree $k$ and $k-\left\lceil\frac{\deg(g_l)}{2}\right\rceil$, and $Z_i$, $Z_{i,l}$ are Hermitian matrices.

That is, in Problem \eqref{eq:SDPQ} one needs to maximize $-\sigma^\star(\nabla_xf(q))$ over the polynomials $q, \{s_i\}_i$ such that $\nabla_xh_j(q)=0$ for all $j$, and the polynomials $\nabla_xg_i(q)+s_ig_i+g_is_i^*,\, i=1,\ldots ,m$ are sums of weighted squares satisfying certain normalization constraints. 

We claim that the solution of \eqref{eq:SDPQ} is zero. This value can be achieved, e.g., by taking $q=Z_i=Z_{i,l}=s_i=0$ for all $i,l$. To see that zero is the optimal value, consider any feasible point $q,s$ of \eqref{eq:SDPQ}. By definition, $q, X^\star$ satisfy the conditions of Lemma \ref{lemma_ODE}. Thus, there exists $\delta>0$ and a trajectory of feasible operators $\{X(t):t\in[0,\delta]\}$ such that 
\begin{equation}
X(0)=X^\star,\; \frac{dX}{dt}\Bigr|_{t=0}=q(X^\star).
\end{equation}
Hence,
\begin{equation}
-\sigma^\star(\nabla_xf(q))=-\frac{d\psi^\star\left(f(X(t))\right)}{dt}\Bigr|_{t=0}\leq 0.
\end{equation}
The solution of Problem \eqref{eq:SDPQ} is therefore zero.

Now, the SDP \eqref{eq:SDPP} is bounded from below by $0$ and it admits strictly feasible points (by taking $\epsilon$ large enough). By Slater's criterion, the problem thus satisfies strong duality, and so the solutions of \eqref{eq:SDPP}, \eqref{eq:SDPQ} coincide. This implies that one can find feasible points of Problem \eqref{eq:SDPP} for any $\epsilon>0$. Hence, conditions \eqref{essential_KKT} hold for arbitrary $\{\epsilon_i\}_i\subset \R^+$.
\end{proof}

\begin{remark}
Note that our derivation of Theorem \ref{essential_theo} does not make full use of the fact that $(\H^\star, X^\star, \psi^\star)$ is an optimal solution of Problem \ref{nc_prob_hilbert}. It is enough to assume that, for a fixed state $\psi^\star$, the operators $X^\star$ minimize the objective function $\psi^\star (f(X))$ (under the corresponding operator constraints). With this in mind, it is easy to re-derive, via Theorem \ref{essential_theo}, the first-order optimality conditions for synchronous non-local games presented in \cite[Prop. 7.1]{helton2023synchronous}, see Appendix \ref{app:rederivation_helton}.
\end{remark}

\subsubsection{Setting $\epsilon_i$ to zero}
\label{sec:epsilon_zero}
In Remark \ref{epsilon_zero}, the presence of $\epsilon_i>0$ in Theorem \ref{essential_theo} was attributed to the impossibility of bounding the norm of the positive semidefinite multipliers $\{\mu_i\}_i$. If this intuition were accurate, then one would expect that any algebraic bound on the norm of $\mu_i$ would allow one to set $\epsilon_i=0$ in Eq.~\eqref{essential_KKT}. This is exactly what the next theorem shows.

\begin{theo}
\label{theo_null_epsilon}
Consider an NPO problem 
\eqref{nc_prob}
that satisfies the Archimedean condition. For some set of indices $I\subset\{1,\ldots ,m\}$, let there exist $r\in\R^+$ and a tuple of Hermitian polynomials $q\in\P^{n}$ such that 
\beq
\begin{aligned}
&\nabla_xg_i(q)+s_ig_i+g_is_i^*-r\in M(\mathbf{g})+I(\mathbf{h}),\quad\forall i\in I,\\
&\nabla_xg_i(q)+s_ig_i+g_is_i^* \in M(\mathbf{g})+I(\mathbf{h}),\quad\forall i\not\in I,\\
&\nabla_xh_j(q)\in I(\mathbf{h}),\quad j=1,\ldots ,m'.
\label{MF}
\end{aligned}
\eeq
for some polynomials $\{s_i:i=1,\ldots ,m\}$. Then, there exist Hermitian linear functionals $\{\mu_i:\P\to\C\}_{i\in I}$, with $\mu_i(M(\mathbf{g})+I(\mathbf{h}))\geq 0$ for all $i\in I$, such that, for any $k\in\N$ and $\{\epsilon_i:i\not\in I\}\subset \R^+$, there exist Hermitian functionals $\{\mu^k_i:\P_{2k}\to\C\}_{i\not\in I}$, $\{\lambda^k_j:\P_{2k}\to\C\}_j$ satisfying
\begin{equation}\label{weaksential_KKT}
\begin{aligned}
&\mu^k_i(pp^*)+\epsilon_i\|p\|^2_2\geq 0,\quad \forall p\in\P,\; \deg (p)\leq k,\;\forall i\not\in I\\ 
&\mu^k_i(pg_lp^*)+\epsilon_i\|p\|^2_2\geq 0, \quad \forall p\in\P,\; \deg (p)\leq k-\left\lceil\frac{\deg(g_l)}{2}\right\rceil,\;l=1,\ldots ,m,\;\forall i\not\in I\\ 
&\mu^k_i(I(\mathbf{h})_{2k})=0,\quad  \forall i\not\in I,\\ 
&\mu^k_i(sg_i)=\mu^k_i(g_is)=0,\quad \forall s\in\P,\;\deg(s)\leq 2k-\deg(g_i),\; \forall i\not\in I,\\ 
&\lambda_j^k(I(\mathbf{h})_{2k})=0, \quad j=1,\ldots ,m',\\
&\sum_{i\not\in I}\mu^k_i\left(\nabla_x g_i(p(x))\right)+\sum_j\lambda^k_j\left(\nabla_x h_j(p(x))\right)=\sigma^\star\left(\nabla_x f(p(x))\right)-\sum_{i\in I}\mu_i\left(\nabla_x g_i(p(x))\right),\quad \forall p\in \P^{n}, \\ 
&\hspace{1cm}\deg \left(\nabla_x f\right),\deg \left(\nabla_x g_i\right),\deg \left(\nabla_x h_j\right)\leq 2k-\deg(p). 
\end{aligned}
\end{equation}
In particular, the \core ncKKT conditions \eqref{essential_KKT} hold with $\epsilon_i=0$ for all $i\in I$.
\end{theo}
\begin{remark}
\label{remark:bounded_mu}
Note that, if weak ncKKT holds, then one can use Eq.~\eqref{MF} to upper bound the norm of the Lagrange multipliers $\{\mu_i:i\in I\}$. Indeed, take $p=q$ in Eq.~\eqref{operator_optimality}. By Eq.~\eqref{MF} it then holds that
\begin{equation}
\sigma^\star(\nabla_xf(q))=\sum_i\mu_i(\nabla_x g_i(q))=\sum_i\mu_i(\nabla_x g_i(q)+s_ig_i+g_is^*_i)\geq r\sum_i\mu_i(1).
\end{equation}
That is, for $i\in I$, $\mu_i(1)$, the norm of $\mu_i$, is bounded above by $\frac{\sigma^\star(\nabla_xf(q))}{r}$.
\end{remark}

To prove Theorem \ref{theo_null_epsilon}, we will need the following simple lemma.

\begin{lemma}
\label{lemma_linear2}
Consider the system of linear equations in the variables $y\in\R^s$ given by
\begin{equation}
A\cdot y=b,
\label{original_system}
\end{equation}
where $A$ is an $r\times s$ real matrix and $b\in\R^r$. Then, there exists a constant $K\in\R^+$ such that, for any $b'\in\im(A)\subset\R^r$ and any solution $y$ of \eqref{original_system}, there exists a solution $y'$ of the system
\begin{equation}
A\cdot y'=b',
\label{perturbed_system}
\end{equation}
with $\|y-y'\|_2\leq K\|b-b'\|_2$. 
\end{lemma}

\begin{proof}
If system \eqref{original_system} is solvable, then any solution $y$ thereof can be expressed as
\begin{equation}
 y=A^+ \cdot b+\bar{y},   
\end{equation}
where $\bar{y}\in \ker(A)$ and $A^+ $ denotes the Moore-Penrose inverse of $A$. Given $y,b'$, the expression $y'=\bar{y}+A^+ \cdot b'$ is thus a solution of Eq.~\eqref{perturbed_system}, as long as Eq.~\eqref{perturbed_system} is solvable. Moreover,
\begin{equation}
\|y-y'\|_2\leq \|A^+ \|\|b-b'\|_2.
\end{equation}
Defining $K:=\|A^+ \|$, we arrive at the statement of the lemma.
\end{proof}

\begin{proof}[Proof of Theorem \ref{theo_null_epsilon}.]
By Theorem \ref{essential_theo}, for any $k\in\N$ and any $\{\epsilon_i\}_i\subset\R^+$, there exist functionals $\{\mu_i^k\}_i$, $\{\lambda_j^k\}_j$ satisfying Eqs.~\eqref{essential_KKT}. Without loss of generality, let us assume in the following that $\epsilon_i\leq 1$ for all $i$.

We next prove that the existence of a polynomial $q$ satisfying conditions \eqref{MF} implies that we can choose the functionals $\{\mu_i^k:i\in I\}$ such that the terms $\{\mu_i^k(p):p\in\M_{2k}\}$ are bounded. Remember, from the proof of Theorem \ref{essential_theo}, that $\M_j$ is the (finite) set of polynomials of the form $s+s^*$ or $i(s-s^*)$, with $s\in\langle x\rangle_j$.

{First, by virtue of the first three lines of Eq.~\eqref{essential_KKT}, we have that, for any polynomial $p\in M(\mathbf{g})_{2k}+I(\mathbf{h})_{2k}$, with SOS decomposition given by Eq.~\eqref{eq:sos} and
\begin{align}
2\deg(p_l),\,2\deg(p_{il})+\deg(g_i),\,\deg(p^+_{jl})+\deg(p^+_{jl})+\deg(h_j)\leq 2k,    
\end{align}
it holds that
\begin{equation}
\mu^k_i(p)+\epsilon_i\left(\sum_j\|p_j\|^2_2+\sum_l\|p_{il}\|^2_2\right)\geq 0.
\label{SOS_condi}
\end{equation}}
Take $p=q$ (the polynomial tuple in Eq.~\eqref{MF}). Provided that $k$ is large enough so that the SOS decompositions in Eq.~\eqref{MF} can be evaluated by $\{\mu_i^k\}_i$, this implies that there exist constants $\{K_i\}_i\subset \R^+$, independent of $\{\epsilon_i\}_i$, such that
\begin{equation}
\sigma^\star(\nabla_xf(q))+\sum_i\epsilon_iK_i=\sum_i \mu^k_i(\nabla_xg_i(q))+\sum_i\epsilon_iK_i\geq r\sum_{i\in I} \mu^k_i(1).
\label{unit_value}
\end{equation}
Since $\sigma^\star(\nabla_xf(q))$ is bounded (by the Archimedean condition) and $\mu^k_i(1)+1\geq \mu^k_i(1)+\epsilon_i\geq 0$, it follows that the values $\{|\mu^k_i(1)|:i\in I\}$ are upper bounded by a constant $C\in\R^+$. 

Also by the Archimedean condition, for any Hermitian polynomial $p$, there exists $K_p\in\R^+$ such that $K_p\pm p$ is an SOS. Eq.~\eqref{SOS_condi} then implies that, for any $p\in \P$, there exists a constant $C_p$, independent of $\{\epsilon_i\}_i$, such that, for $k$ large enough,
\begin{equation}
K_p\mu^k_i(1)\pm\mu^k_i(p)+\epsilon_iC_p\geq 0.
\end{equation}
If $k$ is also large enough to accommodate the SOS decompositions in Eq.~\eqref{MF}, we conclude by Eq.~\eqref{unit_value} that $\{\mu^k_i(p):i\in I\}$ are also bounded.

Thus, for given $k'\in\N$, we can find $k\geq k'$ and $C\in \R^+$ such that any feasible set of Lagrange multipliers $\mu^k,\lambda^k$ satisfies
\begin{equation}
|\mu^k_i(p)|\leq C,\quad \forall i\in I,
\end{equation}
for all polynomials $p\in\M_{2k'}$. Obviously, the restrictions $\lambda^{k'},\mu^{k'}$ of said Lagrange multipliers $\lambda^k,\mu^k$ to the smaller domain $\P_{2k'}$ also satisfy conditions \eqref{essential_KKT}, with the replacement $k\to k'$. Moreover, by the argument above, for all $p\in\P_{2k'}$, $i\in I$, $\mu_i^{k'}(p)$ is bounded. 

We have just proven that, for any $k\in\N$, one can choose $\{\mu_i^k:i\in I\}$ bounded, independently of the value of $\{\epsilon_i\}_i$. Now, by Theorem \ref{essential_theo}, for each $k\in\N$ there exists feasible $\mu^{k},\lambda^{k}:\P_{2k}\to\C$ satisfying Eq.~\eqref{essential_KKT} for $\epsilon_i=\frac{1}{k}$, with $\{\mu^{k}_i:i\in I\}$ bounded. Consider now the sequence of functionals $(\{\mu^{k}_i:\P\to\C\}_{i\in I})_k$ (we extend $\mu^k_i$ from $\P_{2k}$ to $\P$ by mapping all monomials of degree $>k$ to zero). Now, let $i\in I$. Since, for all monomials $m\in\P$, there exists $L_m$ such that $|\mu^k_i(m)|\leq L_m$ for all $k$, we can construct a linear, invertible transformation $\mathbb{L}$ such that $|\mathbb{L}\circ \mu^k_i(m)|\leq 1$ for all monomials $m\in\P$. By the Banach-Alaoglu theorem \cite[Theorem IV.21]{reedsimon}, the sequence $(\{\tilde{\mu}^k_i\}_{i\in I})_k$, with $\tilde{\mu}_i:=\mathbb{L}\circ\mu^k_i$, thus admits a converging subsequence indexed by $(k_\alpha)_\alpha$, call $\{\tilde{\mu}_i:\P\to\C\}_{i\in I}$ its limit. Finally, we define $\mu_i:=\mathbb{L}^{-1}\circ \tilde{\mu}_i$. By construction, we have that, for all $k\in\N$ and $i\in I$, 
\begin{equation}
\lim_{\alpha\to\infty}\mu^{k_\alpha}_i\Bigr|_{\P_{2k}}=\mu_i\Bigr|_{\P_{2k}}.    
\end{equation}

This implies that, for $i\in I$ and any $\delta>0$,
\begin{equation}
M^k(\mu_i)+\delta\id\geq0,\; M^k_l(\mu_i)+\delta\id\geq0.
\end{equation}
It follows that $M^k(\mu_i),\; M^k_l(\mu_i)\geq 0$, for all $k$, and so $\{\mu_i\}_{i\in I}$ are positive linear functionals compatible with the problem constraints.


Now, fix $k\in\N$. For any $\alpha\in\N$ such that $k\leq k_\alpha$, there exist $\{\mu^{k,\alpha}_i\}_{i\not\in I}$, $\{\lambda^{k,\alpha}_j\}_j$ satisfying the first five lines of Eqs.~\eqref{weaksential_KKT} and such that
\begin{equation}
\sum_{i\not\in I}\mu^{k,\alpha}_i(\nabla_xg_i(p))+\sum_j\lambda^{k,\alpha}_j(\nabla_xh_j(p))=\sigma^\star(\nabla_xf(p))-\sum_{i\in I}\mu^{k_\alpha}_i(\nabla_xg_i(p)),
\end{equation}
for all $p\in \P$ with appropriately constrained degree. Indeed, $\{\mu_i^{k,\alpha}\}_{i\not\in I}$, $\{\lambda_j^{k,\alpha}\}_{j}$ can be chosen to be the restrictions to $\P_{2k}$ of $\{\mu^{k_\alpha}_i\}_{i\not\in I}$, $\{\lambda^{k_\alpha}_j\}_j$. The variables $\{\mu^{k,\alpha}_i:i\not\in I\},\{\lambda^{k,\alpha}_j\}_j$ in the equation above, identified by their evaluations on a finite set of polynomials, can thus be seen as a solution of the real system of linear equations
\begin{align}\label{final_linear}
&\mu^k_l(s^+h_js^-+(s^-)^*h_j(s^+)^*)=\mu^k_l(i(s^+h_js^--(s^-)^*h_j(s^+)^*))=0, \nonumber\\
&\hspace{1cm}\forall s^+, s^-\in\langle x\rangle,\; \deg (s^+)+\deg(s^-)\leq 2k-\deg(h_j), \; \forall j,\forall l\not\in I\nonumber\\ 
&\mu^k_l(sg_l+g_ls^*)=\mu^k_l(i(sg_l-g_ls^*))=0,\quad \forall s\in\langle x\rangle,\;\deg(s)\leq 2k-\deg(g_l),\; l\not\in I,\nonumber\\ 
&\lambda_l^k(s^+h_js^-+(s^-)^*h_j(s^+)^*)=\lambda_l^k(i(s^+h_js^--(s^-)^*h_j(s^+)^*))=0,\nonumber\\
&\hspace{1cm}\forall s^+,s^-\in\M,\; \deg (s^+)+\deg(s^-)\leq 2k-\deg(h_j),\, \forall j,l,\nonumber\\
&\sum_{i\not\in I}\mu^k_i(\nabla_xg_i(p))+\sum_j\lambda^k_j(\nabla_xh_j(p))=c^{\alpha}(p), \nonumber\\
&\hspace{1cm}\forall p\in\M^n, \deg \left(\nabla_x f\right),\deg \left(\nabla_x g_i\right),\deg \left(\nabla_x h_j\right)\leq 2k-\deg(p),
\end{align}
The ``constant vector'' $c^{\alpha}(p)$ is given by
\begin{equation*}
c^{\alpha}(p):=\sigma^\star(\nabla_xf(p))-\sum_{i\in I}\mu^{k_\alpha}_i(\nabla_xg_i(p)).
\end{equation*}
Since the system \eqref{final_linear} is solvable for all $\alpha$, then it is also solvable in the limit $\alpha\to\infty$ (because the image of the corresponding matrix of coefficients, being finite-dimensional, is a closed subspace). By Lemma \ref{lemma_linear2}, there exists $K\in\R^+$ such that, for any $\alpha$, there exists a solution $\{\mu^k_i:i\not\in I\},\{\lambda^k_j\}_j$ of the system
\begin{align}\label{balance_limit}
&\mu^k_l(s^+h_js^-+(s^-)^*h_j(s^+)^*)=\mu^k_l(i(s^+h_js^--(s^-)^*h_j(s^+)^*))=0, \nonumber\\
&\hspace{1cm}\forall s^+, s^-\in\langle x\rangle,\; \deg (s^+)+\deg(s^-)\leq 2k-\deg(h_j), \; \forall j,\forall l\not\in I\\ 
&\mu^k_l(sg_l+g_ls^*)=\mu^k_l(i(sg_l-g_ls^*))=0,\quad \forall s\in\langle x\rangle,\;\deg(s)\leq 2k-\deg(g_l),\; l\not\in I,\\ 
&\lambda_l^k(s^+h_js^-+(s^-)^*h_j(s^+)^*)=\lambda_l^k(i(s^+h_js^--(s^-)^*h_j(s^+)^*))=0,\nonumber\\
&\hspace{1cm}\forall s^+,s^-\in\M,\; \deg (s^+)+\deg(s^-)\leq 2k-\deg(h_j), \forall j,l,\nonumber\\
&\sum_{i\not\in I}\mu^k_i(\nabla_xg_i(p))+\sum_j\lambda^k_j(\nabla_xh_j(p))=\lim_{\alpha\to\infty}c^{\alpha}(p)=\sigma^\star(\nabla_xf(p))-\sum_{i\in I}\mu_i(\nabla_xg_i(p)), \nonumber\\
&\hspace{1cm}\forall p\in\M^n, \deg \left(\nabla_x f\right),\deg \left(\nabla_x g_i\right),\deg \left(\nabla_x h_j\right)\leq 2k-\deg(p),
\end{align}
with 
\begin{align}
\sum_{p\in \M_{2k}}\left( \sum_{i\not\in I}(\mu^k_i(p)-\mu^{k,\alpha}_i(p))^2+\sum_j(\lambda^k_j(p)-\lambda^{k,\alpha}_j(p))^2\right)\leq K^2\|c-c^\alpha\|_2^2.
\end{align}
Thus, there exists a sequence of real positive numbers $(\delta_\alpha)_\alpha$, with $\lim_{\alpha\to\infty}\delta_\alpha=0$, such that
\begin{equation}
\|M^k(\mu^k_i)-M^k(\mu^{k,\alpha}_i)\|,\|M_l^k(\mu^k_i)-M^k_l(\mu_i^{k,\alpha})\|\leq \delta_\alpha,\forall i\not\in I.
\end{equation}
 Since, by construction, 
\begin{equation}
 M^k(\mu^{k,\alpha}_i)+\frac{1}{k_\alpha}\id\geq 0, M_l^k(\mu^{k,\alpha}_i)+\frac{1}{k_\alpha}\id\geq 0,
\end{equation}
we have that 
\begin{equation}
 M^k(\mu^{k}_i)+\left(\frac{1}{k_\alpha}+\delta_\alpha\right)\id\geq 0, M_l^k(\mu^{k}_i)+\left(\frac{1}{k_\alpha}+\delta_\alpha\right)\id\geq 0,\forall i\not\in I.
\end{equation}
The quantity $\frac{1}{k_\alpha}+\delta_\alpha$ tends to zero as $\alpha$ tends to infinity; it follows that, for any $\epsilon \in \R^+$ and for all $k\in\N$, there exist functionals $\{\mu^k_i:i\not\in I\},\{\lambda_j^k\}_j$ such that Eq.~\eqref{weaksential_KKT} is satisfied for $\epsilon_i=\epsilon$, for all $i\not\in I$, with $\{\mu_i\}_{i\in I}$ being positive linear functionals compatible with the problem constraints. This implies that Eq.~\eqref{weaksential_KKT} is satisfied for any $\{\epsilon_i\}_{i\not\in I}\subset \R^+$, as long as $\epsilon_i\geq \epsilon$ for all $i\not\in I$. Since $\epsilon \in \R^+$ is arbitrary, the theorem has been proven.
\end{proof}

\subsection{Normed KKT conditions}
\label{sec:normed}
In this section, we will investigate tractable conditions under which the normed ncKKT conditions from Definition \ref{def:normed_ncKKT} hold. To find them, we will adapt known sufficient criteria for the classical case. 

In Section \ref{sec:MF}, we present an algebraic, non-commutative version of the Mangasarian--Fromovitz Constraint Qualification (MFCQ), see Eqs.~\eqref{MF_classical_prelud}, \eqref{MF_classical} and prove that it suffices to guarantee that the corresponding NPO problem satisfies the normed ncKKT conditions. Before we proceed, though, we need to formulate an algebraic notion of gradient linear independence for sets of non-commutative polynomial equality constraints. This is the subject of the next section.

\subsubsection{Linearly independent gradients}
\label{sec:linear_indep}
Given a number of non-commutative polynomial equality constraints $\{h_j(x)=0\}$, we need to find a meaning for the expression ``their gradients are linearly independent''. Our starting point is the classical meaning of the term.

Let $\{x_i\}_i$ be commuting variables. Then the gradient vectors $\{\partial_x h_j\}_j$ are independent iff there exist vectors $v_1(x),\ldots ,v_{m'}(x)$ such that 
\begin{align}
&\braket{\partial_x h_j}{v_{k}(x)}=\delta_{j,k},\quad j,k=1,\ldots ,m'.
\end{align}
Now, define the matrices
\begin{equation}
\begin{aligned}
&\hat{P}_{j}(x):=\ket{v_{j}(x)}\bra{\partial_xh_j},\quad j=1,\ldots ,m',\\ 
&\hat{P}_{0}(x):=\id-\sum_{k=1}^{m'}\hat{P}_k(x).
\end{aligned}
\end{equation}
It is easy to see that, for all $z\in\R^n$, these matrices satisfy
\begin{equation}\label{LI_classical}
\begin{aligned}
&\sum_k\hat{P}_k(x)\ket{z}=\ket{z},\\ 
&\bra{\partial_x h_j} \hat{P}_k(x)\ket{z}=0,\quad \forall j\not=k,\\ 
&\braket{\partial_x h_j}{z}=0\;\rightarrow\; \hat{P}_{j}(x)\ket{z}=0.
\end{aligned}
\end{equation}
Note that, in general classical polynomial optimization problems, the matrices $\{\hat{P}_{j}(x)\}_j$ are rational functions of $x$. In the following lines, though, we will assume that they are actually polynomials. This algebraic constraint will restrict our definition of gradient linear independence more than usual, but otherwise will allow us to generalize this notion to non-commutative systems.

In classical systems, variables form a vector $x=(x_1,\ldots ,x_n)$ of scalars, and the gradient $\partial_x h$ of a function $h$ is also an $n$-dimensional vector of scalars. To find out how $h$ will change if we move the variables in some direction $z$, we compute the scalar product $\partial_x h\cdot z$, thus obtaining a scalar. 

In non-commutative systems, variables form a vector $x=(x_1,\ldots ,x_n)$ of non-commuting objects, and the gradient $\nabla_x h(\bullet)$ of a polynomial $h(x)$ can be regarded as a linear map from $n$-tuples of polynomials $p=(p_1,\ldots ,p_n)\in \C\langle x\rangle^n$ to a single polynomial $\nabla_xh(p(x))$.
A non-commutative analog of relations \eqref{LI_classical} would thus demand the existence of $m'+1$ $n$-tuples of polynomials $P_0(x,z),\ldots ,P_{m'}(x,z)$ in the variables $x=(x_1,\ldots ,x_n)$, $z=(z_1,\ldots ,z_n)$, linear on $z$. Each such $n$-tuple of polynomials $P_k(x,z)$ would play the role that the vector $\hat{P}_k(x)\ket{z}$ played in relations \eqref{LI_classical}. Correspondingly, the tuples $P_0,\ldots ,P_{m'}$ should satisfy the following conditions:
\begin{subequations}
\begin{align}
&\sum_kP_k(X,Z)=Z,\label{completeness}\\
&\nabla_x h_j(P_{k}(X,Z))=0,\quad \forall j\not=k,\label{differentiation}\\
&\nabla_x h_j(Z)=0\;\rightarrow\; P_{j}(X,Z)=0.
\label{charact}
\end{align}
\end{subequations}
We will regard the algebraic version of constraints \eqref{completeness}--\eqref{charact} as the non-commutative generalization of gradient linear independence.

\begin{defin}
We say that a set of equality constraints $\{h_j(x)= 0:j=1,\ldots ,m'\}$ has linearly independent gradients if there exist $n$-tuples of symmetric polynomials in $2n$ variables $P_0(x,z),P_1(x,z),\ldots ,P_{m'}(x,z)$, linear in the $z$ variables, such that
\begin{subequations}\label{alg_LICQ}
\begin{align}
&\sum_{j=0}^{m'}P_j(x,p)-p\in {I(\mathbf{h})}^{n},\quad\forall p\in\P^{n},\label{completeness_LI}\\ 
&\nabla_xh_j(P_{k}(x,p))\in I(\mathbf{h}),\quad\forall k\not=j,\;p\in\P^{n},\label{anni_LI}\\ 
&\exists \beta^+,\beta^-: (P_{j}(x,p))_k-\sum_{l}\left(\beta^+_{jkl}(x)\nabla_xh_j(p)(\beta^-_{jkl}(x))^*+\beta^-_{jkl}(x)\nabla_xh_j(p)(\beta^+_{jkl}(x))^*\right)\in I(\mathbf{h}),\nonumber\\
&\phantom{\exists \beta^+,\beta^-: }\ j=1,\ldots ,m', \quad k =1,\ldots ,n, \quad p\in\P^{n}.\label{charact_LI}
\end{align}
\end{subequations}
Here $(P)_k$ denotes the $k^{th}$ component of the $n$-tuple $P$.
\end{defin}
\begin{remark}
One can replace Eq.~\eqref{charact_LI} by the weaker\footnote{That eq. (\ref{charact_LI}) implies eq. (\ref{weaker_version}) can be seen by substituting $P_j(x,p)$ by minus the second term of eq. (\ref{charact_LI}) in $\nabla_xg_i(P_j(x,p)), \nabla_xf(P_j(x,p))$. Due to the linearity of the gradient, the result will be an expression like minus the second terms of both lines of eq. (\ref{weaker_version}), plus some polynomial in $I(\mathbf{h})$.}, problem-dependent constraints
\begin{align}
&\exists \beta^+,\beta^-,\gamma^+,\gamma^-,s,\text{  such that  }\nonumber\\
&\nabla_xg_i\left(P_{j}(x,p)\right)-\sum_{l}\left(\beta^+_{ijl}(x)\nabla_xh_j(p)(\beta^-_{ijl}(x))^*+\beta^-_{ijl}(x)\nabla_xh_j(p)(\beta^+_{ijl}(x))^*\right)+s_i(p)g_i+g_is_i(p)\in I(\mathbf{h}),\nonumber\\
&\nabla_xf\left(P_{j}(x,p)\right)-\sum_{l}\left(\gamma^+_{jl}(x)\nabla_xh_j(p)(\gamma^-_{jl}(x))^*+\gamma^-_{jl}(x)\nabla_xh_j(p)(\gamma^+_{jl}(x))^*\right)\in I(\mathbf{h}),\nonumber\\
&\phantom{\exists \beta^+,\beta^-,\gamma^+,\gamma^-,s,}\ 
j=1,\ldots ,m', \quad k =1,\ldots ,n, \quad p\in\P^{n}.
\label{weaker_version}
\end{align}
Indeed, the proofs of Lemmas \ref{bounded_interm}, \ref{lemma_MF_bounded} and Theorem \ref{theo_MF} below follow through with such a modified definition. 
\end{remark}

Unless the quotient space $\P/I(\h)$ is finite-dimensional, verifying that conditions \eqref{alg_LICQ} hold for all $p\in\P^{n}$ is a challenging endeavor. The next proposition provides a practical, sufficient condition to guarantee gradient linear independence.

\begin{prop}
Let $\{h_j\}_j$ be a set of Hermitian polynomials, and let $\{r_k\}_k\subset \P$ be such that $[r_k,x_i]\in I(\mathbf{h})$ for all $i,k$. Suppose that there exist $n$-tuples of symmetric polynomials in $2n$ variables $P_0(x,z),P_1(x,z),\ldots ,P_{m'}(x,z)$, linear in the $z$ variables, such that
\begin{subequations}\label{alg_LICQ_crit}
\begin{align}
&\sum_{j=0}^{m'}P_j(x,z)-z\in ({\Pnull})^{n},\\ 
&\nabla_xh_j(P_{k}(x,z))\in \Pnull,\quad\forall k\not=j,\\ 
&\exists \beta^+,\beta^-: (P_{j}(x,z))_k-\sum_{l}\beta^+_{jkl}(x)\nabla_xh_j(z)(\beta^-_{jkl}(x))^*-\beta^-_{jkl}(x)\nabla_xh_j(z)(\beta^+_{jkl}(x))^*\in \Pnull,\nonumber\\
&\phantom{\exists \beta^+,\beta^-: }\ j=1,\ldots ,m', \quad k =1,\ldots ,n,
\end{align}
\end{subequations}
where $\Pnull$ is the set of polynomials $q(x,z)$, linear in $z=(z_1,\ldots ,z_n)$, of the form
\begin{equation}
q(x,z)=\sum_{j,l}p^+_{jl}(x,z)h_j(x)p^-_{jl}(x,z)+\sum_{k,l,i}q_{ikl}^+(x)[r_k(x),z_i]q_{ikl}^-(x).
\label{def_p_null}
\end{equation}
Then, the set of equality constraints $\{h_j(x)=0\}_j$ has linearly independent gradients.
\end{prop}
The proof is obvious.

\begin{eg}
Let $\{h_j(x)=0\}_{j=1}^{m'}$ be of the form
\begin{equation}
h_j(x)=\sum_k\nu_{jk}x_k+b_j,    
\end{equation}
with $\{\nu_{jk}\}_{j,k}\cup\{b_j\}\subset \R$. If the matrix $\nu$ has linearly independent rows, then the gradients of $\{h_j\}_j$ are linearly independent. Indeed, let $\nu$ have rank $m'$. Then, there exist vectors $\{v^j\}_j\subset \R^{n}$ such that $\sum_kv^j_k\cdot\nu_{lk}=\delta_{jl}$. 
Define thus
\begin{equation}
\begin{split}
(P_j(x,z))_k&:= v^j_k\sum_l\nu_{jl}z_l,\quad j=1,\ldots ,m',\\
P_0(x,z)&:=z-\sum_jP_j(x,z).    
\end{split}
\end{equation}
The newly defined $\{P_j\}_j$ satisfy Eqs.~\eqref{alg_LICQ_crit}. That is: for linear constraints, the commutative and non-commutative definitions of gradient linear independence coincide.
\end{eg}

\begin{eg}
\label{example_dichotomic}
Consider the equality constraints $\{h_j(x)=0\}_{j=1}^{n}$ for
\begin{equation}
h_j(x):=x_j^2-1,\quad j=1,\ldots ,n.
\end{equation}
This system also satisfies gradient linear independence. Indeed, define
\begin{equation}
\begin{aligned}
&(P_j(x,z))_k:=\delta_{jk}\frac{1}{2}(z_j+x_jz_jx_j),\\ 
&(P_0(x,z))_k:=\frac{1}{2}(z_k-x_kz_kx_k).
\end{aligned}
\end{equation}
Then it can be easily verified that the $n$-tuples of polynomials $P_0,\ldots ,P_n$ satisfy conditions \eqref{alg_LICQ_crit}.
\end{eg}

The linear independence constraints are very restrictive: if a relaxed form of the weak ncKKT conditions holds, then so does normed ncKKT. This is proved in the following lemma.

\begin{lemma}
\label{bounded_interm}
Let the constraints $\{g_i(x)\geq 0\}_i\cup\{h_j(x)=0\}_j$ be Archimedean, and let there exist bounded positive functionals $\sigma$, $\{\mu_i\}_i$ satisfying $\sigma(M(\mathbf{g})+I(\h)), \mu_i(M(\mathbf{g})+I(\h))\geq 0$, for all $i$, with $\sigma(1),\,\mu_i(1)\leq K$, for some $K\in\R^+,\, K>1$ and such that, for all $k\in\N$, the system of linear equations
\beq
\begin{aligned}
&\lambda^k_j(I(\mathbf{h})_{2k})=0,\quad j=1,\ldots ,m',\\
&\sum_j\lambda^k_j(\nabla_xh_j(p(x)))=\sigma(\nabla_xf(p(x)))-\sum_i\mu_i(\nabla_xg_i(p(x)),\quad \forall p\in\P^{n},\\
&\hspace{1cm}\deg(t)\leq 2k+1-\deg(f)+1,\;2k+1-\deg(g_i),\quad i=1,\ldots ,m,
\label{linear_system_opt}
\end{aligned}
\eeq
has a solution $\lambda^k$. Furthermore, let the gradients of the constraints $\{h_j=0\}_j$ be linearly independent. Then, the normed ncKKT conditions hold, with $\|\lambda^{\pm}\|_{\SOS}\leq O(K)$.
\end{lemma}
\begin{proof}
By Eq.~\eqref{anni_LI}, if we set $p=P_0(x,p')$ in Eq.~\eqref{linear_system_opt}, we arrive at
\begin{equation}
\sigma(\nabla_xf(P_0(x,p')))-\sum_i\mu_i(\nabla_xg_i(P_0(x,p')))=0,
\label{suff_P_0}
\end{equation}
provided that $p'$ has sufficiently small degree (so that all relevant intermediate polynomials have degree equal to or below $2k$). Taking the limit $k\to\infty$, it follows that the relation above holds for all $p'\in\P$.

Now, consider the linear functional $\lambda_j:\P\to\C$, defined by
\beq
\begin{split}
\lambda_j(p)& :=\sigma\left(\nabla_xf\left(\left(\sum_l\beta^+_{jkl}p(\beta^{-}_{jkl})^*+\beta^-_{jkl}p(\beta^{+}_{jkl})^*\right)_k\right)\right)\\
&\phantom{:=}-\sum_i \mu_i\left(\nabla_xg_i\left(\left(\sum_l\beta^+_{jkl}p(\beta^{-}_{jkl})^*+\beta^-_{jkl}p(\beta^{+}_{jkl})^*\right)_k\right)\right),
\label{def_lambdas}
\end{split}
\eeq
where the polynomials $\{\beta^{\pm}_{jkl}\}_{jkl}$ are the ones appearing in Eq.~\eqref{charact_LI}.

Note that, for $r,t\in\P$ and any positive functional $\omega$, the functional
\begin{equation}
\lambda(p):=\omega(rpt^*+ tpr^*)
\end{equation}
equals the difference between two positive functionals, namely,
\begin{equation}
\lambda(p)=\omega^+(p)-\omega^-(p),
\end{equation}
for
\begin{equation}
\omega^{\pm}(p):=\frac{1}{2}\omega((r\pm t)p(r\pm t)^*).
\end{equation}
Now, we expand the arguments of the functionals on the right-hand side of \eqref{def_lambdas}, as polynomials of the form $\sum_is_ipt^*_i+t_ips^*_i$. Since $\sigma$, 
$\{\mu_i\}$ are bounded positive functionals, we find that the right-hand side of Eq.~\eqref{def_lambdas} can be decomposed as a finite sum of differences of positive functionals. Thus,
\begin{equation}
\lambda_j=\lambda_j^+-\lambda_j^-,    
\end{equation}
where the positive functionals $\lambda_j^{\pm}$ inherit from $\sigma$, $\{\mu_i\}_i$ the property of being non-negative on $M(\mathbf{g})+I(\mathbf{h})$. The values $\lambda_j^\pm(1)$ correspond to expressions of the form
\begin{equation}
\lambda^{\pm}_j(1)=\sigma(s^\pm)+\sum_i \mu_i(s^\pm_i),
\end{equation}
where $s^\pm, s_i^\pm$ are SOS. By the Archimedean condition, there exists $\nu\in\R^+$ such that the polynomials $\{\nu-s^{\pm}\}\cup\{\nu-s^{\pm}_i\}_i$ are SOS. Together with the constraints $\sigma(1),\mu_i(1)\leq K$, it follows that $\lambda^{\pm}(1)\leq K(m+1)\nu$.

It remains to be seen that the newly defined $\lambda$'s satisfy condition \eqref{normed_operator_optimality}. Taking $p=P_0(x,p')$ in \eqref{normed_operator_optimality} and invoking Eq.~\eqref{anni_LI} and Eq.~\eqref{suff_P_0} we have that Eq.~\eqref{normed_operator_optimality} is satisfied as long as $p$ is of the form $p=P_0(x,p')$. 

Next, set $p=P_j(x,p')$ in Eq.~\eqref{normed_operator_optimality}, for some $j\in\{1,\ldots ,m'\}$, $p'\in \P^{n}$. By Eq.~\eqref{anni_LI}, all terms of the form $\lambda_k(\nabla_xh_k(P_j(x,p')))$ with $k\not=j$, drop. Furthermore, by Eqs.~\eqref{completeness_LI}, \eqref{anni_LI}, we have that
\begin{equation}
\lambda_j\left(\nabla_xh_j(p')\right)=\sum_k\lambda_j\left(\nabla_xh_j(P_k(x,p'))\right)=\lambda_j\left(\nabla_xh_j(P_j(x,p'))\right).
\end{equation}
It follows that
\begin{align}
\lambda_j(\nabla_xh_j(P_j(x,p')))&=\lambda_j(\nabla_xh_j(p'))\nonumber\\
&=\sigma\left(\nabla_xf\left(\sum_l\beta^+_{jkl}\nabla_xh_j(p')\beta^{-}_{jkl}\right)\right)-\sum_i \mu_i\left(\nabla_xg_i\left(\sum_l\beta^+_{jkl}\nabla_xh_j(p')\beta^{-}_{jkl}\right)\right)\nonumber\\
&=\sigma(\nabla_xf(P_j(x,p')))-\sum_i\mu_i(\nabla_xg_i(P_j(x,p'))),
 \end{align}
where the last equality is a consequence of Eq.~\eqref{charact_LI}. Hence, Eq.~\eqref{normed_operator_optimality} holds for $p=P_j(x,p')$, for $j=0,\ldots,m'$.

By Eq.~\eqref{completeness_LI}, any tuple of polynomials $p'$ can be expressed as a sum of terms of the form $\{P_k(x,p')\}_k$ (modulo elements of $I(\mathbf{h})$), and so Eq.~\eqref{normed_operator_optimality} holds in general.
\end{proof}

\subsubsection{Non-commutative Mangasarian--Fromovitz Constraint Qualification}
\label{sec:MF}
We are ready to present a non-commutative, algebraic version of the Mangasarian--Fromovitz constraint qualification.
\begin{defin}[ncMFCQ]
Consider an NPO Problem \eqref{nc_prob}. We say that the problem satisfies {\em non-commutative Mangasarian--Fromovitz constraint qualification (ncMFCQ)} if, on one hand, the equality constraints have linearly independent gradients and, on the other hand, for $i=1,\ldots ,m$, there exist $r\in\R^+$ and a tuple of Hermitian polynomials $q\in\P^{n}$ such that 
\begin{subequations}
\begin{align}
&\nabla_xg_i(q)+s_ig_i+g_is_i^*-r\in M(\mathbf{g})+I(\mathbf{h}),\quad i=1,\ldots ,m,\label{MF_cons}\\
&\nabla_xh_j(q)\in I(\mathbf{h}),\quad j=1,\ldots ,m'\label{anni_MF},
\end{align}
\end{subequations}
for some polynomials $\{s_i:i=1,\ldots ,m\}\subset \P$.
\end{defin}

\begin{remark}
Note the similarities between ncMFCQ and its classical analog MFCQ (Eqs.~\eqref{MF_classical_prelud}, \eqref{MF_classical}). For starters, both qualifications demand linear independence of the gradients of the equality constraints. In addition, Eq.~\eqref{anni_MF} is clearly the algebraic analog of the second line of Eq.~\eqref{MF_classical}. It remains to justify Eq.~\eqref{MF_cons}. As we argue in Appendix \ref{app:heuristic}, the active constraints associated to the inequality $g_i(x)\geq 0$ are $\{\omega(g_i(X))\geq 0:\omega\geq0,\omega(g_i(X^\star))=0\}$. Now, let $\omega$ be a 
state satisfying $\omega(g_i(X^\star))=0$ and let Eq.~\eqref{MF_cons} hold. Then, we have that
\begin{equation}
\frac{d}{dt}\omega(g_i(X^\star+t q(X^\star)))\Big|_{t=0}=\omega(\nabla_xg_i(q(X^\star)))\geq \omega(-s_ig_i-g_is_i^*+r)=r>0.
\end{equation}
In other words, small variations of the optimal operators $X^\star$ in the $q(X^\star)$ direction are positive. Condition \eqref{MF_cons} is therefore the algebraic, non-commutative analog of the first line in \eqref{MF_classical}. Notice that, contrarily to the literature on classical polynomial optimization, we demand the conditions to hold for all feasible points, not just the optimal ones.
\end{remark}

\begin{eg}
\label{ex:normed_family}
Let $x=(y_1,\ldots ,y_c,z_1,\ldots ,z_d)$, and let $\{h_j(y)=0\}_j$ be a set of equality constraints with linearly independent gradients. Let the remaining constraints be inequalities of the form:
\begin{equation}
g_i(x)=1+\tilde{g}_i(x)\geq 0,i=1,\ldots ,m,
\end{equation}
where $\tilde{g}_i(x)$ is homogeneous in $z=(z_1,\ldots ,z_d)$ with degree $r_i>0$. 

The constraints $\{h_j(y)=0\}_j\cup\{g_i(x)\geq 0\}_i$ satisfy ncMFCQ. Indeed, choose $(q(x))_k=0$, for $k=1,\ldots ,c$, and $(q(x))_{k+c}=-z_k$, for $k=1,\ldots ,d$. Then, on one hand, $\nabla_xh_j(q)=0$, for all $j$, i.e., Eq.~\eqref{anni_MF} holds. On the other hand, $\nabla_xg_i(x)(q)=-r_i\tilde{g}_i(x)$, and so
\begin{equation}
r_ig_i(x)+\nabla_xg_i(x)(q)-r_i=0.
\end{equation}
That is, Eq.~\eqref{MF_cons} holds for $r:=\min_i r_i$. Since, by assumption, the equality constraints have linearly independent gradients, the three conditions defining ncMFCQ are satisfied.
\end{eg}

The next theorem is the non-commutative generalization of a celebrated result in classical optimization, which states the validity of the KKT conditions under Mangasarian--Fromovitz constraint qualification \cite{Nocedal2006}.

\begin{theo}
\label{theo_MF}
Consider an NPO Problem \eqref{nc_prob} that satisfies both the Archimedean condition and ncMFCQ. Then, the problem satisfies the normed ncKKT conditions.
\end{theo}
\begin{proof}
Eq.~\eqref{MF_cons} and \eqref{anni_MF} represent the conditions of Theorem \ref{theo_null_epsilon} for $I=\{1,\ldots ,m\}$. Thus, there exist positive linear functionals $\{\mu_i\}_{i=1}^m$, compatible with the constraints of the problem $\{g_i(x)\geq0\}_i\cup\{h_j(x)=0\}_j$, such that, for all $k$, the system of linear equations
\beq
\begin{aligned}
&\lambda^k_j(s^+h_ks^-)=0,\quad j,k=1,\ldots ,m',\;\deg(s^+h_ks^-)\leq 2k,\\
&\sum_j\lambda^k_j(\nabla_xh_j(p(x)))=\sigma^\star(\nabla_xf(p(x)))-\sum_i\mu_i(\nabla_xg_i(p(x)),\quad \forall p\in\P^{n},
\end{aligned}
\eeq
where the degree of $p$ is appropriately bounded, has a solution $\lambda^k$. The conditions of Lemma \ref{bounded_interm} are, therefore, met. Hence, the problem satisfies the normed ncKKT conditions.
\end{proof}

Like its classical counterpart, ncMFCQ allows one to bound the norm of the Lagrange multipliers $\{\mu_i\}_i$, $\{\lambda_j\}_j$. This implies that each of the SDP relaxations of the normed KKT conditions corresponds to a bounded optimization problem. 

\begin{lemma}
\label{lemma_MF_bounded}
Assume that Problem \eqref{nc_prob} is Archimedean and satisfies ncMFCQ (and thus normed ncKKT). Then we can, without loss of generality, bound the state multipliers $\{\mu_i\}_i$ and $\{\lambda^{\pm}_j\}_j$ in Eq.~\eqref{eq:normed_ncKKT}. That is, we can find $K\in\R^+$ such that
\begin{align}
&\mu_i(1),\lambda^{\pm}_j(1)\leq K,\quad \forall i,j.
 \end{align}
\end{lemma}

\begin{proof}
{By Remark \ref{remark:bounded_mu}, we have that $\sum_i\mu_i(1)\leq \frac{1}{r}\sigma^\star(\nabla_xf(q))$. In turn, provided that the original NPO problem is Archimedean, there exists $\eta\in \R^+$ such that 
\begin{equation}
\eta-\nabla_xf(q)
\end{equation}
admits an SOS decomposition. Hence, for each $i$, $\mu_i(1)\leq \frac{\eta}{r}$.

The set of linear functionals $\{\lambda_j\}_j$, restricted to the set of polynomials of degree $2k$ or smaller, is a solution of Eq.~\eqref{linear_system_opt}. The conditions of Lemma \ref{bounded_interm} are met, and thus there exist bounded linear functionals $\{\tilde{\lambda}^{\pm}_j:\P\to\C\}_j$, with norms bounded by $O(\frac{\eta}{r})$, compatible with the problem constraints and satisfying Eq.~\eqref{normed_operator_optimality}.}
\end{proof}

\subsection{Strong KKT conditions}
\label{sec:strong}

It is observed in practice that many natural NPO problems admit an exact SOS resolution, see, e.g., \cite{PNA2010}. Namely, for some $k\in\N$, the $k^{th}$ level of the hierarchy of SDP relaxations \eqref{k_dual} has a (feasible) maximizer achieving the exact solution of the problem. The next theorem shows that, in any such scenario, strong ncKKT holds.

\begin{theo}\label{thm:k->strong}
Consider the NPO Problem \eqref{nc_prob_hilbert}, and let $f-p^\star \in M(\g)+I(\h)$. Then, Problem \eqref{nc_prob_hilbert} satisfies strong ncKKT.

\end{theo}

\begin{proof}
Let $(\H^\star, X^\star,\psi^\star)$ be a  minimizer of Problem \eqref{nc_prob_hilbert}. For simplicity, in the following we use its abstract, functional form $\sigma^\star:\P\to\C$, namely:
\begin{equation}
\sigma^\star(p)=\psi^\star(p(X^\star)).
\label{sigma2psi}
\end{equation}

By the hypotheses of the theorem, there exist polynomials $s_l, s_{il}, s^+_{jl},s^-_{jl}$ such that
\begin{equation}
f-p^\star =\sum_{l}s_ls_l^*+\sum_{i,l}s_{il}g_is_{il}^* +\sum_{j,l}s^+_{jl}h_j(s^-_{jl})^*+s^-_{jl}h_j(s^+_{jl})^*.
\label{SOS2}
\end{equation}
In addition,  
\begin{equation}
\sigma^\star(f-p^\star)=0.
\label{vanishing_obj}
\end{equation}
It follows that
\begin{equation}
\begin{aligned}
\sigma^\star(s_ls_l^*)&=0,\quad \forall l,\\ 
\sigma^\star(s_{il}g_is_{il}^*)&=0,\quad\forall i,l.
\label{vanishing_averages_easy}
\end{aligned}
\end{equation}
These relations, in turn, imply that, for any $q\in \P$, 
\begin{subequations}\label{vanishing_averages}
\begin{align}
&\sigma^\star(s_lq)=\sigma^\star(qs_l^*)=0,\quad\forall l,\\
&\sigma^\star(s_{il}g_iq)=\sigma^\star(qg_is_{il}^*)=0,\quad\forall i,l.
\end{align}
\end{subequations}
Indeed, the first relation follows from the Cauchy-Schwarz inequality or the positive semidefiniteness of the $2\times 2$ matrix 
\begin{equation}
\left(\begin{array}{cc}\sigma^\star(q q^*)&\sigma^\star(q s_{l}^*)\\ \sigma^\star(s_{l} q^*)&\sigma^\star(s_{l} s_{l}^*)\end{array}\right).    
\end{equation}
The second one, from the positive semidefiniteness of 
\begin{equation}
\left(\begin{array}{cc}\sigma^\star(q g_iq^*)&\sigma^\star(q g_i s_{il}^*)\\ \sigma^\star(s_{il} g_i q^*)&\sigma^\star(s_{ij} g_is_{ij}^*)\end{array}\right).    
\end{equation}

Now, for $\delta\in \R$ and an arbitrary vector of Hermitian polynomials $p=(p_i)_{i=1}^n$, let us define a new state through the relation $\sigma^\delta(a):=\sigma^\star(\pi^\delta(a))$, where $\pi^\delta:\P\to\P$ is the homomorphism given by $\pi^\delta(x_i)=x_i+\delta \cdot p_i(x)$. This linear functional $\sigma^\delta$ is indeed a state, since $\sigma^\delta(qq^*)\geq 0$ for all $q\in\P$. However, it does not necessarily satisfy feasibility conditions of the form $\sigma^\delta(qg_iq^*)\geq 0$, $\sigma^\delta(s^+h_js^-)= 0$.

We apply the state $\sigma^\delta$ on both sides of Eq.~\eqref{SOS2}. Taking into account Eqs.~\eqref{vanishing_obj}, \eqref{vanishing_averages}, and the chain rule of differentiation, the result is
\begin{equation}
\delta \sigma^\star(\nabla_x f(p(x)))+O(\delta^2)=\delta\sum_i\mu_i(\nabla_x(g_i(p(x)))+\delta\sum_j\lambda_j(\nabla_x(h_j(p(x)))+O(\delta^2),
\label{intermediario}    
\end{equation}
where $\mu_i$ denotes the positive linear functional given by 
\begin{equation}
\mu_i(q):=\sigma^\star(\sum_ls_{il}qs_{il}^*),
\label{new_states}
\end{equation}
and $\lambda_j$ is the linear functional
\begin{equation}
\lambda_j(q):=\sigma^\star(\sum_ls^+_{jl}q(s_{jl}^-)^*+\sum_ls^-_{jl}q(s_{jl}^+)^*).
\end{equation}
Note that $\lambda_j$ can be expressed as the difference of two positive functionals $\lambda^{\pm}_j$, namely:
\begin{equation}
\lambda^{\pm}_j(q):=\frac{1}{2}\sum_l\sigma^\star((s^+_{jl}\pm s^-_{jl})q(s^+_{jl}\pm s^-_{jl})^*).
\end{equation}

Collecting the terms in Eq.~\eqref{intermediario} that depend linearly on $\delta$, we have that
\begin{equation}
\sigma^\star(\nabla_x f(p(x)))=\sum_i\mu_i(\nabla_x(g_i(p(x)))+\sum_j\lambda_j(\nabla_x(h_j(p(x))).
\label{KKT_pos}
\end{equation}
This is condition \eqref{strong_operator_optimality}.

Finally, the positive linear functionals $\{\mu_i\}_i$ satisfy complementary slackness \eqref{strong_comp_slackness}, for
\begin{equation}
\mu_i(g_i)=\sum_l\sigma^\star(s_{il}g_is_{il}^*)=0,\quad\forall i,
\label{KKT_pos2}
\end{equation}
by Eq.~\eqref{vanishing_averages_easy}.

It only rests to show that $\{\mu_i\}_i$ and $\{\lambda^{\pm}_j\}_j$ can be expressed as 
\begin{equation}
\mu_i(p)=\tilde{\mu}_i(p(X^\star)),\; \lambda_j^{\pm}(p)=\tilde{\lambda}_j^{\pm}(p(X^\star)),
\end{equation}
for some functionals $\{\tilde{\mu}_i:{\A}(X^\star)\to\C\}_i$, $\{\tilde{\lambda}_j^{\pm}:{\A}(X^\star)\to\C\}_j$. This last bit follows from Eq.~\eqref{sigma2psi} and the fact that both sets of functionals are defined in terms of $\sigma^\star$.
\end{proof}

The reader might question the practical use of Theorem \ref{thm:k->strong}. How can one know in advance that a given NPO problem will admit an exact SOS resolution? To answer this question, we need to examine the relation between positive, non-negative and SOS polynomials. 

Given a set of constraints $\{g_i(x)\geq 0\}_i\cup\{h_j(x)=0\}_j$, a Hermitian polynomial $r$ is said to be positive (non-negative) if $r(X)>0$ ($r(X)\geq 0$), for all tuples of operators $X$ satisfying the problem constraints. Think of the polynomial $f-p^\star$, where $p^\star$ is the solution of \eqref{nc_prob_hilbert}. This polynomial is non-negative, but not positive. 

As we pointed out in Section \ref{sec:NPOprob}, if $r$ is SOS, then it is also non-negative. The converse statement, however, does not hold: for some Archimedean constraints, there exist non-negative polynomials $r$ that are not SOS. It is thus not a surprise that some instances of Problem \eqref{nc_prob} do not admit an SOS resolution.


However, some sets of non-commuting constraints $\{g_i(x)\geq 0\}_i\cup\{h_j(x)=0\}_j$ have the property that any non-negative Hermitian polynomial is SOS. Such sets of constraints are said to generate an \emph{Archimedean closed} quadratic module (or set of SOS polynomials) \cite{ozawa2013connes}. In the following, we provide two families of constraints that are known to generate Archimedean closed quadratic modules.

\subsubsection{Equality constraints with a faithful finite-dimensional $*$-representation}
\begin{prop}
\label{prop:finite_SOS}
Consider a set of equality constraints $\{h_j(x)=0\}_j$ such that, for some finite-dimensional Hilbert space $\H^\star$, there exist a $*$-homomorphism $\pi:\P\to B(\H^\star)$, with $\ker(\pi)=I(\mathbf{h})$. Then, the quadratic module generated by $\{h_j(x)=0\}_j$ is Archimedean closed.
\end{prop}
\begin{proof}
Define $X^\star_k:=\pi(x_k)$, denote by $C^*(X^\star)$ the unital $C^*$-algebra generated by $X_1^\star,\ldots ,X_n^\star$ and let $p\in\P$ be an arbitrary non-negative polynomial. Since $\pi(p)=p(X^\star)\in\A(X^\star)$ is a non-negative operator, then one can define its square root $p(X^\star)^{1/2}$. Due to the finite dimensionality of $\A(X^\star)$, there exists a polynomial $s\in\P$ such that $p(X^\star)^{1/2}=s(X^\star)$. We thus have that
\begin{equation}
p(X^\star)-s(X^\star)s(X^\star)^*=0.
\end{equation}
Hence, $\pi(p-ss^*)=0$, and so $p-ss^*\in I(\mathbf{h})$. We have just shown that $p$ is SOS.
\end{proof}
\begin{remark}
A set of constraints satisfying the conditions of Proposition \ref{prop:finite_SOS} is the Pauli algebra \eqref{Pauli}, \eqref{Pauli_comm} used to model many-body quantum systems, see Section \ref{sec:many_body}.
\end{remark}

\begin{cor}
\label{cor:finite_SOS}
Let $\{h_j(x)=0\}_j$ be a set of constraints satisfying the conditions of Proposition \ref{prop:finite_SOS}. Then, any problem of the form \eqref{nc_prob} with just equality constraints $\{h_j(x)=0\}_j$ satisfies strong ncKKT.
\end{cor}



\subsubsection{Convexity}
As it turns out, any set of convex inequality constraints defines an Archimedean closed quadratic module. To make this statement precise, though, we need to recall the definition of convexity for non-commutative polynomials, over an algebra or in general.
\begin{defin}
A Hermitian non-commutative polynomial $p$ is convex in the $C^*$-algebra ${\A}$ if, for any two $n$-tuples of Hermitian operators $Y_1,Y_2\in{\A}^{n}$, it holds that
\begin{equation}
p\left(\delta Y_1+(1-\delta)Y_2\right)\leq\delta p(Y_1)+(1-\delta)p(Y_2),
\end{equation}
for all $\delta\in\R$, $0\leq \delta\leq 1$. If $p$ is convex for all $C^*$-algebras, then we call it \emph{matrix convex} \cite{convex_poly}.
\end{defin}
As shown in \cite{convex_poly}, $p\in\P$ is matrix convex iff its \emph{Hessian} is \emph{matrix positive}, i.e., if the polynomial
\begin{equation}
\frac{d^2p(x+th)}{dt^2}\Bigr|_{t=0}
\end{equation}
is a sum of squares of the Hermitian variables $x_1,\ldots ,x_n, h_1,\ldots ,h_n$. In \cite{convex_poly} it is also proven that matrix convex polynomials have degree at most two.

Now we are ready to state a sufficient criterion for Archimedean closure.

\begin{theo}
\label{thm:conv}
Let the Archimedean Problem \eqref{nc_prob} be such that 
\begin{enumerate}[\rm(a)]
\item
$\{g_i\}_i$ are matrix concave;
\item
the  equality constraints $\{h_j(x)=0\}_j$ are affine linear, i.e., of the form
\begin{equation}
h_j(x)=\sum_{k}\nu_{jk}x_k-b_j,
\end{equation}
and linearly independent. 
\end{enumerate}
In addition, let there exist $r\in\R^+$ and an $n$-tuple $q\in\P^{n}$ of Hermitian polynomials such that
\begin{equation}
g_i(q(x))-r
\label{posi_X2}
\end{equation}
is SOS, for $i=1,\ldots ,m$, and
\begin{equation}
h_j(q(x))\in I(\mathbf{h}),\quad j=1,\ldots ,m'.
\label{linear_X2}
\end{equation}
Then, the constraints $\{g_i(x)\geq 0\}_i\cup\{h_j(x)=0\}_j$ generate the Archimedean closed quadratic module
$M(\g)+I(\h)$. Thus, by Theorem \ref{thm:k->strong}, Problem \eqref{nc_prob} satisfies the strong ncKKT conditions.
\end{theo}

In the proof of the theorem we shall make use of Schur complements \cite{zhang2006schur}:

\begin{lemma}\label{lem:schur}
The block matrix 
\[
\begin{bmatrix}
a_{11} & a_{12} \\[1mm]
a_{12}^* & I
\end{bmatrix}
\]
is positive (semi)definite iff $a_{11}-a_{12}a_{12}^*$ is positive (semi)definite.
\end{lemma}

\begin{proof}
This is a special case of \cite[Theorem 1.12]{zhang2006schur}. Alternately, observe that
\[
\begin{bmatrix}
a_{11} & a_{12} \\[1mm]
a_{12}^* & I
\end{bmatrix} 
=
\begin{bmatrix}
I & a_{12} \\[1mm]
0 & I
\end{bmatrix}
\begin{bmatrix}
a_{11}-a_{12}a_{12}^* & 0 \\[1mm]
0 & I
\end{bmatrix}
\begin{bmatrix}
I & a_{12} \\[1mm]
0 & I
\end{bmatrix}^*.\qedhere
\]
\end{proof}

\begin{proof}[Proof of Theorem \ref{thm:conv}]
Let $(\H^\star,\sigma^\star,X^\star)$ be a solution of Problem \eqref{nc_prob_hilbert}. 
By \eqref{posi_X2} and \eqref{linear_X2}, the choice $\hat{X}=q(X^\star)$ satisfies the constraints 
\beq\label{eq:strictPos}
g_i(\hat{X})>0\quad \text{for all $i$},\qquad 
h_j(\hat{X})=0\quad \text{for all $j$}.
\eeq

We solve the system of linear equations $\{h_j=0\}_j$.
Without loss of generality, we express the last $n-s$ variables in terms of $x_1,\ldots,x_s$. By back substitution, our Problem \eqref{nc_prob_hilbert} is thus equivalent to one with $m'=0$, i.e., one without equality constraints.
More precisely, if a polynomial $p$ has a SOS decomposition (without $h_j$'s) after this back substitution, then the original form of $p$ has a SOS decomposition as in \eqref{eq:sos}.
It thus suffices to consider an Archimedean Problem \eqref{nc_prob_hilbert} with $m'=0$.

By the Helton-McCullough structure theorem for concave non-commutative polynomials \cite[Corollary 7.1]{convex_poly}, each $g_i$ is of the form
\begin{equation}\label{eq:cvxPoly}
g_i(x)=c_{i} + \Lambda_{i0}(x)-\sum_{j=1}^N\Lambda_{ij}(x)^*\Lambda_{ij}(x)
\end{equation}
for some $c_{i}\in\R$ and homogeneous linear $\Lambda_{ij}(x)$. 
(By 
Carath\'eodory's theorem on
convex hulls in finite dimensions \cite[Theorem I.2.3]{barvinok2002course}, the length $N$ of the sum of squares in \eqref{eq:cvxPoly} can be chosen independently of $i$.)
Such a $g_i$ is the Schur complement of a linear pencil, namely,
\begin{equation}\label{eq:schurPoly}
G_i(x)=\begin{bmatrix}
c_{i} + \Lambda_{i0}(x) & \Lambda_{i1}(x)^* & \cdots & \Lambda_{iN}(x)^* \\[1mm]
\Lambda_{i1}(x) & 1 \\
\vdots & & \ddots \\
\Lambda_{iN}(x) & & & 1
\end{bmatrix}
\end{equation}
Form the large block diagonal pencil $L(x):=G_1(x)\oplus\cdots\oplus G_m(x)$ of size $m(N+1)\times m(N+1)$. 

From $g_i(\hat X)>0$ for all $i$, we deduce from Lemma \ref{lem:schur} that
$L(\hat X)>0$. For a unit vector $\psi\in\H^\star$, consider 
$\hat x:=\psi^*\hat X\psi\in\R^n$. 
Since $L$ is linear and $\psi$ is a unit vector, 
$L(\psi^*\hat X\psi)=(I_{m(N+1)}\otimes \psi)^*L(\hat X)(I_{m(N+1)}\otimes \psi)$.
Letting $0\neq\eta\in\C^{m(N+1)}$, we have
\begin{equation}
\begin{aligned}
\braket{L(\hat x)\eta}{\eta} & = 
\braket{L(\psi^*\hat X\psi)\eta}{\eta} = 
\braket{(I_{m(N+1)}\otimes \psi)^*L(\hat X)(I_{m(N+1)}\otimes \psi)\eta}{\eta} \\
& = 
\braket{L(\hat X)(I_{m(N+1)}\otimes \psi)\eta}{(I_{m(N+1)}\otimes \psi)\eta} 
= 
\braket{L(\hat X)(\eta\otimes \psi)}{(\eta\otimes \psi)} >0. 
\end{aligned}
\end{equation}

Now again by Lemma \ref{lem:schur}, $g_i(\hat x)>0$ for all $i$. By translating $x$ by $-\hat x$, we may assume w.l.o.g.~that $g_i(0)>0$. By rescaling we further reduce to $g_i(0)=1$, i.e., all $g_i$ are monic. We are now in a position to apply the convex Positivstellensatz \cite[Theorem 1.2]{Helton_2012} to deduce that 
$f - p^\star$ has a SOS decomposition. Thus, by Theorem \ref{thm:k->strong}, Problem \eqref{nc_prob} admits the strong ncKKT conditions. 
\end{proof}

\begin{cor}
\label{cor:conv2}
Let the  Problem \eqref{nc_prob} be such that 
\begin{enumerate}[\rm(a)]
\item
$\{g_i\}_i$ are matrix concave;
\item
the  equality constraints $\{h_j(x)=0\}_j$ are affine linear, and linearly independent. 
\end{enumerate}
If there exists a feasible point $X$ for Problem \eqref{nc_prob} such that $g_i(X)>0$ for all $i$, then Problem \eqref{nc_prob} satisfies the strong ncKKT conditions.
\end{cor}
\begin{proof}
Immediate from the proof of Theorem \ref{thm:conv}.
\end{proof}

\begin{remark}
The Convex Positivstellensatz of \cite{Helton_2012} has minimal degree. Enforcing the strong ncKKT conditions on an NPO problem with matrix convex constraints is therefore superfluous: the first level of relaxation \eqref{k_dual} with degree high enough to encode the problem will already achieve convergence. 
\end{remark}

\section{Partial operator optimality conditions}
\label{sec:partialop}
As we will see in Section \ref{sec:bell}, the computation of the maximum quantum value of a linear Bell functional only seems to satisfy \core ncKKT. In this particular problem, and some interesting variants thereof \cite{statePoly, Ligthart_2023, Wolfe_2021}, the non-commuting variables $x=(x_1,\ldots ,x_n)$ can be partitioned as $x=(y,z)$, and the only constraints relating the parts $y$ and $z$ are commutation relations of the form:
\begin{equation}
[z_k,y_l]=0,\quad\forall k,l.
\end{equation}
As it turns out, if the remaining constraints on $z$ are convex and $f(y,z)$ is convex on $z$, then a partial form of strong ncKKT holds for the variables $z$. Similarly, if the remaining constraints on $z$ satisfy ncMFCQ, then the variables $z$ will satisfy a form of normed ncKKT. Both scenarios thus allow us to derive stronger SDP constraints than the mere \core ncKKT would permit, further boosting the speed of convergence of the vanilla SDP hierarchy \cite{Pironio2010}. 

In this section, we define these partial forms of operator optimality and find sufficient criteria to ensure that they hold. One of these criteria will be used in Section \ref{sec:bell} to derive new SDP constraints for maximizations over Bell functionals.

\subsection{Partial normed ncKKT conditions}
\begin{theo}
\label{theo_partial_MF}
Consider an Archimedean NPO \eqref{nc_prob} with variables $x=(y,z)$, with $z=(z_1,\ldots ,z_q)$ such that
\begin{enumerate}[\rm(a)]
    \item The only constraints involving both types of variables $y$ and $z$ are the following:
    \begin{equation}
    [y_r,z_s]=0,\quad\forall r,s.
    \end{equation}
    \item {The remaining constraints involving variables $z$ are 
    \beq
    \begin{aligned}
    &\hat{g}_i(z)\geq 0,\quad i=1,\ldots ,m_Z,\\ 
    &\hat{h}_j(z)=0,\quad j=1,\ldots ,m'_Z,
    \end{aligned}
    \eeq
    and satisfy ncMFCQ.}
    
\end{enumerate}
Then, for any bounded solution $\sigma^\star$ of Problem \eqref{nc_prob}, there exist Hermitian linear functionals $\{\mu_i:\Z\to\C\}_i$, $\{\lambda^{\pm}_i:\Z\to\C\}_j$, non-negative on $M(\{\hat{g}_i(z)\}_i)+I(\{\hat{h}_j(z)\}_j)$, such that
\begin{subequations}
\begin{align}
&\mu_i(\hat{g}_i)=0,\quad i=1,\ldots ,m_Z,\\ 
&\sigma^\star\left(\nabla_z f(p)\right) - \sum_i\mu_i\left(\nabla_z \hat g_i(p)\right)-\sum_j\lambda_j\left(\nabla_z \hat h_j(p)\right)=0,\quad \forall p\in\Z^{q},
\end{align}
with $\lambda_j=\lambda^+_j-\lambda^-_j$, for $j=1,\ldots ,m_Z$.
\end{subequations}
    
\end{theo}

\begin{proof}
Let $\sigma^\star$ be a solution of Problem \eqref{nc_prob}. For fixed $k\in\N$, consider the following SDP:
\begin{align}
\min_{\epsilon,\mu,\lambda}{} &\epsilon\nonumber\\
\mbox{s.t. }&\epsilon\geq 0,\nonumber\\
&\mu^k_i(pp^*)+\epsilon\|p\|^2_2\geq 0,\quad \forall p\in\Z,\; \deg (p)\leq k,\;i=1,.,,,m_Z,\nonumber\\ 
&\mu^k_i(p\hat{g}_lp^*)+\epsilon\|p\|^2_2\geq 0, \quad \forall p\in\Z,\; \deg (p)\leq k-\left\lceil\frac{\deg(\hat{g}_l)}{2}\right\rceil,\;i,l=1,\ldots ,m_Z,\nonumber\\
&\mu^k_i(s^+\hat{h}_js^-)=0, \quad \forall s^+, s^-\in\Z,\; \deg (s^+)+\deg(s^-)\leq 2k-\deg(\hat{h}_j), \; j=1,\ldots ,m_Z',\nonumber\\ 
&\mu^k_i(s\hat{g}_i)=\mu^k_i(\hat{g}_is)=0,\quad \forall s\in\Z,\;\deg(s)\leq 2k-\deg(g_i), i=1,\ldots ,m,\nonumber\\ 
&\lambda_j^k(s^+\hat{h}_js^-)=0,\quad \forall s^+,s^-\in\Z,\; \deg (s^+)+\deg(s^-)\leq 2k-\deg(h_j), \; j=1,\ldots ,m'_Z,\nonumber\\
&\sigma^\star\left(f'(p)\right)-\sum_i\mu^k_i\left(\nabla_z \hat{g}_i(p)\right)-\sum_j\lambda^k_j\left(\nabla_z h_j(p)\right)=0,\quad \forall p\in \Z^{q}, \nonumber\\ 
&\qquad\qquad \deg_z \left(f'(p)\right),\;\deg \left(\nabla_z \hat{g}_i\right),\;\deg \left(\nabla_z \hat{h}_j\right)\leq 2k-\deg(p),
\label{SDP_partial_MF}
\end{align}
where 
\begin{equation}
f'(p)=\lim_{\delta\to 0}\frac{f(Y^\star,Z^\star+\delta p(Z^\star))-f(Y^\star, Z^\star)}{\delta},
\end{equation}
and the expression $\deg_z \left(\bullet\right)$ denotes the degree of $\bullet$ with respect to the variables $z$.

Invoking Lemma \ref{lemma_ODE} as in the proof of Theorem \ref{essential_theo}, we find that this SDP problem has a solution for every $\epsilon>0$. Next, using relations \eqref{MF_cons}, \eqref{anni_MF}, like in the proof of Theorem \ref{theo_null_epsilon}, it is shown that Problem \eqref{SDP_partial_MF} is also feasible for $\epsilon=0$. Moreover, we can replace the variables $\mu_i^k$ by bounded positive functionals $\{\mu_i\}_i$, compatible with the constraints $\{\hat{g}_i(z)\geq0\}_i\cup\{\hat{h}_j(z)=0\}_j$. Finally, we invoke the linear independence of the gradients of $\{\hat{h}_j\}_j$ as in Lemma \ref{bounded_interm} to prove that $\{\lambda^k_j\}_j$ can be replaced by linear functionals.
\end{proof}

\subsection{Partial strong ncKKT conditions}

\begin{theo}
\label{theo_partial_conv}
Consider an NPO Problem \eqref{nc_prob}, not necessarily Archimedean, with variables $x=(y,z)$, with $z=(z_1,\ldots ,z_q)$ such that
\begin{enumerate}[\rm(a)]
    \item The only constraints involving both types of variables $y$ and $z$ are the following:
    \begin{equation}
    [y_r,z_s]=0,\quad\forall r,s.
    \label{commut_YZ}
    \end{equation}
    
    \item The remaining constraints involving variables of type $z$ are 
    \beq
    \begin{aligned}
    &\hat{g}_i(z)\geq 0,\quad i=1,\ldots ,m_Z,\\ 
    &\hat{h}_j(z)=0,\quad j=1,\ldots ,m'_Z,
    \end{aligned}
    \eeq
    where $\{\hat{g}_i\}_i$ are matrix concave non-commutative polynomials and $\{\hat{h}_j\}_j$ are affine linear  polynomials with linearly independent gradients.
    \item There exist $r\in\R^+$ and polynomials $Q=(Q_1,\ldots ,Q_q)$ such that
    \begin{equation}
    \hat{g}_i(Q(z))-r
    \label{posi_Z}
    \end{equation}
    is a weighted sum of squares, for $i=1,\ldots ,m_Z$, and
    \begin{equation}
    \hat{h}_j(Q(z))=\sum_{l,j'}s_{jj'l}(z)\hat{h}_{j'}(z)s'_{jj'l}(z),
    \label{linear_Z}
    \end{equation}
    for some polynomials $s_{jj'l}, s'_{jj'l}$, for $j=1,\ldots ,m_Z'$.
    \item Let $\{\tilde{g}_i(y)\geq 0\}_i\cup\{\tilde{h}_j(y)=0\}_j$ be the remaining constraints on $y$. Then $f(y,z)$ satisfies
    \begin{align}
    &\frac{d^2f(y,z+th)}{dt^2}\Bigr|_{t=0}=\sum_{l,k}\zeta_{lk}^+(y,z,h)[y_k,z_l]\zeta_{lk}^-(y,z,h)+ \eta_{lk}^+(y,z,h)[y_k,h_l]\eta_{lk}^-(y,z,h)\nonumber\\
    &+\sum_l\theta_l(z,h,y)\theta^*_l(z,h,y)+\sum_i\theta_{il}(z,h,y)\tilde{g}_i(y)\theta^*_{il}(z,h,y)+\sum_j\iota^+_{jl}(z,h,y)\tilde{h}_j(y)\iota^-_{jl}(z,h,y),
    \label{alg_proof_conv}
    \end{align}
    for some polynomials $\zeta^{\pm},\eta, \theta,\iota^{\pm}$.
\end{enumerate}
Let $(\H^\star,X^\star,\sigma^\star)$ be any bounded solution of Problem \eqref{nc_prob_hilbert}, with $X^\star=(Y^\star,Z^\star)$, and denote by $C^*(Z^\star)$ the $C^*$-algebra generated by $Z^\star_1,\ldots ,Z_q^\star$. Then, there exist positive linear functionals $\{\mu_i:C^*(Z^\star)\to\C\}_i$ and bounded Hermitian linear functionals $\{\lambda_i:C^*(Z^\star)\to\C\}_j$ such that
\begin{subequations}
\begin{align}
&\mu_i(\hat g_i(Z^\star))=0,\quad i=1,\ldots ,m_Z,\\ 
&\sigma^\star\left(\nabla_z f(p)\right) - \sum_i\mu_i\left(\nabla_z \hat g_i(p)\right)-\sum_j\lambda_j\left(\nabla_z \hat h_j(p)\right)=0,\quad \forall p\in C^*(Z^\star)^{q}.
\end{align}
\end{subequations}
\end{theo}

The following lemma is the key to arrive at this result.
\begin{lemma}
\label{slater_lemma}
Let $\A$ be a $C^*$-algebra, call $\A_h$ its set of Hermitian elements. Let $Z=(Z_1,\ldots ,Z_q)$ be a set of Hermitian operator variables and let $\tilde{f}:\A^{q}\to\C$ be a Hermitian convex function\footnote{Namely, $f(\delta Z^1+(1-\delta) Z^2)\leq \delta f(Z^1)+(1-\delta)f(Z^2)$, for all $\delta\in\R$, $0\leq\delta\leq 1$, $Z^1,Z^2\in\A_h^{q}$.}. Given some \emph{concave} non-commutative polynomials $\{\hat{g}_i\}$ and linearly independent affine linear  polynomials $\{\hat{h}_i\}$, consider the optimization problem
\begin{equation}\label{nc_prob_algebras}
\begin{aligned}
\min_{Z\in\A^{q}}\ &\hat{f}(Z)\\ 
\text{s.t. }&\hat{g}_i(Z)\geq 0,\quad i=1,\ldots ,m_Z,\\ 
&\hat{h}_j(Z)=0,\quad j=1,\ldots ,m'_Z.
\end{aligned}
\end{equation}
Suppose that an optimal solution exists, call it $Z^\star$. 

Further assume that Problem \eqref{nc_prob_algebras} admits a strictly feasible point, i.e., there exists a feasible tuple $\hat{Z}\in\A^{q}$ such that
\beq
\begin{split}
&\hat{g}_i(\hat{Z})>0,\quad i=1,\ldots ,m_Z,\\ 
&\hat{h}_j(\hat{Z})=0,\quad j=1,\ldots ,m'_Z.
\end{split}
\eeq
Then, there exist positive linear functionals $\mu_i:\A\to\C$, $i=1,\ldots,m_Z$, and bounded Hermitian linear functionals $\lambda_j:\A\to\C$, $j=1,\ldots ,m'$ satisfying
\begin{equation}
\mu_i(\hat{g}_i(Z^\star))=0,\quad i=1,\ldots ,m_Z,\label{conv_slackness}    
\end{equation}
such that $Z^\star$ is a solution of the unconstrained optimization problem
\begin{equation}
\min_{Z\in\A_h^{q}}{\cal L}(Z;\mu,\lambda),
\label{global_optim}
\end{equation}
with
\begin{align}
{\cal L}(Z;\mu, \lambda):=\hat{f}(Z)-\sum_i\mu_i(\hat{g}_i(Z))-\sum_j\lambda_j(\hat{h}_j(Z)).
\end{align}
\end{lemma}

\begin{proof}
It suffices to follow the classical proof of the Slater criterion for strong duality (cf.~\cite[\S 4.2]{nesterov}). Given $Z^\star$, we define the sets:
\beq
\begin{aligned}
A:=&\{(r,S,T):r\in\R,\,S\in \A_h^{m_Z},\,T\in \A_h^{m'_Z},\\ 
&\phantom{(r,s,T):\ \ }\,\exists Z_1,\ldots ,Z_n\in \A_h, \hat{f}(Z)\leq r,\, -\hat{g}_i(Z)\leq S_i,\,\hat{h}_j(Z)=T_j,\;\forall i,j\},\\ 
B:=&\{(\nu, 0,0):\nu<\hat{f}(Z^\star)\}.
\end{aligned}
\eeq
Clearly, $A\cap B=\emptyset$. Also, both sets are convex. Since they live in a real normed space (namely, $\R\times \A^{m_Z+m'_Z}$) and $B$ is open, the Hahn-Banach separation theorem
\cite[Theorem V.4(a)]{reedsimon} 
 implies that there exists a separating linear functional $(\phi, \mu, \lambda)$ and $\alpha\in\R$ such that
 \begin{subequations}
\begin{align}
&\phi r+\sum_i\mu_i(S_i)+\sum_j \lambda_j(T_j)\geq \alpha,\quad\forall (r,S,T)\in A,\label{eqs_A}\\
&\phi\nu\leq \alpha,\quad \forall (\nu,0,0)\in B.\label{eqs_B}
\end{align}
\end{subequations}
Note that, from the definition of $A$, for any $y\in \A$, $(r,S,T)\in A$ implies that $(r,S',T)\in A$, with $S'_i=S_i+yy^* $, and $S'_j=S_j,$ for $j\not=i$. Now, suppose that there exists $y$ such that $\mu_i(yy^* )<0$. Then, we could make the left-hand side of Eq.~\eqref{eqs_A} arbitrarily small, just by replacing $S_i$ with $S_i+uyy^* $, with $u\in \R^+$ sufficiently large. It follows that, for all $i$, $\mu_i(yy^* )\geq0$, i.e., $\{\mu_i\}_i$ are positive linear functionals of $\A$.

Notice as well that we can choose $\nu$ to be arbitrarily small in Eq.~\eqref{eqs_B}. It follows that $\phi\geq 0$. We next prove that $\phi>0$. 

From Eqs.~\eqref{eqs_A}, \eqref{eqs_B} we have that
\begin{equation}
\phi\hat{f}(Z) - \sum_i\mu_i(\hat{g}_i(Z))-\sum_j\lambda_j(\hat{h}_i(Z))\geq\alpha\geq \phi \hat{f}(Z^\star),\quad\forall Z\in\A^{\p}.
\label{casi_fin}
\end{equation}
Now, take $Z=\hat{Z}$. We have that
\begin{equation}
\phi(\hat{f}(\hat{Z})-\hat{f}(Z^\star)) \geq \sum_i\mu_i(\hat{g}_i(\hat{Z})).
\end{equation}
Now, suppose that $\phi=0$. Then, the equation above implies that $\mu_i(\hat{g}_i(\hat{Z}))=0$, for all $i$. Since $\hat{g}_i(\hat{Z})>0$, it follows that $\mu_i=0$ for all $i$. Hence we deduce that $\phi=0$ implies $\mu_i=0$ for all $i$. Therefore, Eq.~\eqref{eqs_A} implies that
\begin{equation}
\sum_j\lambda_j(\hat{h}_j(Z))\geq \alpha,\quad\forall Z.
\end{equation}
This can only be true if the left-hand side does not depend on $Z$ at all. Now, let $\hat{h}_j(Z):=\sum_k\beta_{jk}Z_k-b_j$. Non-dependence on $Z_k$ implies that the functional $\sum_j\beta_{jk}\lambda_j$ vanishes, for all $k$. Now, take any $W\in \A$ such that there exists $l$ with $\lambda_l(W)\not=0$. Then, we have that $\sum_j\beta_{jk}\lambda_j(W)=0$ for all $k$, and thus the rows of the matrix $\beta$ are not linearly independent. It follows that $\lambda_j=0$ for all $j$. However, that would imply that the separating linear functional $(\phi,\mu, \lambda)$ is zero, which contradicts the Hahn-Banach theorem.

From all the above it follows that $\phi>0$. Dividing Eq.~\eqref{casi_fin} by $\phi$, we have that
\begin{equation}
\hat{f}(Z) - \sum_i\tilde{\mu}_i(\hat{g}_i(Z))-\sum_j\tilde{\lambda}_j(\hat{h}_j(Z))\geq \hat{f}(Z^\star),\quad\forall Z,
\label{casi_fin2}
\end{equation}
where $\tilde{\mu}_i:=\frac{1}{\phi}\mu_i$ are positive linear functionals and $\tilde{\lambda}_j:=\frac{1}{\phi}\lambda_j$ are linear functionals. 

Finally, take $Z=Z^\star$ in Eq.~\eqref{casi_fin2}. We arrive at:
\begin{equation}
f(Z^\star) -\sum_i\tilde{\mu}_i(\hat{g}_i(Z^\star))\geq \hat{f}(Z^\star).
\end{equation}
This can only be true if the second term of the left-hand-side of the equation above vanishes, i.e., if $Z^\star$ satisfies the complementary slackness condition \eqref{conv_slackness}. In that case, 
\begin{equation}
{\cal L}(Z^\star,\tilde{\mu},\tilde{\lambda})=\hat{f}(Z^\star),
\end{equation}
and so, by the above equation and \eqref{casi_fin2}, $Z^\star$ is a global solution of the unconstrained problem 
\eqref{global_optim}.
\end{proof}

\begin{proof}[Proof of Theorem \ref{theo_partial_conv}]
Let $(\H^\star,\sigma^\star,X^\star)$ be a  solution of Problem \eqref{nc_prob_hilbert}, with $X^\star=(Y^\star,Z^\star)$. Call $\A$ the algebra generated by $Z^\star$. Since the only relations connecting $y$ with $z$ are the commutation relations \eqref{commut_YZ}, it follows that the solution $p^\star$ of Problem \eqref{nc_prob} satisfies
\beq\label{nc_aux}
\begin{aligned}
p^\star=&\min_{Z\in\A^{q}}\hat{f}(Z)\\ 
\text{s.t. }&\hat{g}_i(Z)\geq 0,\quad i=1,\ldots ,m_Z,\\ 
&\hat{h}_j(Z)=0,\quad j=1,\ldots ,m'_Z,
\end{aligned}
\eeq
with the function $\hat{f}:\A\to\C$ defined as:
\begin{equation}
\hat{f}(Z)=\sigma^\star(f(Y^\star, Z)).
\end{equation}
Moreover, one of the minimizers of \eqref{nc_aux} is $Z=Z^\star$. This function is convex by virtue of Eq.~\eqref{alg_proof_conv}, which ensures that its Hessian is non-negative \cite{convex_poly}.

In addition, by Eqs.~\eqref{posi_Z}, \eqref{linear_Z}, we know that the choice $\hat{Z}=Q(Z^\star)$ satisfies the constraints $\hat{g}_i(\hat{Z})>0$ for all $i$, $\hat{h}_j(\hat{Z})=0$ for all $j$. We can thus invoke Lemma \ref{slater_lemma} and conclude that $Z^\star$ is the solution of the unconstrained Problem \eqref{global_optim}, for some positive linear functionals $\{\mu_i\}_i$ and bounded Hermitian linear functionals $\{\lambda_j\}_j$. Next, for any $q$-tuple of symmetric polynomials $p$ on $z$, consider the following trajectory in $\A^{q}$ 
\begin{equation}
Z(t):=Z^\star+ tp(Z^\star).
\end{equation}
Since $Z^\star$ is a minimizer of Problem \eqref{global_optim}, it follows that
\begin{equation}
\begin{split}
0 & =\frac{d{\cal L}(Z(t);\mu,\lambda)}{dt}\Bigr|_{t=0}\\
&=\sigma^\star\left(\nabla_zf(Y^\star,z)(p(Z^\star))
\right)-\sum_i\mu^k_i\left(\nabla_z \hat{g}_i(p)\Bigr|_{Z=Z^\star}\right)-\sum_j\lambda_j\left(\nabla_z \hat{h}_j(p)\Bigr|_{Z=Z^\star}\right).
\end{split}
\end{equation}
Since this relation is valid for arbitrary $p\in\Z^{q}$, we arrive at the statement of the theorem.    
\end{proof}

\section{Applications}
\label{sec:applications}
\subsection{Many-body quantum systems}
\label{sec:many_body}
A consequence of the state optimality condition \eqref{state_optimality} is that the computation of the properties of condensed matter systems at zero temperature admits an NPO formulation. Consider, for instance, an $n$-qubit quantum system. Each such qubit or subsystem $j$ has an associated set of operators $\sigma_x^j,\sigma_y^j,\sigma_z^j$, which form a \emph{Pauli algebra}:
\begin{equation}
\begin{aligned}
&(\sigma_x^j)^2=(\sigma_y^j)^2=(\sigma_z^j)^2=1,\\ 
&\sigma_x^j\sigma_y^j-i\sigma^j_z=\sigma_y^j\sigma_z^j-i\sigma^j_x=\sigma_z^j\sigma_x^j-i\sigma^j_y=0.
\end{aligned}
\label{Pauli}
\end{equation}
In a sense, these operators represent everything we can measure in any such subsystem. Being independent systems, the operators of different subsystems commute:
\begin{equation}
[\sigma^j_a,\sigma^k_b]=0,\quad a,b\in\{x,y,z\},\;j\not=k.
\label{Pauli_comm}
\end{equation}
The set of constraints \eqref{Pauli}, \eqref{Pauli_comm} admits 
(up to unitary equivalence)
a unique irreducible operator representation $\pi:\P\to B(\C^2)^{\otimes n}$, with
\beq
\begin{split}
&\pi(\sigma_x^j):=\id_2^{\otimes {j-1}}\otimes \left(\begin{array}{cc}
 0 & 1 \\
  1   &0 
\end{array}\right) \otimes \id_2^{\otimes {n-j}},\\
&\pi(\sigma_y^j):=\id_2^{\otimes {j-1}}\otimes \left(\begin{array}{cc}
 0 & -i \\
  i   &0 
\end{array}\right)\otimes \id_2^{\otimes {n-j}},\\
&\pi(\sigma_z^j):=\id_2^{\otimes {j-1}}\otimes \left(\begin{array}{cc}
 1 & 0 \\
  0   &-1 
\end{array}\right)\otimes \id_2^{\otimes {n-j}}.
\end{split}
\eeq
It is easy to see that $\pi$ satisfies the conditions of Proposition \ref{prop:finite_SOS}.

The $n$ qubits jointly interact through a $2$-local Hamiltonian. This is an operator of the form
\begin{equation}
H(\sigma)=\sum_{j>k}^nP_{jk}(\sigma^j,\sigma^k),
\label{hamiltonian}
\end{equation}
where $\sigma^j:=(\sigma_x^j,\sigma_y^j,\sigma_z^j)$ and $P_{jk}$ is a polynomial of degree $2$. At zero temperature, the system is described by one of the eigenvectors of $\pi(H)$ with minimum eigenvalue. Any such eigenvector is called a \emph{ground state}.

For large $n$, computing $E_0(H)$ is Quantum-Merlin-Arthur-hard (\textbf{QMA}-hard)~\cite{Liu2007}. Quantum chemists~\cite{Nakata2001,Mazziotti2004,Mazziotti2023} (and, more recently, condensed matter physicists~\cite{Han2020quantum,Kull2022,Requena2023,Wang2023}) use NPO to lower bound $E_0(H)$. In essence, they relax the problem
\begin{equation}
\begin{aligned}
E_0(H)=&\min\rho(H)\\ 
\mbox{s.t. }&(\sigma_x^j)^2=(\sigma_y^j)^2=(\sigma_z^j)^2=1,\quad j=1,\ldots ,n,\\ 
&\sigma_x^j\sigma_y^j-i\sigma^j_z=\sigma_y^j\sigma_z^j-i\sigma^j_x=\sigma_z^j\sigma_x^j-i\sigma^j_y=0,\quad j=1,\ldots ,n,\\ 
&[\sigma^j_a,\sigma^k_b]=0,\quad a,b\in\{x,y,z\},\;j\not=k\\ 
\end{aligned}
\end{equation}
through hierarchies of SDPs. 

Knowing the ground state energy of a condensed matter system is very useful: if positive, it signals that the system is unstable; if negative, its absolute value corresponds to the minimum energy required to disintegrate it. 

However, both physicists and chemists are also interested in estimating other properties of the set of ground states. Take, for instance, the magnetization of the sample. Basic quantum mechanics teaches us that the magnetization $M$ of a condensed matter system at zero temperature lies in $[M^-,M^+]$, with
\begin{equation}
M^{\pm}:=\mp\min\{\mp\bra{\psi}\pi(\sum_j\sigma_z^j)\ket{\psi}:\bra{\psi}\bullet\ket{\psi}\in\Gr (\pi(H))\}.
\label{magnet}
\end{equation}
For instance, \citet{Wang2023} study a relaxation of this problem. First, using variational methods, they derive an upper bound $E^+_0$ on $E_0(H)$. Next, they relax the NPO problem:
\begin{equation}\label{variational}
\begin{aligned}
\bar{M}^{\pm}=&\mp\min\{\mp\rho(\sum_j\sigma_z^j)\}\\ 
\mbox{s.t. }&(\sigma_x^j)^2=(\sigma_y^j)^2=(\sigma_z^j)^2=1,\quad j=1,\ldots,n,\\ 
&\sigma_x^j\sigma_y^j-i\sigma^j_z=\sigma_y^j\sigma_z^j-i\sigma^j_x=\sigma_z^j\sigma_x^j-i\sigma^j_y=0,\quad j=1,\ldots,n,\\ 
&[\sigma^j_a,\sigma^k_b]=0,\quad a,b\in\{x,y,z\},\;j\not=k,\\ 
&\rho(H)\leq E^+_0.
\end{aligned}
\end{equation}
Any SDP relaxation of the problem above of order $k$ will produce two quantities $\bar{M}_k^{\pm}$, with the property that $M\in[\bar{M}_k^-,\bar{M}_k^+]$.

However, the method proposed by \citet{Wang2023} is only feasible when good variational methods for the considered Hamiltonian are available. Indeed, given a loose upper bound $E^+_0$ on $ E_0(H)$, one should not expect great results. Correspondingly, the numerical results of~\cite{Wang2023} are remarkable for 1D quantum systems. Those have Hamiltonians of the form $H=\sum_jP_{j,j+1}(\sigma^j,\sigma^{j+1})$, and one can obtain good approximations to their ground state energies via tensor network state methods~\cite{Cirac2021,Landau2015,Verstraete2006}. The results of~\cite{Wang2023} are not that good for 2D systems, namely, qubit systems with a Hamiltonian of the form \eqref{hamil_2D} below. For such systems, current variational tools are very imprecise~\cite{Schuch2007}.

The state optimality condition \eqref{state_optimality} allows us to formulate Problem \eqref{magnet} as the following NPO:
\begin{equation}\label{chemistry_opt}
\begin{aligned}
m^{\pm}=&\mp\min\rho(\mp\sum_j\sigma_z^j)\\ 
\mbox{s.t. }&(\sigma_x^j)^2=(\sigma_y^j)^2=(\sigma_z^j)^2=1,\quad j=1,\ldots,n,\\ 
&\sigma_x^j\sigma_y^j-i\sigma^j_z=\sigma_y^j\sigma_z^j-i\sigma^j_x=\sigma_z^j\sigma_x^j-i\sigma^j_y=0,\quad j=1,\ldots,n,\\ 
&[\sigma^j_a,\sigma^k_b]=0,\quad a,b=x,y,z;\quad j\not=k, \\
&\rho\in \Gr (H).
\end{aligned}
\end{equation}
In turn, the last constraint can be modeled by enforcing the relations:
\begin{subequations}
\begin{align}
\sigma([H,p])&=0
\label{commut_state_opt}
\\
\sigma\left(p^* H p-\frac{1}{2}\{H,p^* p\}\right)&\geq 0.
\label{pos_state_opt}
\end{align}
\end{subequations}
{Note that the quotient of $\C\langle\{\sigma_a^i:a=1,2,3,i=1,\ldots ,n\}\rangle$ with the Pauli constraints \eqref{Pauli}, \eqref{Pauli_comm} has dimension $4^n$. Hence, by Remark \ref{remark_suff_conds_state_opt} and Theorem \ref{theo:suff_state_optimality}, the resulting SDP hierarchy converges to the solution of Problem \eqref{chemistry_opt}.}

The advantage of the formulation \eqref{chemistry_opt} with respect to \eqref{variational} is that it does not require any upper bound on $E_0(H)$. Problem \eqref{chemistry_opt} is thus appropriate to tackle 2D and 3D systems, and even spin glasses~\cite{Mezard1987}. 

We illustrate our technique by bounding the ground state energy and magnetization of a translation-invariant Heisenberg model in 1D and 2D with periodic boundary conditions. All calculations were done using the toolkit for non-commutative polynomial optimization Moment \cite{Garner2024}, the modeller YALMIP \cite{yalmip}, and the solver MOSEK \cite{mosek}. The code used to produce these results is available at \url{https://github.com/ajpgarner/heisenberg_energy}, together with code doing the direct diagonalization of the sparse Hamiltonian used to compute the exact values we compare to.

For simplicity, the only symmetry of the problem we exploited was translation invariance. A full use of the symmetries of problem, as done in Ref.~\cite{Wang2023}, leads to dramatic improvements in performance.

In the 1D case, the Hamiltonian reads
\begin{equation}
H = \frac14\sum_{i=0}^{n-1}\sum_{a\in\{x,y,z\}} \sigma^i_a\sigma_a^{i\oplus 1},
\end{equation}
where addition $\oplus$ is modulo $n$.

Results for the energy are shown in Table \ref{tb:1denergy}. Here, the lower bounds are much tighter than the upper bounds. This is because when calculating the lower bound the state optimality condition \eqref{state_optimality} is only a tightening of the SDP: without it, the SDP hierarchy \eqref{k_relaxation} would converge anyway to the ground state energy. In calculating the upper bound, however, the state optimality condition \eqref{state_optimality} is doing all the work, as without it the SDP would converge to the system's maximum energy.

	\begin{table}[ht] 
	\begin{tabular}{c|ccc}
		$n$ & Lower bound & Exact value & Upper bound \\ \hline
        $6$ & $-0.4671$ & $-0.4671$ & $-0.4671$ \\
        $7$ & $-0.4079$ & $-0.4079$ & $-0.4079$ \\
		$8$ & $-0.4564$ & $-0.4564$ & $-0.4564$ \\
		$9$ & $-0.4251$ & $-0.4219$ & $-0.4189$ \\
		$10$& $-0.4515$ & $-0.4515$ & $-0.4306$ \\
		$11$& $-0.4460$ & $-0.4290$ & $-0.4020$ \\
        $12$& $-0.4492$ & $-0.4489$ & $-0.3886$ \\
        $13$& $-0.4475$ & $-0.4330$ & $-0.3987$ \\
        $14$& $-0.4518$ & $-0.4474$ & $-0.3013$ \\
        $15$& $-0.4506$ & $-0.4356$ & $-0.3001$ \\
        $16$& $-0.4509$ & $-0.4464$ & $-0.3013$ \\
        $17$& $-0.4501$ & $-0.4373$ & $-0.3004$ \\
	\end{tabular}
	\caption{Ground state energy per site of 1D Heisenberg model. Up to $n=10$ we use all nearest-neighbor monomials of degree up to 4, for $n=11$ until $n=13$ degree up to 3, and for higher $n$ degree up to 2.}
	\label{tb:1denergy}
	\end{table}

The magnetization is given by
\begin{equation}
M = \sum_{i=0}^{n-1} \sigma^i_z.
\end{equation}
Note that $H$ has the symmetry $\sigma_x^{\otimes n} H \sigma_x^{\otimes n} = H$, whereas the magnetization obeys $\sigma_x^{\otimes n} M \sigma_x^{\otimes n} = -M$. This implies that if the magnetization of the ground state $\ket{g}$ is $m$, then $\sigma_x^{\otimes n}\ket{g}$ will also be a ground state with magnetization $-m$. If these states are equal (up to a global phase), this implies that $m=0$. Otherwise the ground state is degenerate and both alternatives show up. We have found numerically that for even $n$ the magnetization per site (i.e., the magnetization divided by the number $n$ of qubits) is always zero, and for odd $n$ it is $\pm 1/n$.

Since the SDP respects the same symmetries as the original problem, if it gives $-m$ as a lower bound to the magnetization, it will give $m$ as an upper bound. Therefore we have reported the numerical results only for the lower bound of the magnetization, together with the lowest exact value. Results are shown in Table \ref{tb:1dmag}.

	\begin{table}[ht]
	\begin{tabular}{c|ccc}
		$n$ & Lower bound & Lowest exact value \\ \hline
        $6$ & $0$ & $0$  \\
        $7$ & $-0.1469$ & $-0.1429$ \\
        $8$ & $0$ & $0$ \\
        $9$ & $-0.1118$ & $-0.1111$ \\
        $10$& $-0.0315$ & $0$ \\
        $11$& $-0.1379$  & $-0.0909$ \\
        $12$& $-0.1422$  & $0$ \\       
        $13$& $-0.1378$  & $-0.0769$ \\       
        $14$& $-0.1780$  & $0$ \\       
        $15$& $-0.1742$  & $-0.0667$ \\       
        $16$& $-0.1715$  & $0$ \\       
        $17$& $-0.1693$  & $-0.0588$ \\            
	\end{tabular}
	\caption{Magnetization per site of 1D Heisenberg model. Up to $n=10$ we use all nearest-neighbor monomials of degree up to 4, for $n=11$ until $n=13$ degree up to 3, and for higher $n$ degree up to 2.}
	\label{tb:1dmag}
	\end{table}

In the 2D case the Hamiltonian reads
\begin{equation}
H = \frac14\sum_{i,j=0}^{L-1}\sum_{a\in\{x,y,z\}} \sigma^{i,j}_a(\sigma_a^{i\oplus 1,j} + \sigma_a^{i,j\oplus 1}),
\label{hamil_2D}
\end{equation}
where $\sigma^{i,j}_a$ denotes the Pauli matrix $a$ at the $(i,j)$ site of the square lattice. Results for the energy are shown in Table \ref{tb:2denergy}, and for the magnetization in Table \ref{tb:2dmag}.

	\begin{table}[ht]
	\begin{tabular}{c|ccc}
		$n$ & Lower bound & Exact value & Upper bound \\ \hline
        $3^2$ & $-0.4637$ & $-0.4410$ &  $-0.3647$ \\
        $4^2$ & $-0.7077$ & $-0.7018$ & $-0.2383$ \\
        $5^2$ & $-0.6732$ & & $-0.2388$ \\
        $6^2$ & $-0.7086$ & & $-0.2389$
	\end{tabular}
	\caption{Ground state energy per site of 2D Heisenberg model. For $L=3$ we used all nearest-neighbor monomials of degree up to 3. For $L=4$ we used degree up to $3$ for the moment matrix and up to $2$ for the state optimality condition \eqref{state_optimality}. For higher $L$ we used degree up to $2$.}
	\label{tb:2denergy}
	\end{table}

	\begin{table}[ht]
	\begin{tabular}{c|ccc}
		$n$ & Lower bound & Lower exact value  \\ \hline
        $3^2$ & $-0.2097$ & $-0.1111$  \\
        $4^2$ & $-0.2106$ & $0$ \\
        $5^2$ & $-0.3124$ & \\
        $6^2$ & $-0.3161$ & 
	\end{tabular}
	\caption{Magnetization per site of 2D Heisenberg model. For $L=3$ we used all nearest-neighbor monomials of degree up to 3. For $L=4$ we used degree up to $3$ for the moment matrix and up to $2$ for the state optimality condition \eqref{state_optimality}. For higher $L$ we used degree up to $2$.}
	\label{tb:2dmag}
	\end{table}

Problem \eqref{chemistry_opt} can also be adapted to deal with the thermodynamic limit, $n=\infty$. In that case, we demand the Hamiltonian to have a special symmetry called translation invariance. For one-dimensional materials, $H$ would be of the form:
\begin{equation}
H=\sum_{j=-\infty}^{\infty}P(\sigma^j,\sigma^{j+1}).
\label{hamiltonian_1D}
\end{equation}
The reader could be worried by the fact that there are infinitely many operator variables. However, we can take the state $\rho$ to be translation-invariant, i.e., invariant under the $*$-isomorphisms 
\begin{equation}
\pi_R(\sigma^j_a)= \sigma^{j+1}_a, \quad \pi_L(\sigma^j_a)= \sigma^{j-1}_a.
\end{equation}
In that case, one can relax the problem of minimizing the energy-per-site $e_0(H):=\min_\rho\rho(P(\sigma^1,\sigma^{2}))$ to
\begin{equation}
\begin{aligned}
e_0^n:=&\min\rho(P(\sigma^1,\sigma^{2}))\\ 
\mbox{s.t. }&(\sigma_x^j)^2=(\sigma_y^j)^2=(\sigma_z^j)^2=1,\quad j=1,\ldots,n,\\ 
&\sigma_x^j\sigma_y^j-i\sigma^j_z=\sigma_y^j\sigma_z^j-i\sigma^j_x=\sigma_z^j\sigma_x^j-i\sigma^j_y=0,\quad j=1,\ldots,n,\\ 
&[\sigma^j_a,\sigma^k_b]=0,\quad a,b\in\{x,y,z\},\; j\not=k\\ 
&\rho(p)=\rho(\pi_L(p))=\rho(\pi_R(p)),\quad\mbox{ for }\  p,\pi_L(p),\pi_R(p)\in \C\langle\sigma^1,\ldots,\sigma^n\rangle.
\end{aligned}
\end{equation}
It can be proven that $\lim_{n\to\infty}e_0^n$ coincides with the energy-per-site in the thermodynamic limit. The \emph{bootstrap technique} adds to this NPO the first optimality condition \eqref{commut_state_opt}~\cite{Han2020,Han2020quantum}. Note that, if $p\in\C\langle\sigma^2,\ldots,\sigma^{n-1}\rangle$ has degree $k$, then the expression in \eqref{commut_state_opt}
 has degree $k+1$ and only involves the variables $\sigma^1,\ldots,\sigma^n$.

The bootstrap technique thus allows computing lower bounds on $e_0(H)$. It cannot be used, however, to bound other properties of the ground states of $H$.

Things change dramatically when we add the optimality condition \eqref{pos_state_opt}, for it also allows us to bound whatever local property of the system, such as the magnetization. For $1D$ Hamiltonians \eqref{hamiltonian_1D}, if $p$ has degree $k$ and depends on the variables $\sigma^2,\ldots,\sigma^{n-1}$, the polynomial in Eq.~\eqref{pos_state_opt} will be of degree $2k+1$ and only depend on $\sigma^1,\ldots,\sigma^n$. Thus, even though we are working in the thermodynamic limit, the state optimality condition can be evaluated. This is, in fact, the case for any translation-invariant scenario in arbitrarily many spatial dimensions.

For any local property $o$, the corresponding SDP hierarchies will converge to the exact interval of allowed values for $o$ (at zero temperature). 

In order to illustrate our technique we bounded the ground state properties of the 1D Heisenberg model, with Hamiltonian given by 
\begin{equation}
H = \frac14\sum_{i=-\infty}^{\infty}\sum_{a\in\{x,y,z\}} \sigma^i_a\sigma_a^{i+1}.
\end{equation}
Results for the energy are shown in Table \ref{tb:1denergy_thermo}.

	\begin{table}[ht] 
	\begin{tabular}{c|ccc}
		$n$ & Lower bound & Upper bound \\ \hline
        $6$ & $-0.4671$ & $-0.3751$ \\
        $7$ & $-0.4564$ & $-0.3930$ \\
		$8$ & $-0.4564$ & $-0.4004$ \\
		$9$ & $-0.4520$ & $-0.4045$ \\
		$10$& $-0.4516$ & $-0.4069$ \\
		$11$& $-0.4500$ & $-0.4084$ \\
        $12$& $-0.4490$ & $-0.4097$ \\
	\end{tabular}
	\caption{Ground state energy per site of 1D Heisenberg model in the thermodynamic limit. For comparison the exact value is $1/4-\log(2) \approx -0.4431$. We used all nearest-neighbor monomials of degree up to 4.}
	\label{tb:1denergy_thermo}
	\end{table}


\subsection{The curious case of quantum Bell inequalities}
\label{sec:bell}

Consider a quantum bipartite Bell experiment~\cite{Bell1964,Tsirelson1987}, where two separate parties conduct measurements on an entangled quantum state. The first party, Alice, conducts measurement $x\in \{1,...,n\}$ and obtains outcome $a\in\{1,...,d\}$. The second party, Bob, respectively calls $y\in\{1,...,n\},b\in\{1,...,d\}$ his measurement setting and outcome. This configuration defines an \emph{$nndd$ Bell scenario}\footnote{In general bipartite Bell scenarios, the two parties can have different numbers of measurements and different numbers of possible outcomes. We will not consider such scenarios in this article.}. If Alice and Bob conduct many experiments, then they can estimate the probabilities $P=(P(a,b|x,y):x,y=1,\ldots ,n;\,a,b=1,\ldots ,d)$. Given a linear functional $C$ on $P$ (also called a Bell functional), we wish to determine the minimum value of 
\begin{equation}
C(P):=\sum_{a,b,x,y}C(a,b,x,y)P(a,b|x,y)
\end{equation}
compatible with quantum mechanics. This leads us to formulate the following NPO:
\begin{equation}\label{quantum_npo}
\begin{aligned}
c^\star=&\min\sigma\left(\frac{1}{2}\sum_{a,b,x,y}C(a,b,x,y)\{E_{a|x},F_{b|y}\}\right)\\ 
\mbox{s.t. }&E_{a|x}\geq 0,\quad \forall a,x,\\ 
&\sum_aE_{a|x}-\id=0,\quad\forall x,\\ 
&F_{b|y}\geq 0,\quad\forall b,y,\\ 
&\sum_bF_{b|y}-\id=0,\quad\forall y,\\ 
&[E_{a|x},F_{b|y}]=0,\quad\forall a,b,x,y.
\end{aligned}
\end{equation}
As we can appreciate, taking the partition $X=(E,F)$, the NPO satisfies the conditions of Theorem \ref{theo_partial_conv}, with $Q^A_{a|x}(E)=Q^B_{b|y}(F)=\frac{1}{d}$ for all $a,b,x,y$.

It so happens that the solution $(E^\star, F^\star)$ of Problem \eqref{quantum_npo} can be chosen such that the non-commuting variables are all projectors~\cite{Junge_2011}. That is,
\begin{equation}
\begin{aligned}
(E^\star_{a|x})^2&=E^\star_{a|x},\quad \forall a,x,\\ 
(F^\star_{b|y})^2&=F^\star_{b|y},\quad \forall b,y.
\label{rel_proj}
\end{aligned}
\end{equation}
The standard SDP hierarchy of relaxations for problem (\ref{quantum_npo}) under restrictions (\ref{rel_proj}) is usually dubbed the Navascu\'es-Pironio-Ac\'in (NPA) hierarchy \cite{NPA2007,NPA2008}. Since we will next define extra constraints, based on optimality conditions, to boost its speed of convergence, we will in the following refer to NPA without optimality conditions as ``vanilla NPA''.

Next, we apply Theorem \ref{theo_partial_conv} independently to Alice's algebra $\A$ (generated by the projectors $E^\star$'s) and to Bob's algebra $\B$ (generated by the projectors $F^\star$'s). The partial strong ncKKT conditions imply that we can, not only demand the state $\sigma$ to be compatible with relations \eqref{rel_proj}, but also the Lagrangian multipliers of Alice's $\mu_{a|x}^A,\lambda_x^A$ and Bob's $\mu_{b|y}^B,\lambda_y^B$ to be respectively compatible with the constraints $\{E_{a|x}^2-E_{a|x}=0\}_{x,a}\cup\{1-\sum_aE_{a|x}=0\}_x$ and $\{F_{b|y}^2-F_{b|y}=0\}_{y,b}\cup\{1-\sum_bF_{b|y}=0\}_y$. 

Calling ${\E}$ ($\F$) the set of polynomials on the $E$'s ($F's)$, the operator optimality relations for the $E$'s read:
\begin{subequations}
\begin{align}
&\mu^A_{a|x}(ss^*)\geq 0,\quad\forall a,x,\; \forall s\in\E,\\
&\mu^A_{a|x}\Big(s\big((E_{a'|x'})^2-E_{a'|x'}\big)s'\Big)=0,\quad\forall a,a',x,x',\;\forall s,s'\in\E\\
&\mu^A_{a|x}\Big(s\Big(\sum_{a'}E_{a'|x'}-1\Big)s'\Big)=0,\quad\forall a,x,x',\;\forall s,s'\in\E,\\
&\lambda^A_{x}\Big(s\big((E_{a|x'})^2-E_{a|x'}\big)s'\Big)=0,\quad\forall a,x,x',\;\forall s,s'\in\E\\
&\lambda^A_{x}\Big(s\Big(\sum_{a}E_{a|x}-1\Big)s'\Big)=0,\quad\forall x,\;\forall s,s'\in\E,\\
&\mu^A_{a|x}(E_{a|x})=0,\quad\forall a,x,\\
&\sigma\left(\frac{1}{2}\sum_{b,y}C(a,b,x,y)\{p,F_{b|y}\}\right)=\mu^A_{a|x}\left(p\right)+\lambda^A_x(p),\quad\forall a,x,\;\forall p\in \E.
\end{align}
\end{subequations}
The operator optimality relations for the $F$'s are the same, under the replacements $E\to F$, $a\to b$, $x\to y$, $\E\to \F$, $A\to B$. The reader can find the full optimization problem, including the state optimality conditions in Appendix \ref{app:rollo}.

\subsubsection{Only two outcomes}
When $a,b$ can only take two values, it is customary to rewrite Problem \eqref{quantum_npo} in terms of ``dichotomic operators'' $A_x:=E_{1|x}-E_{2|x}, B_y:=F_{1|y}-F_{2|y}$. The problem to solve is thus
\beq
\begin{aligned}
\min\  &\sigma\left(H\right)\\ 
\mbox{s.t. }&\frac{1-A_x}{2}\geq 0,\;\frac{1+A_x}{2}\geq 0,\quad\forall x,\\ 
&\frac{1-B_y}{2}\geq 0,\;\frac{1+B_y}{2}\geq 0,\quad\forall y,\\ 
&[A_x,B_y]=0,\quad\forall x,y
\end{aligned}
\eeq
where $H$ is the Bell polynomial
\begin{equation}
    \frac{1}{2}\sum_{x,y}c_{xy}\{A_x,B_y\}+\sum_x d_x A_x+\sum_y e_yB_y.
\end{equation}
To simplify notation, we define the polynomials
\begin{subequations}
\begin{align}
&{\cal F}_x:=\sum_y \frac12 c_{xy}B_y + \frac12 d_{x}\id, \\
&{\cal G}_y:=\sum_x \frac12c_{xy}A_x + \frac12 e_{y}\id,
\end{align}
\end{subequations}
which allow us to express $H$ as
\begin{equation}
H =\sum_{x} \{{\cal F}_x,A_x\} + \sum_ye_yB_y = \sum_{y} \{{\cal G}_y, B_y\} +  \sum_x d_xA_x.
\end{equation}
As before, it can be shown that the minimizers $(A^\star,B^\star)$ can be chosen such that 
\begin{equation}
(A^\star_x)^2=(B^\star_y)^2=1.
\end{equation}
Thus, once more we can apply Theorem \ref{theo_partial_conv} to the algebra $\A$ generated by $A^\star_1,\ldots ,A^\star_n$ and conclude that one can add new state multipliers $\mu_x^+,\mu_x^-$ to the problem, with the properties:
\begin{subequations}
\begin{align}
&\mu_x^{\pm}(s(A_{x'}^2-1)s')=0,\quad \forall x,x',\;\forall s,s'\in \A,\label{dual_dicho}\\
&\mu^{\pm}_x\left(\frac{1\pm A_x}{2}\right)=0,\quad\forall x,\label{slack_dicho}\\
&\sigma\left(\{p,{\cal F}_x\}\right)=\mu_x^+(p)-\mu_x^-(p),\quad\forall x, \;\forall p\in\A.\label{optimality_dicho}
\end{align}
\end{subequations}

If one does not wish to introduce new variables $\mu^{\pm}_x$ into the NPO problem, it is easy to get a relaxation of the conditions above that only involves evaluations with the already existing variable $\sigma$. 

Let $E^{\pm}_x:=\frac{1\pm A_x}{2}$. From Eq.~\eqref{slack_dicho} and the positivity of $\mu_x^\pm$, an analogous argument to the one used to derive Eqs.~\eqref{vanishing_averages} shows that
\begin{equation}
\mu^{\pm}_x\left(E_x^{\pm}p\right)=\mu^{\pm}_x\left(p E_x^{\pm}\right)=0\quad \forall p \in \A.
\label{anni_Es}
\end{equation}
Taking $p=\{A_x,q\}$ we find that
\begin{equation}
\mu^{\pm}_x\left(\{A_x,q\}\right)=\mp\mu^{\pm}_x\left(q+A_xqA_x\right)=0.
\end{equation}
Thus, if we set $p=-\{A_x,q\}$ in Eq.~\eqref{optimality_dicho}, we arrive that
\begin{equation}
-\sigma\left(\{\{A_x,q\},{\cal F}_x\}\right)=\mu^+(q+A_xqA_x)+\mu^-(q+A_xqA_x).
\end{equation}
In particular, taking $q=ss^*$, we have that
\begin{equation}
-\sigma\left(\{\{A_x,ss^*\},{\cal F}_x\}\right)=\mu^+(ss^*+A_xss^*A_x)+\mu^-(ss^*+A_xss^*A_x)\geq 0,\quad\forall s\in\E.
\label{positivity_dicho}
\end{equation}
Setting $p=[A_x,q]$ in Eq.~\eqref{optimality_dicho} and using Eq.~\eqref{anni_Es}, we obtain another useful constraint:
\begin{equation}
\sigma\left(\{[A_x,q],{\cal F}_x\}\right)=0,\quad\forall q\in\A.
\label{commutator}
\end{equation}
Constraints \eqref{positivity_dicho}, \eqref{commutator} are, respectively, extra positivity and linear conditions that one can apply to the already existing variables of the `quantum NPO' \eqref{quantum_npo}.

\subsubsection{Numerical implementation}

In order to implement numerically the constraints \eqref{positivity_dicho} and \eqref{commutator}, together with the analogous constraints for Bob and the state optimality conditions \eqref{state_optimality}, we express them in terms of a basis of monomials. Let $\{m_i^A\}_i$ and $\{m_i^B\}_i$ be a basis of monomials belonging to Alice's and Bob's algebra of operators, and $\{m_i\}_i$ a basis for the entire algebra. Then the equality constraints become
\begin{subequations}
    \begin{align}
    \sigma\left({\cal F}_x[A_x,m_i^A]\right)&=0,  \\
    \sigma\left({\cal G}_y[B_y,m_i^B]\right)&=0,  \\
    \sigma\left([H,m_i]\right)&=0,  \label{eq:bell_equality_gamma}
    \end{align}
\end{subequations}
where we are using the fact that ${\cal F}_x$ commutes with every element of Alice's algebra, and the analogous condition for Bob. The positivity conditions \eqref{positivity_dicho} are equivalent to the positive semidefiniteness of the matrices $\{\alpha^x\}_x,\{\beta^y\}_y,\gamma$, with elements given by
\begin{subequations}\label{eq:bell_positivity}
    \begin{gather}
\alpha^x_{ij} := -\sigma\left({\cal F}_x\{A_x,{m_i^A}^*m_j^A\}\right), \\
\beta^y_{ij} := -\sigma\left({\cal G}_y\{B_y,{m_i^B}^*m_j^B\}\right), \\
\gamma_{ij} := \sigma\left(m_i^* H m_j-\frac{1}{2}\{H,m_i^* m_j\}\right). \label{eq:bell_positivity_gamma}
    \end{gather}
\end{subequations}
Note that, when dealing with Bell inequalities, it is more usual to formulate the problem as a maximization instead of a minimization \cite{NPA2007}. One can adapt the KKT constraints for maximization by simply flipping the sign of the positivity conditions \eqref{eq:bell_positivity}.

Most commercial SDP solvers are based on interior-point methods, as they guarantee fast and accurate results. In order for interior point algorithms to work reliably, it is however vital to ensure that the problem we are solving is strictly feasible, that is, that there exists a point that satisfies all the equality constraints and has strictly positive eigenvalues in the positive semidefiniteness constraints \cite{Drusvyatskiy2017}. Although the vanilla NPA hierarchy is always strictly feasible \cite{Tavakoli2023}, this is in general not true when additional constraints are enforced \cite{Araujo2023}. This is in fact the case here, as the matrix $\gamma$ will necessarily have linearly dependent columns, and therefore some of its eigenvalues will be zero. To see that, we use Eq.~\eqref{eq:bell_equality_gamma} to rewrite Eq.~\eqref{eq:bell_positivity_gamma} as
\begin{equation}
\gamma_{ij} = \sigma(m_i^*[H,m_j]).
\end{equation}
This implies that a sufficient condition for some columns $\gamma_{\_j}$ to be linearly dependent is that the corresponding operators $[H,m_j]$ are linearly dependent. This is always the case if $m_j = \id$, as $[H,\id]=0$, or if $\{m_j\}_j$ is a set of monomials that can express $H$ itself, as $[H,H]=0$. Additional linear dependencies can show up for specific choices for $H$. In the Bell inequalities we consider in this section, the only additional dependencies that appeared were in the case of the tilted CHSH \eqref{eq:tilted_chsh}, for which $[H,\{A_0,A_1\}] = [H,\{B_0,B_1\}] = 0$. We removed these dependencies simply by removing enough monomials from the set used to define $\gamma$. We verified numerically that after doing that, the problem was always strictly feasible.

We illustrate the technique by minimizing Bell functionals in the 2222, 3322, and 4422 Bell scenarios. All calculations were done using the toolkit for non-commutative polynomial optimization Moment \cite{Garner2024}, the modeller YALMIP \cite{yalmip}, and the arbitrary-precision solver SDPA-GMP \cite{Nakata2010}. We are particularly interested in checking whether we have achieved convergence at some level. To do so, we verify that the moment matrix is flat \cite{Henrion2005, Nie2013b, Laurent2014, Burgdorf2016, Magron2023} or, in physicists' slang, that it has a \emph{rank loop} \cite{NPA2008}. We remark under the corresponding table whether it holds. 

We start with a tilted version of the CHSH inequality~\cite{CHSH1969}, with an additional $\tau (A_0 + B_0)$ term~\cite{Eberhard1993,Liang2011}. In full correlation notation the coefficients table is given by
\begin{equation}\label{eq:tilted_chsh}
\left(\begin{array}{@{}r|rr@{}}
0 & \tau & 0 \\ \hline
\tau & \phantom{-}1 & 1\\
0 & 1 & -1
\end{array}\right).
\end{equation}
This table, which represents a Hermitian polynomial on the operators $1$, $\{A_x\}_x, \{B_y\}_y$ has to be understood as follows: the rows are labeled by the operators $O=(1,A_1,A_2)$; the columns, by the operators $O'=(1,B_1,B_2)$. The $O_j,O_k'$-th entry of the table corresponds to the coefficient multiplying $\frac{1}{2}\{O_j,O_k'\}$. 

For two different values of $\tau$, we upper bound the maximum average of this polynomial, see Tables \ref{tb:chsh_kkt95} and \ref{tb:chsh_kkt99}. As $\tau$ tends to $1$ the level at which the NPA hierarchy converges exactly seems to get ever higher.
	\begin{table}[htb]
	\begin{tabular}{c|c|c}
		level & NPA & NPA+KKT \\ \hline
        $2$ & $3.9003\ 2967$ & $3.9003\ 1859$  \\
        $3$ & $3.9001\ 6474$ & $3.9001\ 6389$ \\
        $4$ & $3.9001\ 6389$ &   \\
	\end{tabular}
	\caption{Results for the tilted CHSH inequality with $\tau=0.95$. For comparison, the analytical value is $3.9001\ 6389\ 9372$ \cite{gigena24}. With the KKT constraints we get a rank loop at level 3, and without at level 4.}
	\label{tb:chsh_kkt95}
	\end{table}

	\begin{table}[!h]
	\begin{tabular}{c|c|c}
		level & NPA & NPA+KKT \\ \hline
        $2$ & $3.9800\ 1157$ & $3.9800\ 1078$  \\
        $3$ & $3.9800\ 0416$ & $3.9800\ 0280$ \\
        $4$ & $3.9800\ 0217$ & $3.9800\ 0132$  \\
        $5$ & $3.9800\ 0156$ &  \\
        $6$ & $3.9800\ 0132$ & 
	\end{tabular}
	\caption{Results for the tilted CHSH inequality with $\tau=0.99$. For comparison, the analytical value is $3.9800\ 0132\ 8893$ \cite{gigena24}. With the KKT constraints we find a rank loop at level $4$; without, at level $7$.}
	\label{tb:chsh_kkt99}
	\end{table}

Our next example is the well-studied I3322 inequality~\cite{Collins2004,Pal2010,Rosset2018}. In full correlation notation the coefficients table is given by
\begin{equation} \frac14
\left(\begin{array}{@{}r|rrr@{}}
     0   & -1  &  -1  &   0 \\ \hline
    -1   & -1  &  -1  &  -1 \\
    -1   & -1  &  -1  &   1 \\
     0   & -1  &   1  &   0
\end{array}\right).
\end{equation}
The results are shown in Table \ref{tb:i3322_kkt}. Its maximal violation is conjectured to occur only for an infinite-dimensional system \cite{Pal2010}, and a rank loop implies the existence of a finite-dimensional system achieving the maximum. Therefore we expected to find no rank loop here, as was indeed the case. For increased efficiency the calculations here were done without the positivity conditions \eqref{eq:bell_positivity}, as they did not seem to improve the results.
	\begin{table}[!h]
	\begin{tabular}{c|c|c}
		level & NPA & NPA+KKT  \\ \hline
        $2$ & $1.2509\ 3972$ & $1.2509\ 3965$  \\
        $3$ & $1.2508\ 7556$ & $1.2508\ 7554$  \\
        $4$ & $1.2508\ 7540$ & $1.2508\ 7538$ \\
        $5$ & $1.2508\ 7538$
	\end{tabular}
	\caption{Results for the I3322 inequality. For comparison, the best known lower bound is $1.2508\ 7538\ 4513$. No rank loop was found.}
	\label{tb:i3322_kkt}
	\end{table}

Our final example is the $I_{4422}^{20}$ inequality~\cite{Brunner2008}, that had a gap between the best known lower bound and the best known upper bound (see Table~IV of Ref.~\cite{Pal2009}). In full correlation notation the coefficients table is given by
\begin{equation}\frac14
\left(\begin{array}{@{}r|rrrr@{}}
     -12   & -1  &  -1  &  -2 &  4 \\ \hline
    -1   & -1  &   1  &   1 &  2 \\
    -1   &  1  &  -1  &   1 &  2 \\
    -2   &  1  &   1  &  -2 &  2 \\
     4   &  2  &   2  &   2 & -2
\end{array}\right).
\end{equation}
The results are shown in Table \ref{tb:i20_kkt}. For increased efficiency the calculations here were done without the positivity conditions \eqref{eq:bell_positivity}, as they did not seem to improve the results.
 \begin{table}[!h]
	\begin{tabular}{c|c|c}
		level & NPA & NPA + KKT  \\ \hline
        $2$ & $0.5070\ 6081$ & $0.5020\ 4577$  \\
        $3$ & $0.4677\ 5783$ & $0.4676\ 7939$  \\
        $4$ & $0.4676\ 7939$
	\end{tabular}
	\caption{Results for the inequality $I_{4422}^{20}$. For comparison, the best known lower bound is $0.4676\ 7939$. With the KKT constraints we get a rank loop at level 3, and without at level 4.}
	\label{tb:i20_kkt}
	\end{table}

We have also experimented with the weak ncKKT conditions from Definition \ref{def:weak_ncKKT}, even though they are not proven to hold for this problem. In all cases we got the same numerical answer as with the partial KKT conditions, except in the tilted CHSH case \eqref{eq:tilted_chsh}, where we got numerical problems.



\section{Conclusion}
In this work, we have generalized the KKT optimality conditions to non-commutative polynomial optimization problems (NPO). These enforce new equality and positive semidefinite conditions on the already existing hierarchies of SDPs used in NPO. 

The state optimality conditions (Eq.~\eqref{state_optimality}) and \core ncKKT (Eq.~\eqref{essential_KKT}), a loose version of the operator optimality conditions, hold for all problems. In contrast, normed and strong ncKKT conditions need to be justified through some constraint qualification. The existence of an SOS certificate to solve the NPO problem is enough to guarantee that the strong ncKKT conditions hold. This property is difficult to verify for most NPO problems \emph{a priori}. However, we found that it is satisfied by all Archimedean NPO problems with strictly feasible convex constraints or a faithful finite-dimensional representation. 

In addition, we generalized a known `classical' qualification constraint: Mangasarian--Fromovitz Constraint Qualification (MFCQ), which legitimates the use of normed ncKKT conditions in NPO. We also presented very mild conditions that guarantee that at least some relaxed form of either the normed or the strong ncKKT conditions holds. Those conditions are satisfied in the NPO formulation of quantum nonlocality, and thus have immediate practical applications.

We tested the effectiveness of the non-commutative KKT conditions by upper bounding the maximal violation of bipartite Bell inequalities in quantum systems. We found that the partial ncKKT conditions do improve the speed of convergence of the SDP hierarchy, sometimes achieving convergence at a finite level. This hints that the collapse of Lasserre's hierarchy of SDP relaxations \cite{Lasserre2001} under the KKT constraints, proven in  \cite{Nie2013}, might extend to the non-commutative case.

Similarly, we applied the state optimality conditions to bound the local properties of ground states of many-body quantum systems. Prior to our work, there was no mathematical tool capable of delivering rigorous bounds that did not rely on variational methods, see \cite{Wang2023}. It is intriguing whether the state optimality conditions can be integrated within renormalization flow techniques, like those proposed in \cite{Kull2022}. That would allow one to skip several levels of the SDP hierarchy through a careful (Hamiltonian-dependent) trimming of irrelevant degrees of freedom, thus delivering much tighter bounds on key physical properties.

\begin{acknowledgements}
\begin{wrapfigure}{r}[0cm]{2cm}
\begin{center}
\includegraphics[width=2cm]{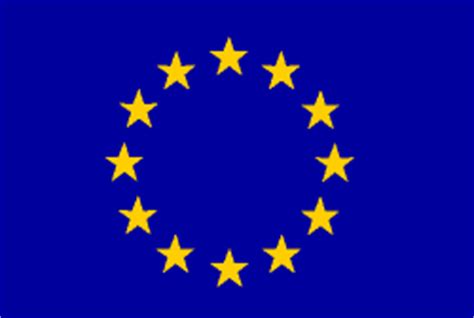}    
\end{center}
\end{wrapfigure} 
This project was funded within the QuantERA II Programme that has received funding from the European Union's Horizon 2020 research and innovation programme under Grant Agreement No 101017733, and from the Austrian Science Fund (FWF), project I-6004. M.A and A.G. acknowledge funding from the FWF stand-alone project P 35509-N. M.A. also acknowledges support by the European Union--Next Generation UE/MICIU/Plan de Recuperación, Transformación y Resiliencia/Junta de Castilla y León, by the Department of Education of the Junta de Castilla y León and FEDER Funds (Reference: CLU-2023-1-05), and by the Spanish Agencia Estatal de Investigación, Grants No. RYC2023-044074-I and PID2024-161725OA-I00. 
I.K.'s work was performed within the project COMPUTE, funded within the QuantERA II Programme that has received funding from the EU's H2020 research and innovation programme under the GA No 101017733.
I.K. was also supported by the Slovenian Research Agency
program P1-0222 and 
grants J1-50002, N1-0217, J1-60011, J1-50001, J1-3004 and J1-60025. Partially supported by the Fondation de l’Ecole polytechnique as part of the Gaspard Monge Visiting Professor Program. I.K. thanks Ecole Polytechnique and Inria Paris Saclay
for hospitality during the preparation of this manuscript.
T.V. acknowledges the support of the European Union (QuantERA eDICT, CHIST-ERA MoDIC), the National Research, Development and Innovation Office NKFIH (Grants No.~2019-2.1.7-ERA-NET-2020-00003,  No.~2023-1.2.1-ERA\_NET-2023-00009 and No.~K145927) and the `Frontline' Research Excellence Program of the NKFIH (No.~KKP133827).

\end{acknowledgements}

\bibliography{biblios_KKT}

\eject
\begin{appendix}
\section{Heuristic derivation of the first-order optimality conditions}
\label{app:heuristic}
The KKT conditions can be non-rigorously derived by considering infinitesimal variations over $x^\star$ of the Lagrangian functional
\begin{equation}
{\cal L}_c(x,\mu,\lambda)=f(x)-\sum_{i\in\mathbb{A}(x^\star)}\mu_ig_i(x)-\sum_j\lambda_jh_j(x),
\label{classical_lagrangian}
\end{equation}
with $\{\mu_i\}_i\subset \R^+$, $\{\lambda_j\}_j\subset \R$.

In this work, we seek the non-commutative analogs of the first-order optimality conditions \eqref{KKT_classical}. To find them, we start by writing a Lagrangian for Problem \eqref{nc_prob_hilbert}. 

Let $(\H^\star,X^\star,\psi^\star)$ be any bounded solution of Problem \eqref{nc_prob_hilbert}. For technical convenience, rather than considering general variations of the operators $X_1^\star,\ldots,X_n^\star$ within $B(\H^\star)$, we will demand those to be within $C^*(X^\star)$, the unital $C^*$-algebra generated by the operators $X^\star:=(X_1^\star,\ldots , X_n^\star)$. Consequently, from now on we regard the state $\psi^\star$ as a linear, positive, normalized functional on $C^*(X^\star)$. 

Note that any operator inequality constraint $g(\bar{X})\geq 0$ can be interpreted as an infinity of inequality constraints of the form $\omega(g(\bar{X}))\geq 0$ for all states $\omega:C^*(X^\star)\to\C$. The set of active constraints at $X^\star$ thus corresponds to $\{\omega(g(\bar{X}))\geq 0 :\omega\geq 0,\;\omega(g(X^\star))=0\}$. Similarly, the equality constraint $h(\bar{X})=0$ is equivalent to the condition $\xi(h(\bar{X}))=0$, for all Hermitian linear functionals $\xi:C^*(X^\star)\to\C$, and the set of active constraints associated to the positivity of the state $\psi^\star$ is $\{\psi(w)\geq0:w\geq0,\psi^\star(w)=0\}$.

Bearing the last two paragraphs in mind, the non-commuting analog of the classical Lagrangian \eqref{classical_lagrangian} is:
\begin{align}
{\cal L}& =\psi(f(\bar{X}))-\int_{\substack{\psi^\star(w)=0,\\[.3mm] w\geq 0\hfill}} M(w)\mathrm{d}w\psi\left(w\right)+\alpha(1-\psi(1))\nonumber\\
&\phantom{=}-\sum_{i}\int_{\substack{\omega(g_i(X^\star))=0,\\[.3mm] \omega\text{ state}\hfill}} \nu_i(\omega)\mathrm{d}\omega\,\omega(g_i(\bar{X}))-\sum_j\int \Theta_j(\xi)\mathrm{d}\xi\,\xi(h_j(\bar{X})),
\label{lagrangian_first_atempt}
\end{align}
where the Hermitian operator variables $\bar{X}=(\bar{X}_1,\ldots ,\bar{X}_n)$ are elements of $C^*(X^\star)$ and $\psi:C^*(X^\star)\to\C$ represents our state variable. The multipliers $\alpha$, $M(w)\mathrm{d}w$, $\{\nu_i(\omega)\mathrm{d}\omega\}_i$, $\{\Theta_j(\xi)\mathrm{d}\xi\}_j$ respectively denote a real variable and measure-type variables over the set of positive semidefinite $w\in C^*(X^\star)$, the set of states $\omega$ and the set of Hermitian linear functionals $\xi$.

Notice that integration over the set of states (or functionals, or positive semidefinite elements) of a $C^*$-algebra might not be well defined if the latter is infinite dimensional. This lack of mathematical rigor is all right: we are not aiming to \emph{prove} optimality conditions, just to guess their form. 

For simplicity, for each $i$, we next define a new multiplier $\mu_i$, of the form
\begin{equation}
\mu_i:=\int_{\substack{\omega(g_i(X^\star))=0,\\[.3mm] \omega\text{ state}\hfill}} \nu_i(\omega){\rm d}\omega\, \omega.
\label{def_mu}
\end{equation}
We do likewise with the integration over $w$, i.e., we define
\begin{equation}
M:=\int_{\substack{\psi^\star(w)=0,\\[.3mm] w\geq 0\hfill}} M(w)\mathrm{d}w\, w.
\label{def_M}
\end{equation}
By \eqref{def_mu}, \eqref{def_M}, it follows that, for each $i$, $\mu_i$ is a positive linear functional and that $M$ is positive semidefinite. Moreover,
\begin{equation}
\begin{aligned}
\mu_i(g_i(X^\star))&=0,\quad i=1,\ldots,m,\\
\psi^\star(M)&= 0.
\label{comp_slack_primitive}
\end{aligned}
\end{equation}
These conditions are a non-commutative analog of \emph{complementary slackness} \cite{Nocedal2006}.

Analogously, we absorb the integrals over $\xi$ by defining the Hermitian linear functionals
\begin{equation}
\lambda_j:=\int \Theta_j(\xi)\mathrm{d}\xi\,\xi.
\end{equation}
Substituting in \eqref{lagrangian_first_atempt}, this expression is simplified to:
\begin{equation}
{\cal L}=\psi(f(\bar{X}))-\psi(M)+\alpha(1-\psi(1))-\sum_{i}\mu_i(g_i(\bar{X}))-\sum_j\lambda_j(h_j(\bar{X})).
\label{lagrangian}
\end{equation}

We next arrive at candidate first-order optimality conditions for Problem \eqref{nc_prob_hilbert} by varying the problem state variable $\psi$ and the operator variables $\bar{X}$ in Eq.~\eqref{lagrangian}, all the while imposing the stationarity of ${\cal L}$. That way, we will obtain two sets of independent constraints: the state optimality conditions \eqref{state_optimality} and strong ncKKT \eqref{strong_conditions}. 

\subsection{State optimality conditions}
Varying the state $\psi:C^*(X^\star)\to\C$ in \eqref{lagrangian} from the optimal $\psi^\star$ to $\psi^\star+\delta\psi$ leads to the condition
\begin{equation}
f(X^\star)-\alpha\id=M\geq 0.
\end{equation}
On the other hand, complementary slackness implies that $\psi^\star(M)=0$, i.e., the optimal state $\psi^\star$ is a ground state of $f(X^\star)-\alpha$ and $E_0(f(X^\star)-\alpha)=0$. From $\psi^\star(f(X^\star))=p^\star$, we have that $\alpha=p^\star$.

Call $H=f(X^\star)$. By Proposition \ref{prop:eigen}, we have
\begin{equation}
\psi^\star(H\bullet)=\psi^\star(\bullet H)=p^\star\psi^\star(\bullet).
\label{eigenvector}
\end{equation}
This implies that 
\begin{equation}
\psi^\star([H,q])=0,\quad\forall q\in C^*(X^\star).
\label{heuris_comm}
\end{equation}
In addition, assuming (w.l.o.g.\footnote{If $\psi^\star$ is not cyclic, then we can apply the GNS construction~\cite{Gelfand1943,Segal1947} and find another solution of Problem~\eqref{nc_prob_hilbert} with the same moments, such that the new state is cyclic.}) that $\psi^\star$ is cyclic, the condition $f(X^\star)-p^\star\geq 0$ is equivalent to
\begin{equation}
\psi^\star(q^* (f(X^\star)-p^\star)q)\geq 0,\quad\forall q\in C^*(X^\star).
\label{heuris_lind}
\end{equation}
Due to \eqref{eigenvector}, the second term of the left-hand side of the above equation can be written as $\frac{1}{2}\psi^\star(\{f(X^\star),q^*q\})$. In addition, by continuity it is enough to demand Eqs.~\eqref{heuris_comm} and \eqref{heuris_lind} to hold for $q\in \A(X^\star)$. Putting everything together, we arrive at Eq.~\eqref{state_optimality_hilbert}.
\subsection{Strong ncKKT}
We next derive the strong ncKKT conditions. Let $p\in C^*(X^\star)^{n}$ be a tuple of Hermitian operators, and let $\epsilon\in\R^+$. Setting $\bar{X}=X^\star+\epsilon p$ in Eq.~\eqref{lagrangian} and demanding stationarity of ${\cal L}$ at order $\epsilon$, we arrive at the condition:
\begin{equation}
\psi^\star\left(\nabla_xf(p)\Bigr|_{x=X^\star}\right)-\sum_i\mu_i\left(\nabla_x g_i(p)\Bigr|_{x=X^\star}\right)-\sum_j\lambda_j\left(\nabla_x h_j(p)\Bigr|_{x=X^\star}\right)=0,\quad\forall p\in C^*(X^\star)^{n}.
\label{optim_middle}
\end{equation}
Again by continuity, it is sufficient to demand the equation above to hold for all $p\in \A(X^\star)^n$.

Together with complementary slackness \eqref{comp_slack_primitive}, we arrive at the definition of strong ncKKT in Section \ref{sec:first_order}.

\section{Re-derivation of the first-order optimality conditions of \cite{helton2023synchronous} through Theorem \ref{essential_theo}}
\label{app:rederivation_helton}
Consider the following problem:
\begin{align}
\inf_{\H,\tau,E} & \tau\left(\sum_{x,y,a,b}W(a,b,x,y)E_{a|x}E_{b|y}\right)\nonumber\\
\mbox{s.t. }&\tau\mbox{ tracial state},\nonumber\\
&E_{a|x}-E_{a|x}^2=0,\quad\forall a,x\nonumber\\
&\sum_aE_{a|x}-1=0,\quad\forall x,
\label{tracial_prob}
\end{align}
where $x=1,\ldots ,n$, $a=1,\ldots ,d$. Here, by tracial state, we mean a state $\tau$ with the property $\tau(ab)=\tau(ba)$, for all $a,b\in B(\H)$.
This problem is considered in \cite{helton2023synchronous} in the context of quantum non-locality. As shown in \cite{Klep_2021}, optimizations over tracial states can also be relaxed through hierarchies of SDPs. Since the associated quadratic module is Archimedean, such hierarchies converge and optimal solutions for (\ref{tracial_prob}) exist. Moreover, the optimal state can be assumed faithful, i.e., $\tau(aa^*)=0$ iff $a=0$, for  $a\in B(\H)$.

Let $(\H^\star,\tau^\star,E^\star)$ be one of those solutions. By Theorem \ref{essential_theo}, it follows that, for any $k\in \N$, there exist linear functionals $\{\lambda^k_{a|x}:\C\langle e\rangle \to\C\}$, $\{\lambda_x^k:\C\langle e\rangle \to\C\}$ such that
\begin{equation}\label{tau_optimality}
\begin{split}
&\tau^\star\left(\sum_{x,y,a,b}W(a,b,x,y)(p_{a|x}(E^\star)E^\star_{b|y}+E^\star_{a|x}p_{b|y}(E^\star))\right)\\
&=\sum_{a,x}\lambda^k_{a|x}\left(p_{a|x}-\{p_{a|x},E_{a|x}\}\right)+\sum_x\lambda^k_x\left(\sum_{a}p_{a|x}\right),
\end{split}
\end{equation}
for any tuple of polynomials $p=(p_{a|x})_{a|x}$ of degree $2(k-1)$. 

Now, for some $z\in\{1,\ldots ,n\}$, $s\in \C\langle e\rangle$, choose $p_{a|z}=i[s,e_{a|z}]$, for all $a$ and $p_{a|x}=0$, for all $x\not=z$. Substituting in (\ref{tau_optimality}), we have that
\begin{equation}
\tau^\star\left(i\left[\sum_{y,a,b}W(a,b,z,y)[s(E^\star),E^\star_{a|z}]E^\star_{b|y}+\sum_{x,a,b}W(a,b,x,z)E^\star_{a|x}[s(E^\star),E^\star_{b|z}])\right]\right)=0.
\end{equation}
Invoking the tracial properties of the state, this can be rewritten as $\tau^\star\left(s(E^\star)q_z(E^\star)\right)=0$, with
\begin{equation}
q_z(E^\star):=\sum_{y,a,b}W(a,b,z,y)[E^\star_{a|z},E^\star_{b|y}]+\sum_{x,a,b}W(a,b,x,z)[E^\star_{b|z},E^\star_{a|x}].    
\end{equation}
For $k\geq 2$, we can choose $s=q_z^*$, and thus $\tau^\star\left(q_z(E^\star)^*q_z(E^\star)\right)=0$. Since by assumption the state $\tau^\star$ is faithful, we have that
\begin{equation}
q_z(E^\star)=0,
\end{equation}
for $z=1,\ldots ,n$. This is the first-order optimality optimality condition derived in \cite[Prop. 7.1]{helton2023synchronous}.

\section{NPO for quantum nonlocality}
\label{app:rollo}
What follows is the original NPO formulation to compute the maximum quantum value of a Bell functional. We have added the partial operator KKT conditions derived in Section \ref{sec:bell}, together with the state optimality conditions \eqref{state_optimality}.
\small
\begin{equation*}
\begin{aligned}
c^\star=&\min\sigma\left(\frac{1}{2}\sum_{a,b,x,y}C(a,b,x,y)\{E_{a|x},F_{b|y}\}\right)\\ 
\mbox{s.t. }&\sigma(ss^*)\geq0,\quad\forall s\in\P,\\ 
&\sigma\Big(s\big((E_{a|x})^2-E_{a|x}\big)s'\Big)=0,\quad\forall a,x,\forall s,s'\in\P\\ 
&\sigma\Big(s\Big(\sum_aE_{a|x}-1\Big)s'\Big)=0,\quad\forall x,\forall s,s'\in\P\\ 
&\sigma\Big(s\big((F_{b|y})^2-F_{b|y}\big)s'\Big)=0,\quad\forall b,y,\forall s,s'\in\P\\ 
&\sigma\Big(s\Big(\sum_bF_{b|y}-1\Big)s'\Big)=0,\quad\forall y,\forall s,s'\in\P\\ 
&\sigma\big(s[E_{a|x},F_{b|y}]s'\big)=0,\quad\forall a,b,x,y,\forall s,s'\in\P\\ 
&\mu^A_{a|x}(ss^*)\geq 0,\quad\forall a,x, \forall s\in\E,\\
&\mu^A_{a|x}\Big(s\big((E_{a'|x'})^2-E_{a'|x'}\big)s'\Big)=0,\quad\forall a,a',x,x',\forall s,s'\in\E\\
&\mu^A_{a|x}\Big(s\Big(\sum_{a'}E_{a'|x'}-1\Big)s'\Big)=0,\quad\forall a,x,x',\forall s,s'\in\E,\\
&\lambda^A_{x}\Big(s\big((E_{a|x'})^2-E_{a|x'}\big)s'\Big)=0,\quad\forall a,x,x',\forall s,s'\in\E\\
&\lambda^A_{x}\Big(s\Big(\sum_{a}E_{a|x}-1\Big)s'\Big)=0,\quad\forall x,\forall s,s'\in\E,\\
&\mu^A_{a|x}(E_{a|x})=0,\quad\forall a,x,\\
&\sigma\left(\frac{1}{2}\sum_{b,y}C(a,b,x,y)\{p,F_{b|y}\}\right)=\mu^A_{a|x}\left(p\right)+\lambda^A_x(p),\quad\forall a,x,\forall p\in \E. \\
&\mu^B_{b|y}(ss^*)\geq 0,\quad \forall b,y,\forall s\in\F\\ 
&\mu^B_{b|y}\Big(s\big(F_{b'|y'})^2-F_{b'|y'}\big)s'\Big)=0,\quad\forall b,b',y,y',\forall s,s'\in\F\\ 
&\mu^B_{b|y}\Big(s\Big(\sum_{b'}E_{b'|y'}-1\Big)s'\big)=0,\quad\forall b,y,y',\forall s,s'\in\F\\ 
&\lambda^B_{y}\Big(s\big(F_{b|y'})^2-F_{b|y'}\big)s'\Big)=0,\quad\forall b,y,y',\forall s,s'\in\F\\ 
&\lambda^B_{y}\Big(s\Big(\sum_{b}F_{b|y}-1\Big)s'\Big)=0,\quad\forall y,\forall s,s'\in\F\\ 
&\mu^B_{b|y}(F_{b|y})=0,\quad\forall b,y,\\ 
&\sigma\left(\frac{1}{2}\sum_{a,x}C(a,b,x,y)\{E_{a|x},p\}\right)=\mu^B_{b|y}\left(p\right)+\lambda^B_y(p),\quad\forall b,y,\forall p\in \F,\\ 
&\sigma\left(\left[\frac{1}{2}\sum_{a,b,x,y}C(a,b,x,y)\{E_{a|x},F_{b|y}\},p\right]\right)=0, \\ 
&\sigma\left(p^* \left(\frac{1}{2}\sum_{a,b,x,y}C(a,b,x,y)\{E_{a|x},F_{b|y}\}\right) p-\frac{1}{2}\left\{\frac{1}{2}\sum_{a,b,x,y}C(a,b,x,y)\{E_{a|x},F_{b|y}\},p^* p\right\}\right)\geq 0,\quad\forall p\in \P.
\end{aligned}
\end{equation*}

\end{appendix}

\end{document}